\documentclass[format=acmsmall, review=false]{acmart}
\usepackage{acm-ec-24-arxiv}
\usepackage{amsfonts}
\usepackage{fancyhdr}
\usepackage{comment}
\usepackage{array}
\usepackage{mdframed}
\usepackage{amsmath}
\usepackage{amsthm}
\usepackage{thmtools}
\usepackage{thm-restate}
\usepackage{cleveref}
\usepackage{makecell}
\usepackage{graphicx}%
\usepackage{bbm}
\usepackage{pgfplots}
\usepackage{xcolor}
\usepackage{caption}
\usepackage{subcaption}

\usepackage{booktabs} 
\usepackage[ruled]{algorithm2e} 

\setcitestyle{parens}

\pgfplotsset{compat=1.17}

\newtheorem{theorem}{Theorem}
\newcommand{\poly}{\textsc{Poly}}

\newtheorem{corollary}{Corollary}

\newtheorem{definition}{Definition}
\newtheorem{example}{Example}

\newtheorem{lemma}{Lemma}

\newtheorem{remark}{Remark}

\newcommand{\Q}{\mathbb{Q}}
\newcommand{\N}{\mathbb{N}}
\newcommand{\R}{\mathbb{R}}

\newcommand{\E}{\mathbb{E}}
\usepackage{graphicx}
\graphicspath{ {./images/} }

\newcommand{\advAccount}{A}
\newcommand{\knownHonest}{B}
\newcommand{\opaqueHonest}{C}
\newcommand{\VRF}{\textsc{VRF}}
\newcommand{\CDF}{\textit{CDF}}

\newcommand{\seed}[1]{Q_{#1}}
\DeclareMathOperator{\credential}{\textsc{Cred}}
\newcommand{\cred}[2]{\credential^{#1}_{#2}}
\DeclareMathOperator{\expdist}{exp}
\newcommand{\expd}[1]{\expdist(#1)}

\DeclareMathOperator{\Reward}{Rew}
\newcommand{\Rew}[1]{\Reward(#1)}
\newcommand{\RewT}[1]{\Reward_{\round}(#1)}
\newcommand{\optTk}{\optrew{\round, \coin}}
\newcommand{\optTI}{\optrew{\round, \infty}}

\DeclareMathOperator{\Lin}{Lin}
\newcommand{\LinRew}[1]{\Reward^{\Lin}(#1)}
\newcommand{\LinLamRew}[2]{\Reward^{\Lin}_{#1}(#2)}
\newcommand{\strategy}{\pi}

\newcommand{\advPot}{\hat{W}}
\newcommand{\potWinners}{W}

\newcommand{\fst}{\tau}

\DeclareMathOperator{\Omni}{OMNI}
\newcommand{\RewOmni}[1]{\Reward^{\Omni}(#1)}
\newcommand{\omniconst}{\kappa}

\newcommand{\restrictedOmniConst}{\phi}

\newcommand{\omniChoiceTree}{\Gamma}

\newcommand{\event}[1]{E_{#1}}
\newcommand{\nevent}[1]{\overline{E}_{#1}}
\newcommand{\prevent}[1]{e_{#1}}
\newcommand{\prnevent}[1]{\overline{e}_{#1}}
\newcommand{\IndEvent}{E}
\newcommand{\prindevent}{e}
 \newcommand{\h}{h}
 \newcommand{\omniOPT}{\strategy^{\Omni}}

\newcommand{\AdvC}{\vec{c}_{-0}}
\newcommand{\AdvR}{\vec{r}_{-0}}
\newcommand{\stakelist}{\alpha, \beta, \lambda}

\newcommand{\round}{T}
\newcommand{\coin}{k}

\DeclareMathOperator{\D}{D}
\newcommand{\distribution}[2]{\D^{#1}_{#2}}
\newcommand{\dist}[1]{\D^{\optimal}_{#1}}
\newcommand{\dis}[1]{\D^{#1}}
\newcommand{\optdis}{\D^{\optimal}}
\newcommand{\estdist}[1]{\hat{\D}^{\optimal}_{#1}}

\newcommand{\UBD}[1]{\overline{\D}_{#1}}
\newcommand{\LBD}[1]{\underline{\D}_{#1}}
\newcommand{\estUBD}{\UBD{\round, \coin}}
\newcommand{\estLBD}{\LBD{\round, \coin}}
\newcommand{\estk}[1]{\estdist{#1, \coin}}

\newcommand{\apxd}[1]{\hat{\D}_{#1}}

\newcommand{\q}{q}

\newcommand{\samples}{n}
\newcommand{\vecc}{\vec{c}}
\newcommand{\vecr}{\vec{r}}

\DeclareMathOperator{\CSSPAOP}{CSSPA}
\newcommand{\CSSPA}{\CSSPAOP(\alpha, \beta)}
\newcommand{\ParCSSPA}[1]{\CSSPAOP(\alpha, #1, \beta)}
\newcommand{\FinCSSPA}[1]{\CSSPAOP(\alpha, \beta, #1, \coin)}
\newcommand{\InfCoinCSSPA}[1]{\CSSPAOP(\alpha, \beta, #1, \infty)}

\DeclareMathOperator{\LinCSSPAOP}{LinearCSSPA}
\newcommand{\LinearCSSPA}{\LinCSSPAOP(\alpha, \beta, \lambda)}
\newcommand{\FinLinCSSPA}[1]{\LinCSSPAOP(\alpha, \beta, \lambda, #1)}
\newcommand{\FinLinLamCSSPA}[1]{\LinCSSPAOP(\alpha, \beta, #1, \round, \coin)}

\newcommand{\addl}[1]{\text{AddLayer}(\alpha, \beta, \lambda, #1)}
\newcommand{\kaddl}[1]{\text{AddLayer}(\alpha, \beta, \lambda, \coin, \optTk, #1)}

\newcommand{\finsampleaddl}[1]{\text{FiniteSampleAddLayer}(\alpha, \beta, \lambda, \coin, t, \samples, \optTk, \chernoff, \strat, #1)}

\newcommand{\drawadv}[1]{\text{DrawAdv}(\alpha, #1)}

\newcommand{\sample}[1]{\text{sample}(\alpha, \beta, \lambda, \AdvC, \AdvR, #1)}
\newcommand{\sampleFromPrecomp}[1]{\text{SampleFromPrecompute}(\alpha, \beta, \lambda, \AdvC, \AdvR, \optTk, G, #1)}

\newcommand{\simulate}[1]{\text{Simulate}(\stakelist, #1)}
\newcommand{\tsimulate}{\text{TruncatedSimulate}(\stakelist, \round, \coin, \optTk)}
\newcommand{\tsimulateextraparameters}[1]{\text{TruncatedSimulate}(\stakelist, \round, \coin, \optTk, \chernoff, \strat)}
\newcommand{\tcsimulate}{\text{TruncatedSimulate}(\stakelist, \round, \coin, \optTk, \chernoff, \strat, \epsilon, \eta)}
\newcommand{\tcstrategysimulate}[1]{\text{TruncatedSimulate}(\stakelist, \round, \coin, #1, \chernoff, \strat, \epsilon, \eta)}
\newcommand{\tclamstrategysimulate}[1]{\text{TruncatedSimulate}(\alpha, \beta, #1, \round, \coin, \optTk(#1), \chernoff, \strat, \epsilon, \eta)}

\DeclareMathOperator{\optimal}{OPT}
\newcommand{\optrew}[1]{\strategy^{\optimal}_{#1}}
\newcommand{\opt}{\strategy^{\optimal}}

\newcommand{\Infl}[1]{\text{Inflate}(\samples, #1, t)}
\newcommand{\Defl}[1]{\text{Deflate}(\samples, #1, t)}
\newcommand{\chernoff}{\gamma}
\newcommand{\strat}{\omega}

\newcommand{\Precomp}[1]{\text{Precompute}(#1, \epsilon, \eta, t)}

\newcommand{\challengeConverge}[1]{\textcolor{black}{#1}}
\newcommand{\challengeTk}[1]{\textcolor{black}{#1}}
\newcommand{\challengeSample}[1]{\textcolor{black}{#1}}
\newcommand{\challengeComp}[1]{\textcolor{black}{#1}}

\SetAlFnt{\small}
\SetAlCapFnt{\small}
\SetAlCapNameFnt{\small}
\SetAlCapHSkip{0pt}
\IncMargin{-\parindent}

\setcitestyle{authoryear}

\ccsdesc[500]{Theory of computation~Algorithmic game theory and mechanism design}
\ccsdesc[300]{Security and privacy~Cryptanalysis and other attacks}

\keywords{blockchain, cryptocurrency, proof-of-stake, cryptography, strategic mining, provably correct estimations, fixed point estimations}

\copyrightyear{2024}
\acmYear{2024}
\setcopyright{rightsretained}
\acmConference[EC '24]{The 25th ACM Conference on Economics and Computation}{July 8--11, 2024}{New Haven, CT, USA}
\acmBooktitle{The 25th ACM Conference on Economics and Computation (EC '24), July 8--11, 2024, New Haven, CT, USA}
\acmDOI{10.1145/3670865.3673602}
\acmISBN{979-8-4007-0704-9/24/07}

\title[Computing Optimal Manipulations in Cryptographic Self-Selection Proof-of-Stake Protocols]{Computing Optimal Manipulations in Cryptographic Self-Selection Proof-of-Stake Protocols}

\author{Matheus V. X. Ferreira}
\email{matheus@seas.harvard.edu}
\orcid{0000-0002-7695-026X}
\affiliation{%
  \institution{University of Virginia}
  \city{Charlottesville}
  \state{VA}
  \country{USA}
}

\author{Aadityan Ganesh}
\email{aadityanganesh@princeton.edu}
\orcid{0009-0000-3567-8178}
\affiliation{%
  \institution{Princeton University}
  \city{Princeton}
  \state{NJ}
  \country{USA}
}

\author{Jack Hourigan}
\email{hojack@seas.upenn.edu}
\orcid{0000-0001-7744-795X}
\affiliation{%
  \institution{University of Pennsylvania}
  \city{Philadelphia}
  \state{PA}
  \country{USA}
}

\author{Hannah Huh}
\email{hannahchuh@gmail.com}
\orcid{0009-0008-6237-3685}
\affiliation{%
  \institution{Princeton University}
  \city{Princeton}
  \state{NJ}
  \country{USA}
}

\author{S. Matthew Weinberg}
\email{smweinberg@princeton.edu}
\orcid{0000-0001-7744-795X}
\affiliation{%
  \institution{Princeton University}
  \city{Princeton}
  \state{NJ}
  \country{USA}
}

\author{Catherine Yu}
\email{catyu6000@gmail.com}
\orcid{0009-0007-4605-2964}
\affiliation{%
  \institution{Princeton University}
  \city{Princeton}
  \state{NJ}
  \country{USA}
}

\begin{abstract}
    Cryptographic Self-Selection is a paradigm employed by modern Proof-of-Stake consensus protocols to select a block-proposing ``leader.'' Algorand~\citep{CM19} proposes a canonical protocol, and~\citet{FHWY22} establish bounds $f(\alpha,\beta)$ on the maximum fraction of rounds a strategic player can lead as a function of their stake $\alpha$ and a network connectivity parameter $\beta$. While both their lower and upper bounds are non-trivial, there is a substantial gap between them (for example, they establish $f(10\%,1) \in [10.08\%, 21.12\%]$), leaving open the question of how significant of a concern these manipulations are. We develop computational methods to provably nail $f(\alpha,\beta)$ for any desired $(\alpha,\beta)$ up to arbitrary precision, and implement our method on a wide range of parameters (for example, we confirm $f(10\%,1) \in [10.08\%, 10.15\%]$).
    
\phantom{.} \quad  Methodologically, estimating $f(\alpha,\beta)$ can be phrased as estimating to high precision the value of a Markov Decision Process whose states are countably-long lists of real numbers. Our methodological contributions involve (a) reformulating the question instead as computing to high precision the expected value of a distribution that is a fixed-point of a non-linear sampling operator, and (b) provably bounding the error induced by various truncations and sampling estimations of this distribution (which appears intractable to solve in closed form). One technical challenge, for example, is that natural sampling-based estimates of the mean of our target distribution are \emph{not} unbiased estimators, and therefore our methods necessarily go beyond claiming sufficiently-many samples to be close to the mean.
\end{abstract}

\begin{document}


\maketitle

\section{Introduction}

Blockchain protocols have attracted significant interest since Bitcoin's initial development in 2008~\citep{Nakamoto08}, and several parallel research agendas and developments arose in that time. This paper lies at the intersection of two of these agendas: (a) strategic manipulability of consensus protocols, and (b) the return of Byzantine Fault Tolerant (BFT)-style consensus protocols via Proof-of-Stake (PoS). We briefly elaborate on both stories below, before discussing our contributions.\\

\noindent\textbf{Manipulating Consensus Protocols.} Initially following Nakamoto's whitepaper, Bitcoin and related blockchain protocols were studied through a classical security lens: some fraction of participants were honest, others were malicious, and the goal of study was to determine the extent to which a malicious actor can compromise security with a particular fraction of the computational power in the network. For example, Nakamoto's whitepaper already derives that with $51\%$ of the computational power, a malicious actor could completely undermine Bitcoin's consensus protocol. However, the seminal work of \citet{ES14}, now referred to as ``Selfish Mining'', identified a fundamentally different cause for concern: an attacker with $34\%$ of the computational power could manipulate the protocol in a way that does not violate consensus, but earns that attacker a $>34\%$ fraction of the mining rewards.\footnote{Earlier work of~\cite{BabaioffDOZ12} introduces the strategic manipulation aspects of the Bitcoin protocol, although the style of manipulation in~\cite{ES14} became more mainstream for subsequent work.} This agenda has exploded over the past decade, and there is now a vast body of work considering strategic manipulation of consensus protocols (e.g.~\citealp{ES14, KiayiasKKT16, SSZ16, CarlstenKWN16,TsabaryE18,GorenS19,FiatKKP19, BrownCohenNPW19, ZurET20, FW21, FHWY22, YaishTZ22, YaishSZ23, BahraniW23}). 

These works study \emph{several} different classes of protocols, and \emph{several} avenues for manipulation: some study Proof-of-Work protocols while others study Proof-of-Stake, some study block withholding deviations while others manipulate timestamps, some focus on profitability denoted in the underlying cryptocurrency while others consider the impact of manipulation on that cryptocurrency's value. There are many important angles to this agenda, many of which are cited by practitioners as key motivating factors in design choices.\footnote{For example, \href{https://eips.ethereum.org/EIPS/eip-1559}{EIP-1559}.} The primary goal of this agenda is to understand \emph{under what conditions is it in every participant's interest to follow the prescribed consensus protocol?} That is, these works generally do not focus on understanding complex equilibria with multiple strategic players (and instead immediately consider it a failure of the protocol when it is not being followed), and instead seek to understand whether the strategy profile where all agents follow the protocol constitutes a Nash Equilibrium.
That is, we seek to understand whether being honest is the best response when everyone else in the network is honest.
\\

\noindent\textbf{BFT-based Proof-of-Stake Protocols.} As Bitcoin's popularity surged, the energy demands required to secure it comparably soared, and estimates place its global energy consumption at comparable levels to countries the size of Australia. This motivated discussions over alternate technologies that could still be permissionless and Sybil-resistant, and Proof-of-Stake emerged as a viable alternative. While Proof-of-Work protocols select participants to produce blocks proportional to their computational power, Proof-of-Stake protocols do so proportional to the fraction of underlying cryptocurrency they own. Initial Proof-of-Stake protocols predominantly followed the longest-chain consensus paradigm of Bitcoin~\citep{DaianPS17, KiayiasRDO17}, but modern proposals now look more like classical consensus algorithms from distributed computing~\citep{GiladHMVZ17, CM19}. Specifically, BFT-based protocols run a consensus algorithm, one block at a time, in order to reach consensus on a single block. Once consensus is reached, the block is finalized and consideration of the next block begins.

While Bitcoin still uses a longest-chain Proof-of-Work consensus protocol, and some large Proof-of-Stake cryptocurrencies such as Cardano still use longest-chain protocols~\citep{KiayiasRDO17}, BFT-based protocols are now quite mainstream and are implemented, for example, in Algorand~\citep{GiladHMVZ17,CM19} and Ethereum.\footnote{Ethereum does maintain some longest-chain aspect to its protocol, but the key role that validators play make the protocol closer to BFT-based consensus.}\\

\noindent\textbf{Manipulating BFT-based Proof-of-Stake Protocols.} In practice, there is no `dominant' BFT-based protocol, and different cryptocurrencies each seem to have their own protocol. However, there are some unifying themes. Most BFT-based protocols have the concept of a \emph{leader} in each round, and the consensus goal of each round is for everyone to agree on the leader's proposed block. Leader selection is challenging, though: it should be done proportional to stake, but in a way that neither relies on a trusted external source of randomness nor is manipulable by participants. This has proved to be quite challenging, and to-date there are no nonmanipulable proposals without heavyweight cryptography (such as Multi-Party Computation or Verifiable Delay Functions). While the underlying consensus protocols are often both complex and completely nonmanipulable (without sufficient stake to simply subvert consensus in the first place), the leader selection protocols are more vulnerable, and can be studied independently of the supported consensus protocol.

Algorand's initial proposal serves as a canonical process of study due to its elegance. The initial seed $Q_1$ is a uniformly drawn random number. Then, in each round $t$ with seed $Q_t$, every wallet digitally signs the statement $(Q_t, c)$ for every coin $c$ they own, hashes it,\footnote{We will elaborate on this rigorously in Section~\ref{sec:CSS}. The role of the digital signature is simply to get a signature unique to the owner of coin $c$ that no other player can predict, and the role of the hash function is to turn this into a uniformly drawn random number from $[0,1]$.} and broadcasts the hash as their credential $\cred{c}{t}$. The holder of the coin with the lowest credential is the leader, and their winning credential becomes the seed $Q_{t+1}$ for round $t+1$. This elegant protocol has several desirable properties (for example, it is not vulnerable to any form of `stake grinding' to influence next round's seed -- you can either broadcast your credential or not),\footnote{The live Algorand protocol seems to have recently pivoted from their initial proposal to a leader-selection protocol that has the winning credential of every $k^{th}$ round set the seeds for the next $k$ rounds.} but~\citet{CM19} acknowledge that it may still be manipulable by cleverly choosing not to broadcast credentials, and~\citet{FHWY22} indeed establish that any size staker has such a profitable manipulation.

To get brief intuition for a profitable manipulation, imagine that an attacker controls $10\%$ of the coins. Perhaps they are also well-enough connected in the network so that they can choose which credentials of their own to broadcast in round $t$ as a function of other participants' credentials (this corresponds to $\beta = 1$ -- in general, the adversary is $\beta$-well-connected if they see a $\beta$ fraction of honest credentials before broadcasting their own). Such an attacker might be in a position where they own (say) the three lowest credentials. In this case, the adversary could look one round ahead and determine which of these round-$t$-winning credentials gives them the best chance of winning round $t+1$, and broadcast only that credential. This particularly simple manipulation, termed the One-Lookahead strategy in~\citet{FHWY22}, is strictly profitable for any sized staker. While strictly better than honest, this strategy does not reap enormous profits: even with $10\%$ stake, a $1$-well-connected staker can lead at most $10.08\%$ rounds. On the other hand, the only previous upper bounds derived on the maximum gains come from a loose analysis of an omniscient adversary who not only sees the credentials of honest wallets in round $t$, but can predict their future digital signatures to know exactly which hypothetical future rounds they'd win. This results in an upper bound of $21.12\%$ on the maximum possible rounds led by a $10\%$ staker. Needless to say, the level of concern that would arise from a $10\%$ staker who is able to slightly increase their staking rewards by less than $1\%$ to $10.08\%$ is vastly different than what would arise from a $10\%$ staker who can more than double their staking rewards to $21.12\%$.\\

\noindent\textbf{Our Contributions.} Using both theoretical and computational tools, we precisely nail down the manipulability of Algorand's canonical leader selection protocol. That is, we design computational methods to compute, for any fraction of stake $\alpha$ and network connectivity parameter $\beta$ that an attacker might have,\footnote{We define the network connectivity formally in~\Cref{sec:Game}.} the maximum fraction of rounds the attacker can lead (assuming other players are honest). We rigorously bound the error in our methods -- some bounds hold with probability one (due to discretization, truncation, etc.) while others hold with high probability (due to sampling). We consider this methodology to be our main contribution. We also make the following adjacent contributions:
\begin{itemize}
    \item We implement our computational procedure in Rust, and run it across several personal laptops and university clusters. We plot several of our findings in ~\Cref{sec:SimPlots} (and future researchers can run our code to even higher precision, if desired). For example, we close the gap on the maximum profit of a $1$-well-connected $10\%$ staker from $[10.08\%, 21.12\%]$ to $[10.08\%, 10.15\%]$. We produce several plots in Section~\ref{sec:SimPlots} demonstrating our results in comparison to prior bounds.
    \item One conclusion drawn from our simulations is that the gains from manipulation are quite small. For example, we confirm that $1$-well-connected $10\%$ staker can lead at most $10.15\%$ of all rounds. A $0$-well-connected $10\%$ staker can lead at most $10.09\%$ of the rounds. Even a $0$-well-connected $20\%$ staker can lead at most $20.21\%$ of the rounds. This suggests that, while supralinear rewards are always a cause for concern as a potential centralizing force among stakers, the situation is unlikely to be catastrophic. 
    \item A second conclusion drawn from our simulations is $\beta$ plays a significant role in the magnitude of profitability. For example, a $0$-well-connected $20\%$ staker can lead at most $20.21\%$ of the rounds, while there exists a strategy for a $1$-well-connected $20\%$ staker that leads at least $20.68\%$ of the rounds -- a $320\%$ amplification in the marginal gains. 
    \item Beyond our provably accurate computational methodology, we also provide two analytical results of independent interest.
    \begin{itemize}
        \item We improve~\cite{FHWY22}'s analysis of the omniscient adversary, and in particular describe a recursive formulation that {achieves an arbitrarily good approximation to the precise profit of an optimal omniscient adversary}. This appears in~\Cref{sec:OmniMain}.
        \item Finally, we prove one conjecture and disprove another of~\citet{FHWY22} characterizing ``Balanced Scoring Functions.'' Balanced Scoring Functions are a tool used in leader selection to replace computing a digital-signature-then-hash per coin with computing a digital-signature-then-hash per wallet (in order to appropriately weight the hash before taking the minimum amongst all credentials). We state this result in Section~\ref{sec:Primitives}.
    \end{itemize}
\end{itemize}

As a whole, our results significantly improve our understanding of manipulating the canonical leader selection protocol first introduced in Algorand~\citep{CM19}. First, while supralinear rewards are always a cause for concern, the maximum achievable profits are at quite a small order of magnitude. Second, our results highlight the pivotal role that $\beta$ plays in the rate of supralinear rewards. This suggests that protocol designers may wish to invest in augmentations to bring $\beta$ closer to zero.\footnote{One such possibility is to broadcast credentials using commit-reveal: players make a large deposit along with a cryptographic commitment to their credential, and unlock their deposit only upon revealing it.} Methodologically, our approach provides a blueprint for how similar leader selection protocols (such as Ethereum's) might be analyzed.

\subsection{Very Brief Technical Highlight}
We defer full details to our technical sections, but give a brief overview of the key technical challenges here. The optimal strategy can be phrased as a Markov Decision Process (MDP), and in some sense the obvious approach is to ``write down the MDP and solve it.'' Unfortunately, states in our MDP are countably long lists of real numbers. That is, a state corresponds to (a) the list of credentials the attacker has in this round, but also (b) for each of those credentials $i$, and each other wallet $j$ controlled by the adversary, the credential wallet $j$ would provide the next round if wallet $i$ wins this round (which can be computed as the seed for the next round is simply a digital signature plus hash of its credential $i$), and moreover (c) for each of those pairs of credentials $(i,j)$, and each other wallet $k$ controlled by the adversary, the credential wallet $k$ would provide two rounds from now if credential $i$ wins this round and credential $j$ wins the next round, and (d) so on. 

A first step is to truncate this countably long list of real numbers to (a) look only $T< \infty$ rounds in the future, (b) store only $\coin < \infty$ credentials per round, and (c) discretize each credential to a multiple of $\varepsilon > 0$. These steps can all be done with provable upper bounds on the error they induce. However, even with $\coin = 8$ and $T = 15$, states still correspond to a list of $8^{15}$ multiples of $\varepsilon$, and is clearly intractable. 

Instead, our key idea is to reformulate the question as finding the distribution of future rewards that an optimal strategist receives. That is, consider the process of sampling a state for the attacker (a list of credentials for this round, hypothetical future credentials, etc.), and ask what future reward the attacker would get when playing optimally. If we can compute this distribution of rewards $D$, then its expected value is exactly the number we seek. We define an operator $\Theta(\cdot)$ that takes as input samples from some distribution $F$ and produces samples from the distribution $\Theta(F)$, and establish that $D$ is a fixed point of $\Theta(\cdot)$. 

Again, the process now appears straight-forward: start from any distribution, and iterate $\Theta(\cdot)$ until it stabilizes. This is indeed our approach, and the remaining challenge is to account for sampling error. Essentially, we are looking for $\mathbb{E}[\Theta^{15}(F)]$ for some simple initial distribution $F$, and instead of computing $\Theta(\cdot)$ at each stage we'll take an empirical estimate $\hat{\Theta}(\cdot)$ instead. This appears ripe for a Chernoff plus union bound to bound the error due to sampling, except that $\mathbb{E}[\hat{\Theta}(\cdot)]$ is not an unbiased estimator for $\mathbb{E}[\Theta(\cdot)]$. So even establishing that we take sufficiently many samples to be close to the process's expected value does not guarantee we are close to $\mathbb{E}[\Theta^{15}(F)]$. Instead, at each round we both inflate (resp.~deflate) our empirical $\hat{\Theta}^{i}(F)$ so that we know it stochastically dominates (resp.~is stochastically dominated by) $\Theta^i(F)$ using a variant of the DKW inequality~\citep{DKW56}. 

We share this to give the reader a sense of the technical developments necessary to analyze this particular Markov Decision Process, and ways in which it differs from more common MDPs.

\subsection{Related Work}
\noindent\textbf{Manipulating Leader Selection Protocols.}~\citet{CM19} propose the Algorand leader selection protocol, and acknowledge that it may be manipulable. They also prove an upper bound on the fraction of rounds an adversary can win after being honest in the previous round.~\citet{FHWY22} provide a strategy that is strictly profitable for all $\beta$-well-connected $\alpha$-sized stakers, and upper bound the attainable profit of an omniscient adversary who can predict future digital signatures of honest players. We provide provably accurate computational methodology to nail the optimal manipulability up to arbitrary precision (and implement our algorithms and draw conclusions from the results). In concurrent and independent work,~\cite{CaiLWZ24} establish that any strictly profitable manipulation of our same leader selection protocol is statistically detectable (that is, an onlooker who sees only the seeds of each round can distinguish whether someone is profitably manipulating the protocol from when all players are honest but sometimes offline). This work is orthogonal to ours, but also provides an argument that solid defenses against manipulation exist (we argue that the manipulations are not particularly profitable, they argue that they are always detectable). \\

\noindent\textbf{Manipulating Consensus Protocols.} We have already briefly cited a subset of the substantial body of work studying profitable manipulations of consensus protocols~\citep{BahraniW23, BrownCohenNPW19, CarlstenKWN16, ES14, FHWY22, FW21, FiatKKP19, GorenS19, KiayiasRDO17, SSZ16, TsabaryE18, YaishSZ23, YaishTZ22}. Of these,~\citep{BrownCohenNPW19, FW21} also study Proof-of-Stake protocols, but longest-chain variants (and therefore have minimal technical overlap).~\citet{SSZ16} bears some technical similarity, as they are the unique prior work that finds optimal manipulations (in Bitcoin's Proof-of-Work), and they also use computational tools with theoretical guarantees. We also use a lemma of theirs to reduce from maximizing the fraction of rounds won to maximizing reward in a linear MDP. Still, there is minimal technical overlap beyond these. For example, their problem can be phrased as a Markov Decision Process with countably-many states (i.e.~a state in their setup is of the form ``how many hidden blocks do you have?'', which is an integer), and therefore the key steps in their provable guarantees are truncations. In comparison, we've noted that our problem is a Markov Decision Process with uncountably many states, and therefore the two MDPs have minimal overlap.\footnote{This is also perhaps expected, as there is little technical similarity between creating forks in a longest-chain protocol and manipulating credentials in a leader-selection protocol.}

\section{Preliminaries}

\subsection{Primitives} \label{sec:Primitives}
In this section, we review various cryptographic primitives required to construct a cryptographic self-selection protocol. Since our model is identical, we use notations identical to~\citet{FHWY22}. We begin by discussing a tool central to many Proof-of-Stake protocols-- verifiable random functions. Verifiable random functions are useful in enabling a source of randomness endogenous to the blockchain for the leader election protocol.

\begin{definition}[Ideal Verifiable Random Function (Ideal $\VRF$)] \label{def:VRF}
    An ideal verifiable random function satisfies the following properties:
    \begin{enumerate}
        \item \textbf{Setup:} There is an efficient randomized generator that can produce a pair $(sk, pk)$ of a secret key and a public key that characterizes the instance $f_{sk}(\cdot)$.
        \item \textbf{Private computability:} For a string $x$, there exists an efficient algorithm to compute the encryption $f_{sk}(x)$ of $x$ with the knowledge of $sk$.
        \item \textbf{Perfect randomness:} Without the knowledge of $sk$, the random variables $f_{sk}(x)$ and $f_{sk}(y)$ are distributed i.i.d.~over $U[0, 1]$. In particular, the random variable $f_{sk}(x) \sim U[0, 1]$ even with the knowledge of $\big((y_i, f_{sk}(y_i))\big)_{1 \leq i \leq m}$ such that $y_i \neq x$ for all $1 \leq i \leq m$.
        \item \textbf{Verifiability:} Verifying the claim $y = f_{sk}(x)$ can be done efficiently conditioned on the knowledge of $pk$ and a proof $V_x$, even if $sk$ remains unknown. Generating a proof $V_x$ such that a verifier confirms equality when $y \neq f_{sk}(x)$ is impossible. 
    \end{enumerate}
\end{definition}
An ideal $\VRF$ allows the holder of the secret key $sk$ (through property 3) to provably generate a random number, i.e, show that the random number was generated through a prescribed process. However, it is impossible to construct an ideal $\VRF$ whose outputs are \emph{statistically indistinguishable} from $U[0,1]$. On the other hand, it is possible to construct a $\VRF$ whose outputs are \emph{computationally indistinguishable} from $U[0, 1]$. For the sake of simplicity, we proceed with an Ideal $\VRF$ instead of computational -- this results in only a negligible difference.


\begin{example}[$\VRF$s through digital signatures]
Let $\sigma$ be a digital signature scheme with a public key, secret key pair $(pk, sk)$ and let $h$ be a hash function. Then, $h(\sigma_{sk}(\cdot))$ is a verifiable random function. $y = h(\sigma_{sk}(x))$ can be computed efficiently with the knowledge of $sk$. With a proof $V_x = \sigma_{sk}(x)$, $y = h(\sigma_{sk}(x))$ can be verified as follows- verify that (i) the proof $V_x = \sigma_{sk}(x)$ with the public key $pk$ and $x$ and (ii) verify $y = h(V_x)$.
\end{example}

Next, we proceed to discuss balanced scoring functions that enable electing a leader proportional to its stake.

\begin{definition}[Balanced Scoring Functions] \label{def:ScoringFunctions}
    A scoring rule $S: [0, 1] \times \R \xrightarrow{} [0, 1]$ is balanced if:
    \begin{enumerate}
        \item For $X \sim U[0, 1]$, the distribution of $S(X, \alpha)$ has no point masses for all $\alpha \in [0, 1]$
        \item For all $n \in \N$ and $\big(\alpha_i \big)_{1 \leq i \leq n} \in \R^n_{\geq 0}$,
        $$Pr_{X_1, \dots, X_n \sim U[0, 1]}\Big( \arg \min_{1 \leq i \leq n} \{S(X_i, \alpha_i)\} = j \Big) = \frac{\alpha_j}{\sum_{i = 1}^n \alpha_i}$$
    \end{enumerate}
\end{definition}

At a high level, a fair leader selection to elect a wallet with probability proportional to its stake can be conducted by choosing the wallet with the smallest score, while $\VRF$s provide the source for the random variable $X_i$ for a wallet $i$.

\citet{FHWY22} conjecture that $S(X, \sum_{i =1}^n \alpha_i)$ and $\min_{1 \leq i \leq n} \{S(X, \alpha_i)\}$ are identically distributed for all balanced scoring functions $S$, $n \in \N$ and $\alpha_1, \dots, \alpha_n \in \R_{\geq 0}$\footnote{They also claim that $S(X, \alpha)$ is continuous in $\alpha$. We provide a counterexample to their claim. However, \Cref{thm:UniqueScore} shows that $Pr(S(X, \alpha) \geq s)$ is continuous in $\alpha$. See \Cref{sec:ScoringRule} for a detailed discussion on scoring functions.}. Intuitively, their conjecture claims an adversary with a total stake $\sum_{i = 1}^n \alpha_i$ cannot increase the probability of a smaller score and thus, the probability of getting elected by splitting their stake as $\big(\alpha_i \big)_{1 \leq i \leq n}$ across $n$ different wallets. We settle their conjecture.

\begin{theorem} \label{thm:CanScoringRules}
    Let $S(X, \alpha)$ be any balanced scoring function. Then, for all $n \in \N$ and $\big(\alpha_i \big)_{1 \leq i \leq n}$, the random variables
    $$S(X, \sum_{i = 1}^n \alpha_i) \text{ and } \min_{1 \leq i \leq n} \{S(X_i, \alpha_i)\}$$
    are identically distributed for $X, X_1, \dots, X_n \sim U[0, 1]$.
\end{theorem}
The proof of \Cref{thm:CanScoringRules} and further details on scoring functions are deferred to \Cref{sec:ScoringRule}.

\subsection{Cryptographic Self-Selection} \label{sec:CSS}
We are now ready to describe a cryptographic self-selection protocol.
\begin{definition}[Cryptographic Self-Selection Protocol $A$ (CSSPA); \citealp{FHWY22}]
    A Cryptographic Self-Selection Protocol $A$ is the following:
    \begin{enumerate}
        \item Every wallet $i$ sets up an instance of an ideal $\VRF$ with public key, secret key pair $(pk_i, sk_i)$ prior to round $1$. Wallet $i$ holds a stake $\alpha_i$.
        \item $\seed{t}$ denotes the seed of round $t$. The seed $\seed{1}$ for the initial round is computed through an expensive multi-party computation and is distributed according to $U[0, 1]$.
        \item In each round $t$, the user with wallet $i$ computes its credential $\cred{i}{t} := f_{sk_i}(\seed{t})$.
        \item Each user can choose either to broadcast its credential or remain silent. Any credential broadcast by a user is received by all other users\footnote{We assume this to focus on the relevant aspects of the paper and is consistent in prior work that focuses on incentives \citep{CarlstenKWN16, ES14, FW21, FerreiraWHFC19, SSZ16, KiayiasKKT16, BahraniW23}}.
        \item The wallet with the smallest score $S(\cred{i}{t}, \alpha_i)$ amongst all broadcasted credentials is elected the leader $\ell_t$ for round $t$.
        \item The seed for round $t+1$, $\seed{t+1} = \cred{\ell_t}{t}$, the credential of the winner of round $t$. All wallets learn the seed $\seed{t+1}$.
    \end{enumerate}
\end{definition}
\noindent Importantly, note that the blockchain cannot be forked in the CSSPA as in the case with many BFT-based consensus protocols including Algorand.

We consider strategic manipulations rather than network security attacks, and so the action space of users is restricted to distributing their stakes across multiple wallets and choosing between broadcasting and remaining silent for each of its wallet, as opposed to a network partition attack. An honest player keeps its stake in a single wallet and always broadcasts its credential. Conditioned on all players in the network being honest, observe that the probability of a wallet with stake $\alpha_j$ getting elected equals the probability that wallet $i$ has the smallest score, which happens with a probability proportional to $\alpha_i$.

We discuss choosing an explicit balanced scoring function for our model. \citet{FHWY22} argue that the game induced by the CSSPA is independent of the choice of the scoring function and show a bijection between the strategies of a strategic player in the games induced by two different scoring functions that preserve the player's rewards (\Cref{def:Rew}). We choose the logarithmic scoring function defined below.

\begin{definition}[Logarithmic Scoring Function] \label{def:LogScore}
For $X \in [0,1]$ and $\alpha \in \R_{\geq 0}$,
\begin{equation}
    \notag
    S_{\ln}(X, \alpha) = 
    \begin{cases}
        \infty & \text{when } \alpha = 0 \\
        0 & \text{when } X = 0, \alpha \neq 0 \\
        \frac{- \ln X}{\alpha} & \text{otherwise}
    \end{cases}
\end{equation}    
\end{definition}

\begin{definition}[Exponential Distribution]
    The exponential distribution $\expd{\alpha}$ with rate $\alpha$ is the distribution with a cumulative density function ($\CDF$) $F(x, \alpha) = 1-e^{-\alpha \, x}$ and a probability density function (pdf) $f(x, \alpha) = \alpha e^{-\alpha \, x}$.
\end{definition}

\begin{lemma}[Lemma 2.1 from \citealp{FHWY22}]
    $S_{\ln}(X, \alpha)$ is distributed according to $\expd{\alpha}$ when $X \sim U[0, 1]$.
\end{lemma}

\noindent For notational convenience, we denote $S_{\ln}$ by $S$ unless mentioned otherwise.

Consider a user distributing their stake $\alpha$ equally across $n$ wallets for $n \xrightarrow{} \infty$. The scores of each wallet is distributed according to $\expd{\tfrac{\alpha}{n}}$. It is well-known that the minimum of $n$ i.i.d random variables drawn from $\expd{\theta}$ is distributed according to $\expd{n \, \theta}$ (see Appendix A from \citealp{FHWY22}, for example). Therefore, the minimum score over all wallets of the user is distributed as per $\expd{\alpha}$. The following describes the distribution of the $i\textsuperscript{th}$-smallest score amongst the $n$ wallets.

\begin{lemma}[Lemma 4.3 from \citealp{FHWY22}] \label{thm:SmallestScore}
    Let $\big( X_i \big)_{i \in \N}$ be exponentially distributed i.i.d random variables such that $\min_{i \in \N} \{X_i\}$ is distributed according to $\expd{\alpha}$. Let $Y_i$ be the random variable denoting the $i\textsuperscript{th}$-smallest value in $\big( X_i \big)_{i \in \N}$. Then, $\big( Y_i \big)_{i \in \N}$ is distributed according to the following random process:
    $$Y_1 \xleftarrow{} \expd{\alpha} \text{ and } Y_{i+1} \xleftarrow{} Y_i + \expd{\alpha}$$
\end{lemma}

We have the prerequisites to study the actions of a strategic player in the CSSPA in place.

\section{Model}

\subsection{The Adversarial Game and Reward} \label{sec:Game}

In a network consisting of honest stakers, we study a single strategic adversary whose rewards are proportional to the fraction of rounds it is elected to propose a block. Conditioned on the stake the adversary holds in the system, we want to estimate the optimal marginal utility gained by the adversary from being strategic. We adopt the adversarial model described in \citet{FHWY22}, which we review below.

We define the space of strategies available to the adversary. We abuse notation to denote the cryptographic self-selection protocol, the game played by the adversary and the space of strategies available to the adversary by CSSPA. 

\begin{definition}[$\CSSPA$] \label{def:RefinedCSSPA}
    In $\CSSPA$, the network consists of three players --- the adversary with stake $\alpha$, and two honest players $\knownHonest$ and $\opaqueHonest$ with stakes $\beta \, (1-\alpha)$ and $(1-\beta) (1-\alpha)$ respectively.
    Prior to round $1$, the adversary learns the values of $\alpha$, $\beta$ and that $\knownHonest$ and $\opaqueHonest$ are honest. For $n \xrightarrow{} \infty$, the adversary distributes its stake into a set $\advAccount$ of $n$ wallets, each containing a stake $\tfrac{\alpha}{n}$. The adversary makes the following decisions in round $t$:
    \begin{enumerate}
        \item The adversary learns the seed $\seed{t}$ of round $t$.

        \item The adversary computes the credentials $\cred{i}{t}$ for all wallets $i \in \advAccount$. 
            
        \item Further, the adversary learns the credentials $\cred{B}{t}$ of player $\knownHonest$. The adversary knows that the credential of player $\opaqueHonest$ is drawn from $\expd{(1-\beta) \, (1-\alpha)}$ but does not learn $\cred{C}{t}$.
        
        \item For any $r \geq 0$ and $(i_t, i_{t+1}, \dots, i_{t+r}) \in \big( \advAccount \cup \{\knownHonest\} \big) \times \advAccount^r$, the adversary precomputes the credentials $\cred{i_{t+r'}}{t+r'}$ for $1 \leq r' \leq r$ assuming $i_{t+\hat{r}}$ is elected to lead in round $t+\hat{r}$ for all $0 \leq \hat{r} < r'$.
        \item The adversary either remains silent or broadcasts the credential of a wallet $i \in \advAccount$.
    \end{enumerate}
\end{definition}
The following discussion throws light on bullet 4 of \Cref{def:RefinedCSSPA}. Before broadcasting any credential in round $t$, the adversary observes the credentials of its own wallets and the credential of $\knownHonest$. All credentials $\cred{i}{t}$ for $i \in A \cup \{B\}$ observed by the adversary are potential seeds for round $t+1$. Assuming one of these credentials as a hypothetical seed, the adversary can compute the credentials $\cred{i}{t+1}$ for all of its wallets $i \in A$. These hypothetical credentials are themselves potential seeds for round $t+2$. More generally, the adversary can precompute all possible future credentials of its wallets, assuming the precomputed credentials keep becoming the seed for successive rounds.

$\beta$ denotes the network connectivity of the adversary. The stake of $\opaqueHonest$, and therefore the probability of $\opaqueHonest$ having a small score and being selected, decreases with $\beta$. Thus, for large values of $\beta$, it is much more unlikely for a credential not precomputed by the adversary, namely $\cred{\opaqueHonest}{t}$, to become the seed $\seed{t+1}$ for the next round. Since both $\knownHonest$ and $\opaqueHonest$ have non-negative stakes, $\beta \in [0, 1]$.

Remember that the honest strategy collects all stake into a single wallet and broadcasts the credential of the wallet each round. We assume both $\knownHonest$ and $\opaqueHonest$ play the honest strategy. We normalize the total stake to $1$ and as a consequence, use the stake and the fraction of stake held in a wallet interchangeably.

\begin{definition}[Reward of a Strategy] \label{def:Rew}
    For a strategy $\strategy$ describing the actions taken by the adversary in each round of $\CSSPA$, let the Bernoulli random variable $X_t(\alpha, \beta; \strategy)$ be $1$ if the adversary is elected in round $t$ and $0$ otherwise. Then, the expected reward
    $$\Rew{\alpha, \beta; \strategy} = \E \Big[\lim \inf_{T \xrightarrow{} \infty} \frac{\sum_{t=1}^{T} X_t(\alpha, \beta; \strategy)}{T} \Big]$$
    equals the fraction of rounds led by the adversary in expectation over the outcomes of the $\VRF$s in each round.
\end{definition}
When clear from the context, we drop the parameters $\alpha$ and $\beta$ and denote $\Rew{\alpha, \beta; \strategy}$ and $X_t(\alpha, \beta; \strategy)$ by $\Rew{\strategy}$ and $X_t(\strategy)$.

The above model of the CSSPA appears quite restrictive in more than one aspect -- the adversary can broadcast at most one credential, the adversary cannot strategically distribute its stake into multiple accounts prior to round $1$, and there are only two honest players in the network. In \Cref{sec:GenCSSPA}, we recap a very general model of $\CSSPA$ discussed in \citet{FHWY22} and their results showing that the above restricted version of the CSSPA has the same optimal adversarial reward as the more general version.

\subsection{Biased Seeds and Stopping Times}
We aim to estimate the reward $\Rew{\strategy} = \E \Big[\lim \inf_{T \xrightarrow{} \infty} \frac{\sum_{t=1}^{T} X_t(\strategy)}{T} \Big]$ the adversary wins by playing a strategy $\strategy$. A tractable closed-form expression for $X_t(\strategy)$ is hard to find and computing its expected values for all $1 \leq t \leq \infty$ is infeasible. Therefore, it becomes imperative to find a round $\fst$ such that the expected adversarial reward can be estimated without computing the expected value of $X_{\fst+r}(\strategy)$ for any $r > 0$. We call such a round $\fst$ a stopping time. The expected adversarial reward has a much simpler expression in terms of stopping times.

\begin{lemma}[Lemma 4.1 from \citealp{FHWY22}] \label{thm:ResetST}
    Suppose the strategy $\strategy$ has an expected finite stopping time $\fst$ in $\CSSPA$. Then,
    $$\Rew{\strategy} = \frac{\E[\sum_{t = 1}^{\fst} X_t(\strategy)]}{\E[\fst]}$$
    where $\fst$ is a random variable denoting a stopping time.
\end{lemma}

Suppose we reach a round $\fst + 1$ such that the adversary is indifferent between the current seed $\seed{\fst + 1}$ and a fresh draw from $U[0, 1]$. We say such a seed $\seed{\fst + 1}$ is \emph{unbiased}. The adversary's rewards from the rounds following $\fst + 1$ is similar to restarting $\CSSPA$ from round $1$, whose initial seed $\seed{1}$ is drawn from $U[0, 1]$ (in practise, this is done through an expensive multi-party computation and is not susceptible to manipulation). The expected adversarial reward can be computed by estimating only the distributions of $X_1(\strategy), X_2(\strategy), \dots, X_{\fst}(\strategy)$ and thus, $\fst$ is a stopping time.

For an arbitrary round $t$, it is hard to determine whether the seed $\seed{t+1}$ is unbiased and whether $t$ is a stopping time. We define forced stopping times so that they are much easier to identify. In the next few paragraphs, we motivate forced stopping times through an example adversarial strategy.

For a stake $\alpha$ and a random seed $\seed{t}$ in round $t$, the probability that the adversary gets elected is at most $\alpha$. However, the adversary could have multiple wallets whose scores are smaller than the score of the honest wallets $\knownHonest$ and $\opaqueHonest$. When $\beta = 1$, the adversary knows the smallest honest score and that broadcasting the credentials of any of its wallets with a smaller score would ensure an adversarial wallet getting elected in round $t$. It is convenient to explicitly christen these candidate adversarial wallets.

\begin{definition}[Potential Winners and Adversarial Potential Winners] \label{def:PotentialWinners}
    In round $t$ with seed $\seed{t}$, let $\advPot(\seed{t})$ be the set of all adversarial wallets with a score less than that of $\knownHonest$. Then, $\advPot(\seed{t})$ is the set of adversarial potential winners and $\potWinners(\seed{t}) = \advPot(\seed{t}) \cup \{B\}$ is the set of potential winners in round $t$. 
\end{definition}

Out of these adversarial potential winners, the adversary can choose to broadcast the one that optimizes its future rewards, which can be estimated by computing hypothetical seeds for the rounds following round $t$ (see bullet 4 from \Cref{def:RefinedCSSPA}). The adversary can also choose to sacrifice the current round and remain silent if the future rewards from the honest wallet $\knownHonest$ getting elected more than compensates for losing round $t$.

Suppose a seed $\seed{\fst + 1}$ is realized for which the adversary has not computed any hypothetical future seeds. At the instant in which $\seed{\fst+1}$ is realized, the adversary is indifferent between $\seed{\fst+1}$ and a fresh draw from $U[0, 1]$ and thus, $\seed{\fst+1}$ is unbiased. Now, consider a round $\fst$ in which the smallest score either belongs to $\knownHonest$ or $\opaqueHonest$. The first time the adversary learns the honest credential with the smallest score, the adversary would have pre-computed neither the credential nor any hypothetical future credentials following the honest credential since computing them would require the secret key of the honest wallet with the smallest score. The adversary cannot thwart $\cred{\knownHonest}{\fst}$ or $\cred{\opaqueHonest}{\fst}$ from becoming the seed $\seed{\fst+1}$. Thus, the seed $\seed{\fst+1}$ is unbiased and round $\fst$ is a stopping time. We call such stopping times as forced stopping times. Forced stopping times are easy to identify since we only have to ensure that the smallest score does not belong to an adversarial wallet.

\begin{definition}[Forced Stopping Time] \label{def:FST}
    Let $i$ be the wallet with the smallest score in round $\fst$. $\fst$ is a forced stopping time if $i \not \in \advAccount$.
\end{definition}

\begin{lemma}[Lemma 4.2 from \citealp{FHWY22}] \label{thm:FSTisST}
    If $\fst$ is a forced stopping time, $\fst$ is a stopping time.
\end{lemma}

We will use $\fst$ to denote the first forced stopping time of the adversary. We will only consider forced stopping times (and not any `unforced' stopping times) for the remainder of the paper. Because of this and for convenience, we abuse notation and refer to forced stopping times plainly as stopping times.

\subsection{The Omniscient Adversary}\label{sec:OmniMain}
As a warm up, we look at the omniscient adversary studied by \citet{FHWY22} that is stronger than the adversary in $\CSSPA$ in the following aspects:

\begin{itemize}
    \item $\beta = 1$. Further, for any $r \geq 0$ and $(i_t, i_{t+1}, \dots, i_{t+r}) \in \big( \advAccount \cup \{\knownHonest\} \big)^{r+1}$, the adversary pre-computes the credentials $\cred{i_{t+r'}}{t+r'}$ for $1 \leq r' \leq r$ assuming $i_{t+\hat{r}}$ is elected to lead in round $t+\hat{r}$ for all $0 \leq \hat{r} < r'$. In other words, the omniscient adversary can precompute hypothetical future credentials even when $\knownHonest$ is elected to be the leader.
    \item $X_t = 1$ for all rounds $t < \fst$, i.e., the omniscient adversary is rewarded to delay the first stopping time, even if it entails being elected only for a very small fraction of rounds. By \Cref{thm:ResetST}, the omniscient reward
    \begin{equation*}
        \notag
        \RewOmni{\strategy} = \frac{\mathbb{E}[\sum_{t = 1}^{\fst} X_t(\pi)]}{\mathbb{E}[\fst]} = \frac{\mathbb{E}[\fst - 1]}{\mathbb{E}[\fst]} = 1-\frac{1}{\mathbb{E}[\fst]}
    \end{equation*}
\end{itemize}
Similar to the adversary in $\CSSPA$, we define $\fst$ to be a stopping time for the omniscient adversary if the set of adversarial potential winners $\advPot$ for round $\fst$ is empty. We defer our discussions on the omniscient adversary to \Cref{sec:Omni}. We summarize our findings in \Cref{thm:OmniSummary}.

\begin{theorem} \label{thm:OmniSummary}
    For the omniscient adversary with stake $\alpha$, there exists a constant $\omniconst \approx 0.38$ such that,
    \begin{enumerate}
        \item for $\alpha > \omniconst$, there exists a strategy $\strategy$ such that $\E[\fst]$ is unbounded and $\RewOmni{\strategy} = 1$, and,
        \item for $\alpha \leq \omniconst$ and any strategy $\strategy$, $\E[\fst] \leq \tfrac{1- 3\alpha + 3\alpha^2 - 3 \alpha^3}{(1-3\alpha+\alpha^2)(1-\alpha+\alpha^2)}$ and $\RewOmni{\strategy} \leq \alpha \cdot \big( \tfrac{1 - 2\alpha + \alpha^2 - \alpha^3}{1- 3\alpha + 3\alpha^2 - 3 \alpha^3} \big)$.
    \end{enumerate}
\end{theorem}
We also show a non-closed form upper bound on the optimal omniscient rewards that can be made tight up to an arbitrarily small additive error. The plots comparing our upper bounds to \citet{FHWY22} can be found in \Cref{fig:Omni} in \Cref{sec:Omni}.

The following results on the size of potential winners and first stopping time of the omniscient adversary would be bootstrapped further to get upper bounds on the optimal rewards of the actual adversary in $\CSSPA$.

 \begin{lemma} [Corollary 4.1 from \citealp{FHWY22}] \label{thm:PrPotentialWinner}
     For a random seed $\seed{t} \sim U[0, 1]$ and $i^* \geq 0$, the probability $Pr(|\potWinners_t(\seed{t})| = i^* + 1)$ of the adversary having exactly $i^* +1$ potential winners equals $\alpha^{i^*} \, (1-\alpha)$.
 \end{lemma}

\begin{lemma} \label{thm:PrStoppingTime}
    For the omniscient adversary with stake $\alpha < \omniconst$, $Pr(\fst > 0) = 1$ and $Pr(\fst > 1) = \alpha$. For $T \geq 2$,
    $$Pr(\fst > T) \leq \alpha^2 \cdot \tfrac{2 - 2\alpha + \alpha^2}{1-\alpha+\alpha^2} \cdot \big(\alpha \cdot  \tfrac{2-\alpha}{1-\alpha}\big)^{T-2}\text{.}$$
    
\end{lemma}

\section{Estimating the Optimal Adversarial Reward}
We proceed to designing simulations that estimate the expected adversarial reward from playing a strategy $\strategy$ in $\CSSPA$. We find the optimal adversarial strategy in $\CSSPA$ quite complex to describe. We reformulate $\CSSPA$ such that a succinct description of the optimal adversarial strategy becomes possible. Then, we propose a simulation that computes the adversary's optimal reward precisely but requires an infinite run-time. We describe a sequence of modifications to the simulation that trades off run-time for precision to get provable bounds on the adversary's optimal rewards.

\subsection{A Linear Version of $\CSSPA$}
Optimizing the adversarial reward $\Rew{\strategy} = \tfrac{\E[\sum_{t = 1}^{\fst} X_t(\strategy)]}{\E[\fst]}$ depends on maintaining a balance between getting a myopic gain in the reward by winning the election in the current round and a long-term gain through delaying the first stopping time. This inherent trade-off between the short-term and long-term gains of actions in $\CSSPA$ makes both describing the optimal adversarial strategy and simulating it hard. We reformulate $\CSSPA$ through an approach  similar to \citet{SSZ16} that allows the adversary to myopically optimize its reward without having to worry about long-term consequences.

We introduce an entry fee $\lambda$ in $\CSSPA$ that the adversary is charged to participate in each round and consider the total adversarial reward instead of the rate at which the adversary is elected.

\begin{definition}[$\LinearCSSPA$] \label{def:LinCSSPA}
    In $\LinearCSSPA$, the network consists of three players- the adversary with stake $\alpha$, two honest players $\knownHonest$ and $\opaqueHonest$ with stakes $\beta \, (1-\alpha)$ and $(1-\beta) \, (1-\alpha)$ respectively. Prior to round $1$, the adversary learns the values of $\alpha$, $\beta$, the entry fee $\lambda$ and that $\knownHonest$ and $\opaqueHonest$ are honest. For $n \xrightarrow{} \infty$, the adversary distributes its stake into a set $\advAccount$ of $n$ wallets, each containing a stake $\tfrac{\alpha}{n}$. The adversary makes the following decisions in round $t$:
    \begin{enumerate}
        \item The adversary pays an entry fee $\lambda$.
        \item The adversary learns the seed $\seed{t}$ of round $t$.
        \item The adversary computes the credentials $\cred{i}{t}$ for all wallets $i \in \advAccount$. Further, the adversary learns the credentials $\cred{\knownHonest}{t}$. The adversary knows that the credential of player $\opaqueHonest$ is drawn from $\expd{(1-\beta) \, (1-\alpha)}$ but does not learn $\cred{C}{t}$.
        \item For any $r \geq 0$ and $(i_t, i_{t+1}, \dots, i_{t+r}) \in \big( \advAccount \cup \{\knownHonest\} \big) \times \advAccount^r$, the adversary precomputes the credentials $\cred{i_{t+r'}}{t+r'}$ for $1 \leq r' \leq r$ assuming $i_{t+\hat{r}}$ is elected to lead in round $t+\hat{r}$ for all $0 \leq \hat{r} < r'$.
        \item The adversary either remains silent or broadcasts the credential of a wallet $i \in \advAccount$.
        \item The game terminates if either $\knownHonest$ or $\opaqueHonest$ have scores smaller than all wallets $i \in \advAccount$, i.e, a stopping time is reached.
    \end{enumerate}
\end{definition}

\begin{definition}[Reward of a Strategy] \label{def:LinRew}
    For a strategy $\strategy$ played by the adversary in $\LinearCSSPA$, let the Bernoulli random variable $X_t(\strategy)$ be $1$ if the adversary is elected in round $t$ and $0$ otherwise. The adversary earns an expected reward
    $$\LinRew{\strategy} = \E \big[\sum_{t = 1}^{\fst} (X_t(\strategy) - \lambda) \big]$$
\end{definition}

We conclude the discussion by relating the rewards $\Rew{\strategy}$ in $\CSSPA$ and $\LinRew{\strategy}$ in $\LinearCSSPA$.
\begin{theorem} \label{thm:Signlambda}
    For an entry fee $\lambda$ and a strategy $\strategy$, $\LinRew{\strategy} > 0$ (resp. $\LinRew{\strategy} < 0$) if and only if $\lambda < \Rew{\strategy}$ (resp. $\lambda > \Rew{\strategy}$). Further, $\LinRew{\strategy} = 0$ when $\lambda = \Rew{\strategy}$.
\end{theorem}
A standard binary search would locate the value of $\lambda$ such that $\LinRew{\strategy} = 0$, in which case, $\Rew{\strategy} = \lambda$. We defer the proof to \Cref{sec:ProofofSignLambda}.

\subsection{The Ideal Simulation} \label{sec:Ideal}

$\LinearCSSPA$ has a recursive structure that we exploit while designing simulations to estimate $\LinRew{\strategy}$ of a strategy $\strategy$. In some round $t$, by broadcasting the credential $\cred{i}{t}$ of a wallet $i \in \advAccount$ and winning the election (or remaining silent and letting $\knownHonest$ with credential $\cred{\knownHonest}{t}$ win), the adversary induces an instance of $\LinearCSSPA$ with an initial seed $\seed{0} = \cred{i}{t}$. If the adversary recursively "knew" the expected future reward $r_i$ that it would get from broadcasting $\cred{i}{t}$ for each adversarial potential winner $i$ and the reward $r_0$ from letting $\knownHonest$ win, the adversary can decide its actions just based on $\vecr = \big( r_i \big)_{i \in \N \cup \{0\}}$ and the scores $\vecc = \big( c_i \big)_{i \in \N \cup \{0\}}$ of the wallets in $\advAccount \cup \{\knownHonest\}$\footnote{This is not entirely true. The adversary can still play strategies based on the credentials and rewards from previous rounds. However, given that the adversary's goal is to optimize the total rewards earned across rounds, such strategies can be safely ignored. As we will see, the optimal strategy can be codified in this language.}. Of course, the adversary always runs the risk of the current round being a stopping time, in which case, the total future rewards earned by the adversary equals zero.

Let $\dis{\strategy}$ be the distribution of rewards the adversary achieves by playing the strategy $\strategy$ in $\LinearCSSPA$. We are interested in estimating the expected reward $\E_{s \sim \dis{\strategy}}[s]$. We do so by constructing the $\CDF$ of $\dis{\strategy}$ in $\addl{\strategy, \dis{\strategy}}$ by sampling from $\dis{\strategy}$ infinitely many times. This can be done by setting up infinitely many `induced-instances' of $\LinearCSSPA$. For each of these induced-instances:
\begin{enumerate}
    \item For each $i \geq 1$, we sample the $i\textsuperscript{th}$ smallest score $c_i$ amongst all adversarial wallets using \Cref{thm:SmallestScore}:
    $$c_1 \xleftarrow{} \expd{\alpha}, c_i \xleftarrow{} c_{i-1} + \expd{\alpha}$$
    \item For each $i \geq 1$, we sample the reward $r_i$ earned from broadcasting the credential of the adversarial wallet with the $i\textsuperscript{th}$ smallest score from the distribution $\dis{\strategy}$.
    \item We compute the adversarial reward in expectation over the reward $r_0$ from letting $\knownHonest$ win, the scores $c_0$ and $c_{-1}$ of $\knownHonest$'s and $\opaqueHonest$'s wallets respectively by simulating the behaviour of $\strategy$ for scores $\vecc = \big( c_i \big)_{i \in \N \cup \{0\}}$ and rewards $\vecr = \big( r_i \big)_{i \in \N \cup \{0\}}$.
\end{enumerate}
Estimating the adversarial reward reduces to finding a fixed-point $\dis{\strategy}$ to $\addl{\strategy, \dis{\strategy}}$.

\begin{mdframed}
$\addl{\strategy, \dis{\strategy}}$:
\begin{enumerate}
    \item For $1 \leq \ell \leq \challengeSample{\samples_t = \infty}$.
    \begin{enumerate}
        \item $\drawadv{\dis{\strategy}}$:
        \begin{itemize}
            \item Sample $\AdvR$: Draw $\challengeTk{\coin = \infty}$ rewards $r_1, r_2, \dots, r_{\coin}$ i.i.d~from $\dis{\strategy}$.
            \item Sample $\AdvC$: Draw $\challengeTk{\coin = \infty}$ scores $c_1, c_2, \dots, c_{\coin}$ of adversarial wallet as follows. Draw $c_1 \xleftarrow{} \expd{\alpha}$ and $c_{i+1} \xleftarrow{} c_i + \expd{\alpha}$ (a fresh sample for each $i$) for $1 \leq i \leq \coin -1$. For convenience, set $c_{\coin + 1} = \infty$.
            \item Return $(\AdvR, \AdvC)$.
        \end{itemize}
        \item $\sample{\strategy, \dis{\strategy}}$:
        Simulate the action of the strategy $\strategy$ in the current round given $\dis{\strategy}$, $\AdvR$ and $\AdvC$. Return the reward $s_{\ell}$ in expectation over the reward $r_0$ from letting $\knownHonest$ win, $\knownHonest$'s score $c_0$ and $\opaqueHonest$'s score $c_{-1}$.
    \end{enumerate}
    \item Return $\dis{\strategy}$ to be the uniform distribution over $\{s_1, s_2, \dots, s_{\samples_t}\}$.
\end{enumerate}

$\simulate{\strategy}$:
\begin{enumerate}
    \item \challengeConverge{Compute a fixed-point $\dis{\strategy}$ to $\addl{\strategy, \dis{\strategy}}$.}
    \item Return $\E_{s \sim \dis{\strategy}}[s]$.
\end{enumerate}
\end{mdframed}

See \Cref{sec:Tables} for a summary of the notations and functions used in the simulations.

\subsection{The Optimal Solution}
We describe the optimal strategy $\opt$ in the recursive language introduced in \Cref{sec:Ideal}. Let $\optdis = \dis{\opt}$ be the distribution of rewards from playing $\opt$. At a start of a round $t$, the adversary learns the scores $\vecc$ of wallets in $\advAccount \cup \{\knownHonest\}$ and rewards $\vecr$ from remaining silent and from broadcasting the credential of each wallet in $\advAccount$. We use $\AdvC$ and $\AdvR$ to denote the scores $\big( c_i \big)_{i \in \N}$ and rewards $\big( r_i \big)_{i \in \N}$ associated with the wallets in $\advAccount$. Re-index the adversary's wallets (and therefore, $\AdvC$ and $\AdvR$) in increasing order of its scores. Let $i^*(\vecc) := |\{i > 0 | c_i \leq c_0\}|$ be the number of adversarial potential winners.

We compare the expected future rewards of all possible actions the adversary can take at the start of round $t$.
\begin{enumerate}
    \item Suppose the adversary abstains from broadcasting. It earns a reward $r_0$ unless a stopping time is reached, which happens either when $i^*(\vecc) = 0$ or when $\opaqueHonest$ has a score $c_{-1} \sim \expd{(1-\beta) \, (1-\alpha)}$ smaller than the score $c_0$ of $\knownHonest$. The probability of $\opaqueHonest$ having a larger score than $c_0$ equals $e^{-c_0 \, (1-\beta) \, (1-\alpha)}$. Thus, the expected reward from remaining silent equals $e^{-c_0 \, (1-\beta) \, (1-\alpha)}r_0 \cdot \mathbbm{1}(i^*(\vecc) \neq 0)$. We define
    $$h(c_0, r_0) := e^{-c_0 \, (1-\beta) \, (1-\alpha)}r_0$$
    \item From broadcasting the credential of an adversarial potential winner $i$ with score $c_i$ and future reward $r_i$, the adversary earns a reward $1$ from getting elected in the current round and thus, a total reward $(1+r_i)$. This, once again, is subject to the current round not being a stopping time. The current round is not a forced stopping time if $\opaqueHonest$ has a score larger than $c_i$, which happens with probability $e^{c_i \, (1-\beta) \, (1-\alpha)}$, and if $i^*(\vecc) \neq 0$. The adversary has a potential winner $i$ and $i^*(\vecc)$ is at least $1$ as a consequence. Hence, the expected reward from broadcasting the credential of $i$ equals $e^{c_i \, (1-\beta) \, (1-\alpha)} \, (1+r_i)$. The adversary can broadcast the credential of the wallet that maximizes its reward to earn
    $$g(c_0, \AdvC, \AdvR) = \max_{i \leq i^*(\vecc)} \{e^{-c_i \, (1-\beta) \, (1-\alpha)}(1 + r_i)\}$$
\end{enumerate}
The adversary also pays an entry fee $\lambda$. Between remaining silent and broadcasting its best credential, the adversary wins
\begin{equation*}
    \notag
    \begin{split}
        \max \{h(c_0, r_0) &\mathbbm{1}(i^*(\vecc) \neq 0), g(c_0, \AdvC, \AdvR)\} - \lambda \\
        &= \max \{h(c_0, r_0) \mathbbm{1}(i^*(\vecc) \neq 0), g(c_0, \AdvC, \AdvR) \mathbbm{1}(i^*(\vecc) \neq 0)\} - \lambda \\
        &= \max \{h(c_0, r_0), g(c_0, \AdvC, \AdvR)\} \mathbbm{1}(i^*(\vecc) \neq 0) - \lambda
    \end{split}
\end{equation*}

While using the future rewards from round $t$ to compute the optimal action to take in round $t-1$, the adversary will not know the values $c_0$ and $r_0$ since $\knownHonest$ does not broadcast $\cred{\knownHonest}{t}$ until the start of round $t$. The adversary can only compute the future rewards from playing an action in expectation over $c_0$ and $r_0$. With this in mind, we construct $\optdis$ to be the distribution of (future) rewards in expectation over $c_0$ and $r_0$. Given $\AdvC$ and $\AdvR$, we implement a sampling procedure $\sample{\opt, \optdis}$ by setting the $\ell^{\text{th}}$ sample $s_{\ell}$ to be
$$\E_{c_0 \sim \expd{(1-\beta) \, (1-\alpha)}, r_0 \sim \optdis}[\max \{h(c_0, r_0), g(c_0, \AdvC, \AdvR)\} \mathbbm{1}(i^*(\vec{c}) \neq 0)] - \lambda$$

Finding a fixed-point $\optdis$ for $\addl{\opt, \optdis}$ seems intractable and we resort to heuristic methods instead. One natural heuristic would be to begin at the point-mass distribution $\dist{0}$ at $0$ and iterate infinitely many times to get the sequence $\big( \dist{t} \big)_{t \in \N \cup \{0\}}$ of distributions satisfying $\dist{t+1} = \addl{\opt, \dist{t}}$. We end this section by summarizing the challenges in executing the above heuristic.

\begin{mdframed}
$\addl{\opt, \dist{t-1}}$:
\begin{enumerate}
    \item For $1 \leq \ell \leq \challengeSample{\samples_t = \infty}$.
    \begin{enumerate}
        \item $\drawadv{\dist{t-1}}$:
        \begin{itemize}
            \item Sample $\AdvR$: Draw $\challengeTk{\coin = \infty}$ rewards $r_1, r_2, \dots, r_{\coin}$ i.i.d~from $\dist{t-1}$.
            \item Sample $\AdvC$: Draw $\challengeTk{\coin = \infty}$ scores $c_1, c_2, \dots, c_{\coin}$ of adversarial wallet as follows. Draw $c_1 \xleftarrow{} \expd{\alpha}$ and $c_{i+1} \xleftarrow{} c_i + \expd{\alpha}$ (a fresh sample for each $i$) for $1 \leq i \leq \coin -1$. For convenience, set $c_{\coin + 1} = \infty$.
            \item Return $(\AdvR, \AdvC)$.
        \end{itemize}
        \item $\sample{\opt, \dist{t-1}}$: \\
        Return sample $s_{\ell}$ equal to
        \challengeComp{
        \begin{equation}
        \notag
            \begin{split}
                &\E_{c_0 \sim \expd{\beta \, (1-\alpha)}, r_0 \sim \dist{t-1}}\big[\max_{0 \leq i \leq i^*(\vecc)} \{e^{-c_i\cdot (1-\beta)\cdot(1-\alpha)} \, (r_i+ \mathbbm{1}(i \neq 0)) \} \cdot \mathbbm{1}(i^*(\vecc) \neq 0) \big] -\lambda
            \end{split}
        \end{equation}}
    \end{enumerate}
    \item Return $\dist{t}$ to be the uniform distribution over $\{s_1, s_2, \dots, s_{\samples_t}\}$.
\end{enumerate}

$\simulate{\opt}$:
\begin{enumerate}
    \item Initialize $\dist{0}$ to be the point-mass distribution at $0$.
    \item For $1 \leq t \leq \challengeTk{\round = \infty}$:
    \begin{enumerate}
        \item \challengeConverge{$\dist{t} = \addl{\opt, \dist{t-1}}$.}
    \end{enumerate}
    \item Return $\E_{s \sim \dist{\round}}[s]$.
\end{enumerate}
\end{mdframed}


\begin{enumerate}
    \item \challengeConverge{The iterated-point heuristic does not guarantee convergence. Even if the iteration converges, there could be a multitude of fixed-points and the iteration could converge to a distribution that is not the optimal reward.}
    \item \challengeTk{We run $\addl{\opt, \cdot}$ $\round = \infty$ many times. In each execution of $\addl{\opt, \cdot}$, the adversary can pick one of $\coin = \infty$ actions-- one each for broadcasting credentials of  wallets $i \in \advAccount$ and one for staying silent. $\opt$ compares the rewards of each of these actions before making a decision.}
    \item \challengeSample{Given a distribution $\D_t$, we compute $\addl{\opt, \D_t}$ by constructing $\samples_t = \infty$ samples. As we will see in \Cref{sec:Samples}, the simulation is not even an unbiased estimator of the reward $\E_{s \sim \dist{\round}}[s]$ once we constrain $\samples_t$ to be finite.}
    \item \challengeComp{The sample $s_{\ell}$ is constructed by computing the reward in expectation over $\knownHonest$'s score $c_0$ and reward $r_0$ from remaining silent. This involves calculating a double integral. The integrals can be calculated in finite-time by approximating them by a Riemann sum. However, even a polynomial run-time would not be practical due to the sheer number of samples we construct. We require a linear run-time.}
\end{enumerate}

\subsection{Moving from Ideal to Practical}

\subsubsection{\challengeConverge{Convergence of the Iterated-Point Heuristic}}
We discuss the natural variant $\FinLinCSSPA{\round}$ that terminates after $\round$ rounds if a stopping time has not been reached yet. We will argue that the distribution of optimal rewards $\big( \dist{\round} \big)_{\round \in \N \cup \{0\}}$ satisfies the same recursion as the iterated-point heuristic on $\addl{\opt, \cdot}$ and converges to $\optdis$ as $\round \xrightarrow{} \infty$.

When $\round = 0$, the game terminates even before it starts and the adversary gets a total reward zero. Thus, $\dist{0}$ is the point-mass on zero, identical to the initial point of the iterated-point heuristic. $t$ rounds before termination, by broadcasting the credential $\cred{i}{-t}$ of a wallet $i \in \advAccount$ ($i = \knownHonest$ if the adversary remains silent), the adversary induces an instance of $\FinLinCSSPA{t-1}$ with an initial seed $\seed{0} = \cred{i}{-t}$. Thus, if the adversary recursively knew the rewards $r_i \sim \dist{t-1}$ from each potential winner $i$, the adversary would broadcast the credential (or stay silent) that would maximize its reward from the last $t-1$ rounds. This is the same operation performed by $\addl{\opt, \cdot}$ on $\dist{t-1}$. By induction, the distribution of rewards $\dist{t}$ equals the reward distribution $\addl{\opt, \dist{t-1}}$ output by the $t\textsuperscript{th}$ iteration of the iterated-point method.

As $\round \xrightarrow{} \infty$, the distribution of rewards $\round$ rounds before termination and $\round - 1$ rounds before termination are identical. This is equivalent to claiming $\dist{\round}$ as $\round \xrightarrow{} \infty$ approaches the reward distribution $\optdis$. Thus, the iterated-point method converges and converges to the correct distribution of rewards.

\subsubsection{\challengeTk{Infinite Rounds and Credentials}}

Let $\FinLinCSSPA{\round}$ be the variant of $\LinearCSSPA$ terminating after round $\round$. $\simulate{\opt}$ loops infinitely to construct the distribution of adversarial rewards in $\FinLinCSSPA{\round}$ as $\round \xrightarrow{} \infty$. Further, for each round of the simulation, $\drawadv{\dist{t}}$ samples a score and a reward for each of the infinite wallets the adversary operates. We revisit $\CSSPA$ and argue that terminating $\CSSPA$ after $\round$ rounds and constraining the adversary to broadcasting the credentials of a wallet only if it is amongst the $k$ smallest scores in $\advAccount$ does not cause a significant drop in the adversary's optimal reward. Once established, we can estimate the reward from playing $\optrew{\round, \coin}$, the optimal strategy in $\CSSPA$ that terminates after $\round$ rounds and never uses a score outside the $k$ smallest scores in $\advAccount$, instead of estimating the reward from $\opt$.

We abuse notation to describe the optimal strategy of the adversary in $\CSSPA$ as $\opt$ and its reward distribution by $\optdis$. We say the adversary is $\coin$-scored if the adversary is constrained to either stay silent or broadcast a credential amongst its wallets with the $k$ smallest scores.

\begin{theorem} \label{thm:(Tk)CSSPA}
    For $\alpha \leq 0.29$ and a $k$-scored adversary, the difference in the expected rewards between playing $\opt$ and $\optrew{\round, \coin}$ in $\CSSPA$ satisfies
    $$0 \leq |\Rew{\opt} -\Rew{\optrew{\round, \coin}}| \leq \alpha^2 \cdot \tfrac{2 - 2\alpha + \alpha^2}{1 - \alpha + \alpha^2} \cdot [\alpha \cdot \tfrac{2-\alpha}{1-\alpha}]^{\round-2} + \alpha^\coin$$
\end{theorem}
\noindent We defer the proof to \Cref{sec:Proofof(Tk)CSSPA}.

We estimate the $\coin$-scored adversary's optimal reward in $\FinLinCSSPA{\round}$ and binary search over $\lambda$ to estimate the adversarial reward $\Rew{\optTk}$ in $\CSSPA$. We can then upper bound and lower bound the optimal reward $\Rew{\opt}$ by
\begin{equation}
    \notag
    \Rew{\optTk} \leq \Rew{\opt} \leq \Rew{\optTk} + \alpha^2 \cdot \tfrac{2 - 2\alpha + \alpha^2}{1 - \alpha + \alpha^2} \cdot [\alpha \cdot \tfrac{2-\alpha}{1-\alpha}]^{\round-2} + \alpha^\coin
\end{equation}

We abuse notation as usual and call the $\coin$-scored adversary's optimal strategy in $\FinLinCSSPA{\round}$ as $\optTk$. Let $\dist{\round, \coin}$ be our estimate of the distribution of rewards from playing $\optTk$.  We modify the simulation to terminate after $\round$ rounds and bake the $\coin$-scored adversary into $\addl{\optTk, \cdot}$.

\begin{mdframed}
$\kaddl{\dist{t-1, k}}$:
\begin{enumerate}
    \item For $1 \leq \ell \leq \challengeSample{\samples_t = \infty}$.
    \begin{enumerate}
        \item $\drawadv{\dist{t-1, k}}$:
        \begin{itemize}
            \item Sample $\AdvR$: Draw $\coin$ rewards $r_1, r_2, \dots, r_{\coin}$ i.i.d~from $\dist{t-1}$.
            \item Sample $\AdvC$: Draw $\coin$ scores $c_1, c_2, \dots, c_{\coin}$ of adversarial wallet as follows. Draw $c_1 \xleftarrow{} \expd{\alpha}$ and $c_{i+1} \xleftarrow{} c_i + \expd{\alpha}$ (a fresh sample for each $i$) for $1 \leq i \leq \coin -1$. For convenience, set $c_{\coin + 1} = \infty$.
            \item Return $(\AdvR, \AdvC)$.
        \end{itemize}
        \item $\sample{\optTk, \dist{t-1, k}}$: Return sample $s_{\ell}$
        \end{enumerate}
    \item Return $\estdist{t, \coin}$ to be the uniform distribution over $\{s_1, s_2, \dots, s_{\samples_t}\}$.
\end{enumerate}

$\tsimulate$:
\begin{enumerate}
    \item Initialize $\dist{0, \coin}$ to be the point-mass distribution at $0$.
    \item For $1 \leq t \leq \round$:
    \begin{enumerate}
        \item $\dist{t, \coin} = \kaddl{\dist{t-1, k}}$.
    \end{enumerate}
    \item Return $\E_{s \sim \dist{\round, \coin}}[s]$.
\end{enumerate}

\end{mdframed}

\subsubsection{\challengeSample{Constructing Infinitely many Samples for $\kaddl{\cdot}$}} \label{sec:Samples}

We address the infinite run-time for $\kaddl{\cdot}$ from constructing infinitely many samples to perfectly describe the $\CDF$ of its output. For an input distribution $\D_0$, we want to approximate the sequence of distributions $\big(\D_t\big)_{0 \leq t \leq \round}$ such that $\D_{t} := \kaddl{\D_{t-1}}$ while constructing only finitely many samples for each of them. In \Cref{sec:McD}, we discuss the challenges from using the most natural technique to bound the error from estimation in our simulations-- Chernoff bounds or McDiarmid's inequality followed by a union bound. Even more importantly, we also find that the estimator that arises from constructing finite number of samples might not even be unbiased (\Cref{sec:Bias}).

We tackle the above challenges by maintaining two distributions, $\UBD{t}$ that dominates $\D_t$  and $\LBD{t}$ that is dominated by $\D_t$. We sketch our method for constructing an empirical distribution that is dominated by the true distribution. Suppose for a sufficiently large number of samples, we can guarantee that, for all values $r$, the quantile $\tilde{q}$ in the empirical distribution constructed by sampling $\samples$ times and the quantile $\q$ in the true distribution satisfy $\tilde{q} \in [\q - \delta, \q + \delta]$. Dropping the $\delta$ smallest (strongest) quantiles and replacing them with $\delta \samples$ samples of the infimum of the true distribution would give us a new estimated distribution that is dominated by the true distribution. We call this process deflation. We will compute $\LBD{t}$ by first computing $\kaddl{\LBD{t-1}}$ and then deflating the outcome by a suitable parameter $\delta$. By a straightforward induction, $\LBD{t}$ is dominated by $\D_t$.

We can construct $\UBD{t}$ from $\UBD{t-1}$ through an analogous inflation procedure. However, the error in the estimated reward due to inflation is much larger than the error due to deflation. An inflated reward with a quantile $\Tilde{\q} \in [0, \delta]$ is much more likely to be chosen by an adversary, since the adversary picks its optimal future rewards). This leaves a bigger impact on the estimated reward and influences the rewards in successive rounds too. This differs from deflate since a deflated reward with a quantile $\tilde{q} \in [1-\delta, 1]$ is very likely to be ignored by the adversary since the reward from a different wallet is likely to be higher. To mitigate the strong credentials from drifting the estimated reward far way from $\E_{s \sim \D_t}[s]$, we perform a more nuanced inflation procedure.

\begin{mdframed}
    $\Defl{\chernoff, \D}:$ Given an input $D$ drawn uniformly from $\samples$ samples,
    \begin{enumerate}
        \item Delete the largest $\samples \cdot \sqrt{\tfrac{\ln \chernoff^{-1}}{2\samples}}$ samples from $D$
        \item Append $\samples \cdot \sqrt{\tfrac{\ln \chernoff^{-1}}{2\samples}}$ copies of $-\lambda$ to $\D$
    \end{enumerate}
    $\Infl{\chernoff, \strat, \D}$: Given an input $D$ drawn uniformly from $\big( s_{\ell} \big)_{1 \leq \ell \leq \samples}$ (in descending order),
    \begin{enumerate}
        \item Delete the smallest $\samples \cdot \sqrt{\tfrac{\ln \chernoff^{-1}}{2\samples}}$ samples from $\D$
        \item Append $\strat \samples$ copies of $t\,(1-\lambda)$ to $\D$
        \item For $1 \leq \ell < \frac{\samples}{\strat \, \samples} \cdot \sqrt{\tfrac{\ln \chernoff^{-1}}{2 \samples}}$:
        \begin{itemize}
            \item Append $\strat \samples$ copies of $s_{\ell}$
        \end{itemize}
    \end{enumerate}
\end{mdframed}

\begin{theorem} \label{thm:SampleSummary}
    Let $\tsimulateextraparameters{\estdist{t - 1, \coin}}$ output an upper bound $\estUBD$ and a lower bound $\estLBD$. Then,
    \begin{enumerate}
        \item With probability at least $1- \round \, \Big(\chernoff +  \tfrac{e^{-\strat \samples}}{\strat} \, \sqrt{\tfrac{\ln \chernoff^{-1}}{2\samples}} \Big)$, $\E_{s \sim \estUBD}[s] \geq \E_{s \sim \dist{\round, \coin}}[s]$.
        \item With probability at least $1 - \round \, \chernoff$, $\E_{s \sim \estLBD}[s] \leq \E_{s \sim \dist{\round, \coin}}[s]$.
    \end{enumerate}
\end{theorem}
We defer the proof to \Cref{sec:InflDefl}.

Our estimate is not precisely equal to $\dist{\round, \coin}$ and to differentiate the two, we denote our estimation of $\dist{\round, \coin}$ by $\estdist{\round, \coin}$. We update the simulations to contain inflate and deflate.

\begin{mdframed}
$\finsampleaddl{\estdist{t-1, \coin}}$:
\begin{enumerate}
    \item For $1 \leq \ell \leq \samples$.
    \begin{enumerate}
        \item $\drawadv{\estdist{t-1, k}}$.
        \item $\sample{\optTk, \estdist{t-1, k}}$: Return sample $s_{\ell}$
    \end{enumerate}
    \item $\tilde{\D}^{\optimal}_{t, \coin}$ be the uniform distribution over $\big(s_{\ell}\big)_{1 \leq \ell \leq \samples}$ (in descending order)
    \item Inflate while computing the upper bound and deflate while computing the lower bound.
    \begin{itemize}
        \item $\Defl{\chernoff, \tilde{\D}^{\optimal}_{t, \coin}}:$
        \begin{enumerate}
            \item Delete the largest $\samples \cdot \sqrt{\tfrac{\ln \chernoff^{-1}}{2\samples}}$ samples from $\tilde{\D}^{\optimal}_{t, \coin}$
            \item Append $\samples \cdot \sqrt{\tfrac{\ln \chernoff^{-1}}{2\samples}}$ copies of $-\lambda$ to $\tilde{\D}^{\optimal}_{t, \coin}$
        \end{enumerate}
        \item $\Infl{\chernoff, \strat, \tilde{\D}^{\optimal}_{t, \coin}}$:
        \begin{enumerate}
            \item Delete the smallest $\samples \cdot \sqrt{\tfrac{\ln \chernoff^{-1}}{2\samples}}$ samples from $\tilde{\D}^{\optimal}_{t, \coin}$
            \item Append $\strat \samples$ copies of $t\,(1-\lambda)$ to $\tilde{\D}^{\optimal}_{t, \coin}$
            \item For $1 \leq \ell < \frac{\samples}{\strat \, \samples} \cdot \sqrt{\tfrac{\ln \chernoff^{-1}}{2 \samples}}$:
            \begin{itemize}
                \item Append $\strat \samples$ copies of $s_{\ell}$
            \end{itemize}
        \end{enumerate}
    \end{itemize}
    
    \item Return $\estdist{t, \coin}$ to be the uniform distribution over $\{s_1, s_2, \dots, s_{\samples_t}\}$.
\end{enumerate}

$\tsimulateextraparameters{\estdist{t - 1, \coin}}$:
\begin{enumerate}
    \item Initialize $\estdist{0, \coin}$ to be the point-mass distribution at $0$.
    \item For $1 \leq t \leq \round$:
    \begin{enumerate}
        \item $\estdist{t, \coin} = \finsampleaddl{\estdist{t-1, \coin}}$.
    \end{enumerate}
    \item Return $\E_{s \sim \estdist{\round, \coin}}[s]$.
\end{enumerate}

\end{mdframed}

\subsubsection{\challengeComp{Computing Expectations}}

In this section, we tackle the final challenge of needing to compute integrals accurately while constructing the sample $s_{\ell}$. This involves computing a double integral, one over the score $c_0$ of $\knownHonest$ and the reward $r_0$ from staying silent. Naively integrating over the reward distribution from the previous round described by $\samples$ samples would result in a run-time of $\Omega(\samples)$ for each of the $\samples$ samples, and consequently $\Omega(\samples^2)$ for simulating one round. Even though this is $\poly(\samples)$, it still turns out to be intractable to run for the extremely large number of samples we expect to handle each round. Instead, we aim to reduce the run-time to $\Tilde{O}(\samples)$.

To begin with, observe that for each sample $s_{\ell}$, we are drawing $\coin$ scores for the wallets of the $\coin$-scored adversary. Thus, we end up with a run-time of $\Omega(\coin \cdot \samples)$ no matter what we do. Any additional compute that we perform for each sample is only going to increase the order of the run-time. Thus, our aim is to run as much pre-compute as possible before even constructing the first sample, and minimize the number of fresh computations needed to be performed with each sample. We get a practical run-time through a combination of pre-computes and by computing the integrals involved with taking expectations over $r_0$ and $c_0$ as a discrete sum. We defer the details to \Cref{sec:PreComp}. Note that we introduce two parameters $\epsilon$ and $\eta$ that reflect the precision to which we discretize the distributions constructed during the simulations and the precision to which we approximate the two integrals. We do all our pre-computations through the function $\Precomp{\apxd}$.

\begin{lemma} \label{thm:PrecompRunTime}
    $\Precomp{\apxd{}}$ terminates in $O(\frac{t}{\epsilon \, \eta})$ time.
\end{lemma}

$\tcsimulate$ in \Cref{sec:Summary} reflects the updates in terms of the precomputations. \Cref{sec:Summary} also compiles the changes made to the simulations across various stages and presents a summary.

\begin{lemma} \label{thm:RunTime}
    A single execution of $\tcsimulate$ terminates in time $O(Tkn + \tfrac{T^2}{\epsilon \, \eta})$.
\end{lemma}

\subsection{Locating the Optimal Expected Reward and the Optimal Adversarial Strategy}

Remember that the end goal of the simulations is to compute the optimal reward for an adversary playing $\CSSPA$. We denote the optimal $k$-scored adversarial strategy in $\FinLinCSSPA{\round}$ by $\optTk(\lambda)$ and in $\CSSPA$ that terminates in $\round$ rounds by $\optTk$. Let $\LinLamRew{\lambda}{\strategy}$ be the expected reward from playing $\strategy$ in $\FinLinCSSPA{\round}$. Remember that $\Rew{\round}$ is the expected reward from playing $\strategy$ in $\CSSPA$. By \Cref{thm:Signlambda}, $\LinLamRew{\lambda}{\optTk} \geq 0$ iff $\lambda \leq \Rew{\optTk}$ with equality holding precisely when $\lambda = \Rew{\optTk}$. To estimate $\Rew{\optTk}$, we run \\$\tcstrategysimulate{\optTk(\lambda)}$ and binary search over $\lambda$ until the expected reward output by the simulation is approximately zero. More precisely, we search for $\lambda$ such that the expected values of the upper bound $\UBD{\round, \coin}(\lambda)$ and the lower bound $\LBD{\round, \coin}(\lambda)$ are slightly larger and slightly smaller than zero. However, the binary search could potentially output any $\lambda$ such that $\LinLamRew{\lambda}{\optTk(\lambda)} \approx 0$ even if $\lambda$ is far off from $\Rew{\optTk}$. Such an error will become all the more likely given that the simulations only produce an interval $\big[\E_{s \sim \UBD{\round, \coin}(\lambda_1)}[s], \E_{s \sim \LBD{\round, \coin}(\lambda_1)}[s] \big]$ such that $\LinLamRew{\lambda}{\optTk(\lambda)}$ lies in this interval, instead of exactly computing the rewards. The following theorem rules out such scenarios.

\begin{theorem} \label{thm:BinSearchAPX}
    Let $\tclamstrategysimulate{\lambda_1}$ output the upper bound and lower bound distributions $\UBD{\round, \coin}(\lambda_1)$ and $\LBD{\round, \coin}(\lambda_1)$ respectively, such that $\E_{s \sim \UBD{\round, \coin}(\lambda_1)}[s] - \E_{s \sim \LBD{\round, \coin}(\lambda_1)}[s] \leq \delta$. Suppose for some $r \in \Big[\E_{s \sim \UBD{\round, \coin}(\lambda_1)}[s], \E_{s \sim \LBD{\round, \coin}(\lambda_1)}[s] \Big]$,
    $|r-\LinLamRew{\lambda_2}{\optTk(\lambda_2)}| \leq \zeta\text{.}$
    Then, $|\lambda_1 - \lambda_2| \leq \zeta + \delta$ with probability at least $1 - \Big(2 \, \round \, \chernoff + \round \, \tfrac{e^{-\strat \samples}}{\strat} \, \sqrt{\tfrac{\ln \chernoff^{-1}}{2\samples}} \Big)$.
\end{theorem}
We defer the proof to \Cref{sec:ProofOfBinSearch}. By locating the optimal adversarial reward (approximately), we also uncover a very succinct description of the adversary's (near) optimal strategy, which we describe in \Cref{sec:AdvStrategy}.

\section{Simulation Results} \label{sec:SimPlots}

Below, we summarize the results from our simulations. 
In \Cref{fig:BoundComp}, we compare the bounds from previous works against our results. We see that our bounds on the adversarial rewards of both the omniscient adversary (blue) and the actual adversary in $\CSSPA$ (green) is significantly tighter than the bound for the omniscient adversary from \citet{FHWY22} (orange). The bounds plotted in \Cref{fig:BoundComp} is empirical, and are not provably correct, since we did not inflate samples when constructing the reward distributions. However, they are fairly representative of the scale of the marginal rewards the adversary achieves from being strategic. For instance, even with an extremely large stake of $0.2$, we get the marginal rewards to be in the range $[0.0068, 0.0078]$ (for a simulation with deflated and inflated sampling), which is not considerable. Further, observe that the rewards from the $1$-lookahead strategy (red), which is much more tractable than the optimal strategy to describe and compute, is already close to the optimal adversarial reward. In \Cref{fig:AlphaRange}, we plot the marginal rewards of the adversary against its stake for various values of $\beta$. Observe that the network connectivity $\beta$ plays an important role in the rewards and the adversary has significant gains for larger values of $\beta$. For instance, at $\alpha = 0.2$ the marginal utility when $\beta = 1$ is at least $0.0068$ (from a simulation constructed from deflated sampling) and at most $0.0021$ (from a simulation constructed from inflated sampling) when $\beta = 0$. Finally, \Cref{fig:BetaRange} compares the adversary's marginal rewards as a function of $\beta$ for a stake $\alpha = 0.25$, once again, highlighting the role of connectivity in strategic manipulation in blockchain protocols.

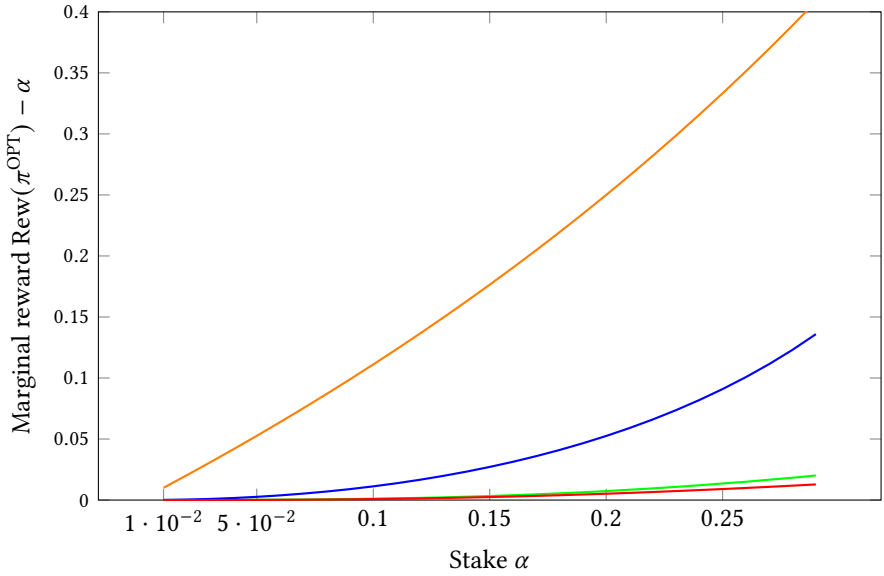
\begin{figure}[h]
    \centering
\begin{tikzpicture}
\begin{axis}[
    xlabel={Stake $\alpha$},
    ylabel={Marginal reward $\Rew{\opt}-\alpha$},
    legend style={at={(0.5,-0.15)}, anchor=north, legend columns=-1},
    width=12cm,
    height=8cm,
    scaled y ticks=false,
    yticklabel style={/pgf/number format/fixed, /pgf/number format/precision=7, font=\footnotesize},
    ytick={0, 0.05, 0.1, 0.15, 0.2, 0.25, 0.3, 0.35, 0.4},
    xtick={0.01, 0.05, 0.1, 0.15, 0.2, 0.25},
    ymin=0,
    ymax=0.4,
]

\addplot[
    color=green,
    thick
    ]
    coordinates {
        (0.01, 0.00007897713) (0.02, 0.00008827107) (0.03, 0.0001752715649) (0.04, 0.0001676452968)
        (0.05, 0.0002671508817) (0.06, 0.00032864042) (0.07, 0.00045785787) (0.08, 0.00062070801)
        (0.09, 0.0008144293) (0.1, 0.0011224448) (0.11, 0.0014291196) (0.12, 0.0017796834)
        (0.13, 0.0022155453) (0.14, 0.0027446558) (0.15, 0.0033048884) (0.16, 0.0039770307)
        (0.17, 0.0046970561) (0.18, 0.0055053986) (0.19, 0.006383743) (0.2, 0.0073420575)
        (0.21, 0.0084461967) (0.22, 0.0095549904) (0.23, 0.0107875803) (0.24, 0.0121093058)
        (0.25, 0.013584692) (0.26, 0.0149942352) (0.27, 0.0166813375) (0.28, 0.0182794178)
        (0.29, 0.0201351471)
    };

\addplot[
    color=blue,
    thick
    ]
    coordinates {
        (0.01, 0.0001010212071) (0.02, 0.0004083278602) (0.03, 0.0009286749167) (0.04, 0.001669356119)
        (0.05, 0.002638243461) (0.06, 0.003843832635) (0.07, 0.005295296435) (0.08, 0.007002547284)
        (0.09, 0.008976310376) (0.1, 0.01122820925) (0.11, 0.01377086609) (0.12, 0.01661801959)
        (0.13, 0.01978466399) (0.14, 0.02328721391) (0.15, 0.02714370072) (0.16, 0.03137400807)
        (0.17, 0.03600015644) (0.18, 0.04104664944) (0.19, 0.04654089938) (0.2, 0.05251375499)
        (0.21, 0.05900016344) (0.22, 0.06604001055) (0.23, 0.07367920213) (0.24, 0.08197107656)
        (0.25, 0.09097828243) (0.26, 0.1007753232) (0.27, 0.1114520836) (0.28, 0.1231188462)
        (0.29, 0.1359136444)
    };

\addplot[
    color=orange,
    thick
    ]
    coordinates {
        (0.01, 0.0101010101) (0.02, 0.02040816327) (0.03, 0.03092783505) (0.04, 0.04166666667)
        (0.05, 0.05263157895) (0.06, 0.06382978723) (0.07, 0.0752688172) (0.08, 0.08695652174)
        (0.09, 0.0989010989) (0.1, 0.1111111111) (0.11, 0.1235955056) (0.12, 0.1363636364)
        (0.13, 0.1494252874) (0.14, 0.1627906977) (0.15, 0.1764705882) (0.16, 0.1904761905)
        (0.17, 0.2048192771) (0.18, 0.2195121951) (0.19, 0.2345679012) (0.2, 0.25)
        (0.21, 0.2658227848) (0.22, 0.2820512821) (0.23, 0.2987012987) (0.24, 0.3157894737)
        (0.25, 0.3333333333) (0.26, 0.3513513514) (0.27, 0.3698630137) (0.28, 0.3888888889)
        (0.29, 0.4084507042)
    };

\addplot[
    color=red,
    thick
    ]
    coordinates {
        (0.01, 9.80E-07) (0.02, 7.68E-06) (0.03, 2.54E-05) (0.04, 5.90E-05)
        (0.05, 0.0001128398949) (0.06, 0.0001909371702) (0.07, 0.0002968558079) (0.08, 0.000433779405)
        (0.09, 0.0006045174963) (0.1, 0.00081152213) (0.11, 0.001056904354) (0.12, 0.001342450511)
        (0.13, 0.001669638255) (0.14, 0.002039652217) (0.15, 0.002453399275) (0.16, 0.002911523365)
        (0.17, 0.003414419814) (0.18, 0.003962249164) (0.19, 0.004554950478) (0.2, 0.005192254095)
        (0.21, 0.005873693864) (0.22, 0.006598618826) (0.23, 0.00736620436) (0.24, 0.008175462804)
        (0.25, 0.009025253555) (0.26, 0.009914292655) (0.27, 0.01084116189) (0.28, 0.01180431738)
        (0.29, 0.01280209777)
    };

\end{axis}
\end{tikzpicture}
    \caption{Marginal reward vs adversarial stake. Legend: orange--upper bound from \citealp{FHWY22}; blue-- tight upper bound for the omniscient adversary; green-- un-inflated simulated upper bound for $\beta = 1$; red-- reward from the $1$-lookahead strategy in \citealp{FHWY22}.}
    \label{fig:BoundComp}
\end{figure}

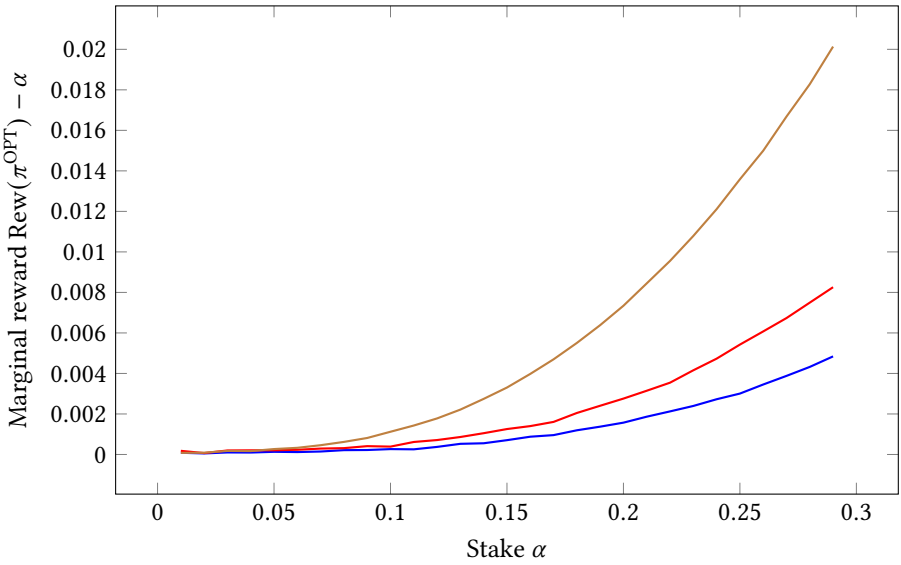
\begin{figure}[h]
    \centering
    \begin{tikzpicture}
    \begin{axis}[
        xlabel={Stake $\alpha$},
        ylabel={Marginal reward $\Rew{\opt}-\alpha$},
        legend style={at={(0.5,-0.15)}, anchor=north, legend columns=-1},
        width=12cm,
        height=8cm,
        scaled y ticks=false,
        yticklabel style={/pgf/number format/fixed, /pgf/number format/precision=7},
        ytick={0, 0.002, 0.004, 0.006, 0.008, 0.01, 0.012, 0.014, 0.016, 0.018, 0.02},
        scaled x ticks=false,
        xticklabel style={/pgf/number format/fixed, /pgf/number format/precision=7},
        xtick={0, 0.05, 0.1, 0.15, 0.2, 0.25, 0.3}
    ]

    \addplot[
        color=blue,
        thick
        ]
        coordinates {
            (0.01, 0.00008419501) (0.02, 0.00005181442) (0.03, 0.0001002968) (0.04, 0.00009583248) 
            (0.05, 0.00013142909) (0.06, 0.00011889049) (0.07, 0.00014662485) (0.08, 0.00021367417) 
            (0.09, 0.00022086477) (0.1, 0.0002655091) (0.11, 0.000254117) (0.12, 0.0003740958) 
            (0.13, 0.0005259905) (0.14, 0.000553322) (0.15, 0.0007065732) (0.16, 0.000877303) 
            (0.17, 0.0009560486) (0.18, 0.0011940648) (0.19, 0.0013731628) (0.2, 0.0015711942) 
            (0.21, 0.0018654795) (0.22, 0.0021266046) (0.23, 0.0023953144) (0.24, 0.0027265613) 
            (0.25, 0.0030064592) (0.26, 0.0034539642) (0.27, 0.0038794858) (0.28, 0.0043231806) 
            (0.29, 0.0048422916)
        };

    \addplot[
        color=red,
        thick
        ]
        coordinates {
            (0.01, 0.00017981546) (0.02, 0.00006858395) (0.03, 0.00019819211) (0.04, 0.00020497246) 
            (0.05, 0.00020898409) (0.06, 0.00023128619) (0.07, 0.00029196718) (0.08, 0.00031279048) 
            (0.09, 0.00041135283) (0.1, 0.0003884035) (0.11, 0.0006171874) (0.12, 0.0007121927) 
            (0.13, 0.0008629341) (0.14, 0.0010505406) (0.15, 0.0012553034) (0.16, 0.0014008845) 
            (0.17, 0.0016081481) (0.18, 0.00204855) (0.19, 0.0024053253) (0.2, 0.0027587233) 
            (0.21, 0.0031465013) (0.22, 0.0035441349) (0.23, 0.0041557134) (0.24, 0.0047302622) 
            (0.25, 0.005424438) (0.26, 0.0060739323) (0.27, 0.0067346216) (0.28, 0.0074965401) 
            (0.29, 0.0082572688)
        };

    \addplot[
        color=brown,
        thick
        ]
        coordinates {
            (0.01, 0.00007897713) (0.02, 0.00008827107) (0.03, 0.0001752715649) (0.04, 0.0001676452968) 
            (0.05, 0.0002671508817) (0.06, 0.00032864042) (0.07, 0.00045785787) (0.08, 0.00062070801) 
            (0.09, 0.0008144293) (0.1, 0.0011224448) (0.11, 0.0014291196) (0.12, 0.0017796834) 
            (0.13, 0.0022155453) (0.14, 0.0027446558) (0.15, 0.0033048884) (0.16, 0.0039770307) 
            (0.17, 0.0046970561) (0.18, 0.0055053986) (0.19, 0.006383743) (0.2, 0.0073420575) 
            (0.21, 0.0084461967) (0.22, 0.0095549904) (0.23, 0.0107875803) (0.24, 0.0121093058) 
            (0.25, 0.013584692) (0.26, 0.0149942352) (0.27, 0.0166813375) (0.28, 0.0182794178) 
            (0.29, 0.0201351471)
        };

    \end{axis}
    \end{tikzpicture}
    \caption{Marginal reward vs adversarial stake. Legend: brown-- un-inflated simulated upper bounds for $\beta = 1$; red-- un-inflated simulated upper bounds for $\beta = 0.5$; blue-- un-inflated simulated upper bounds for $\beta = 0$.}
    \label{fig:AlphaRange}
\end{figure}

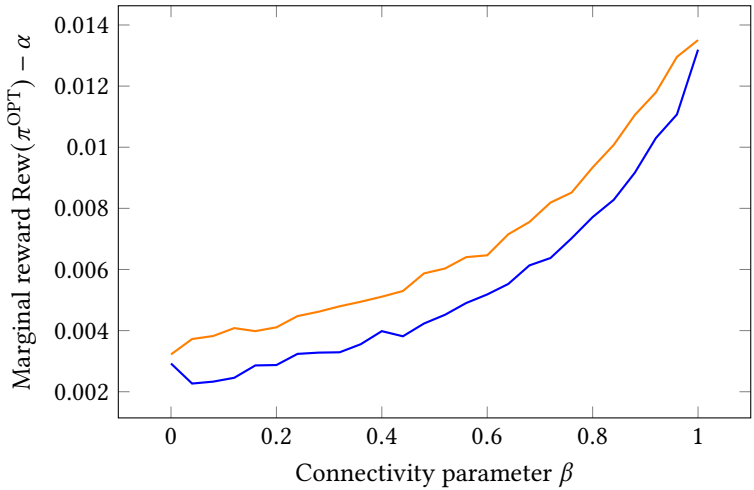
\begin{figure}
\centering
\begin{tikzpicture}
\begin{axis}[
    xlabel={Connectivity parameter $\beta$},
    ylabel={Marginal reward $\Rew{\opt}-\alpha$},
    legend style={at={(0.5,-0.15)}, anchor=north, legend columns=-1},
    width=10cm,
    height=7cm,
    ytick={0, 0.002, 0.004, 0.006, 0.008, 0.01, 0.012, 0.014},
    scaled y ticks=false,
    yticklabel style={/pgf/number format/fixed, /pgf/number format/precision=3},
    yticklabels={0, 0.002, 0.004, 0.006, 0.008, 0.01, 0.012, 0.014}
]

\addplot[
    color=orange,
    thick
    ]
    coordinates {
        (0, 0.0032207012) (0.04, 0.0037244252) (0.08, 0.0038230371) (0.12, 0.0040811953) (0.16, 0.0039825645) 
        (0.2, 0.0041072692) (0.24, 0.0044743816) (0.28, 0.0046179552) (0.32, 0.0047952462) (0.36, 0.0049439587) 
        (0.4, 0.0051094458) (0.44, 0.005295178) (0.48, 0.0058709972) (0.52, 0.0060326259) (0.56, 0.0064028961) 
        (0.6, 0.0064643485) (0.64, 0.0071520172) (0.68, 0.0075532467) (0.72, 0.0081886361) (0.76, 0.0085160112) 
        (0.8, 0.0093385273) (0.84, 0.0100759752) (0.88, 0.0110541723) (0.92, 0.0117926977) (0.96, 0.0129569557) 
        (1, 0.0135052933)
    };

\addplot[
    color=blue,
    thick
    ]
    coordinates {
        (0, 0.0029234149) (0.04, 0.0022675023) (0.08, 0.0023304179) (0.12, 0.0024563372) (0.16, 0.0028608255) 
        (0.2, 0.0028732275) (0.24, 0.0032383423) (0.28, 0.0032795805) (0.32, 0.0032907035) (0.36, 0.0035548058) 
        (0.4, 0.0039823102) (0.44, 0.00381407) (0.48, 0.0042316791) (0.52, 0.004518937) (0.56, 0.0049020085) 
        (0.6, 0.0051832394) (0.64, 0.0055266489) (0.68, 0.0061335121) (0.72, 0.0063745533) (0.76, 0.0070187423) 
        (0.8, 0.0077139142) (0.84, 0.0082787147) (0.88, 0.0091647299) (0.92, 0.0102998497) (0.96, 0.0110733687) 
        (1, 0.0131923059)
    };

\end{axis}

\end{tikzpicture}
    \caption{Marginal reward vs network connectivity. Legend: orange-- un-inflated simulated upper bound for $\alpha = 0.25$, blue-- un-deflated simulated lower bound for $\alpha = 0.25$}
    \label{fig:BetaRange}
\end{figure}

\bibliographystyle{apalike}
\bibliography{MasterBib}

\appendix

\section{Characterizing Balanced Scoring Functions} \label{sec:ScoringRule}
We begin the discussion on balanced scoring functions by proving \Cref{thm:CanScoringRules}.

\begin{proof}[Proof of \Cref{thm:CanScoringRules}]
    Let $\alpha = \sum_{i = 1}^n \alpha_i$, $U$ be the random variable denoting $S(X, \alpha)$. Let $V_1^i, \dots, V_{\Gamma}^i$ be $\Gamma$ independent copies of $S(X_i, \alpha_i)$ and $V_j$ denote the random variable $\min \{V_j^1, \dots, V_j^n\}$.

    From the definition of a balanced scoring function,
    $$Pr\Big(\min_{1 \leq j \leq \Gamma, 1 \leq i \leq n} \{S(X, \alpha), S(X_i^j, \alpha_i)\} = S(X,\alpha)\Big) = \frac{\alpha}{(\Gamma + 1) \, \alpha} = \frac{1}{(\Gamma + 1)}$$
    Here, we use $X_i^j$ to denote the random variable copy of $X_i$ corresponding to $V_j$ (i.e, the copy that determines $V_j^i$). Rephrasing,
    \begin{equation}
    \notag
        \begin{split}
            \frac{1}{\Gamma+1} &= Pr\Big(U \leq V_j^1, \dots, V_j^n \text{ for all } 1 \leq j \leq \Gamma\Big) \\
            &= Pr\Big(U \leq V_1, \dots, V_\Gamma\Big) \\
            &= Pr\Big(U \leq V_1, \dots, V_\Gamma\Big)
        \end{split}
    \end{equation}
    We move to the quantile space of the i.i.d.~variables $V_1, \dots, V_\Gamma$. Let $F$ be the $\CDF$ of $V_i$ and let the quantile $q_i$ be the random variable given by $V_i = F^{-1}(1-q_i)$. Thus, the random variables $q_1, \dots, q_n$ satisfy $q_i = 1-F(V_i)$. Note that $q_i \sim U[0, 1]$. Let $q_0$ be the random variable denoting $1 - F(U)$, distributed according to the distribution $G$ with a density function $g$. Since $F(\cdot)$ is an increasing function, $1-F(\cdot)$ is decreasing, and,
    \begin{equation}
        \notag
        \begin{split}
            \frac{1}{\Gamma+1} &= Pr(U \leq V_1, \dots, V_n) \\
            &= Pr(q_0 \geq q_1, \dots, q_n) \\
            &= \int_{0}^1 \Big(\Pi_{i=1}^\Gamma Pr(q_i \leq q_0)\Big) g(q_0) \,d q_0 \\
            &= \int_0^1 q_0^\Gamma g(q_0) \,dq_0 \\
            &= \E_{q_0 \sim G}[q_0^\Gamma]
        \end{split}
    \end{equation}
    The third equality follows since $q_i \sim U[0, 1]$ for $1 \leq i \leq \Gamma$. From the above, we know that the $\Gamma\textsuperscript{th}$ moment of the distribution $G$ equals $\frac{1}{\Gamma+1}$ which is the $\Gamma\textsuperscript{th}$ moment of $U[0, 1]$. $U[0, 1]$ has a well-defined moment generating function and thus, any distribution with the same moments as the uniform distribution is the uniform distribution.

    Thus, $q_0 = 1-F(U) \sim U[0, 1]$, the same as $1-F(V_i)$ and thus, $U$ is distributed identically to $V_1, \dots, V_{\Gamma}$. In other words, $S(X, \sum_{i = 1}^n \alpha_i) \text{ and } \min_{1 \leq i \leq n} \{S(X_i, \alpha_i)\}$ are identically distributed.
\end{proof}

\noindent We characterize all ``different'' balanced scoring rules in \Cref{thm:UniqueScore}.
\begin{theorem} \label{thm:UniqueScore}
    Let $S$ be a balanced scoring function. Then, the distribution of $S(X, 1)$ for $X \sim U[0, 1]$ uniquely determines the distribution of $S(X, \alpha)$ for all $\alpha \geq 0$. In particular, $Pr(S(X, \alpha) \geq s) = Pr(S(X, 1) \geq s)^{\alpha}$.
\end{theorem}

To prove \Cref{thm:UniqueScore}, we require the following intuitive corollary of \Cref{thm:CanScoringRules}: for a user with stake $\alpha$, we expect the user's probability of getting a large score (and, thus not get elected) to decrease with an increase in the stake.
\begin{corollary} \label{thm:IncStake}
    Let $S$ be a balanced scoring function. Then, $Pr(S(X, \alpha) \geq s)$ is monotonically non-increasing in $\alpha$.
\end{corollary}
\begin{proof}
    For $\alpha, \epsilon \geq 0$,
    \begin{equation}
        \notag
        \begin{split}
            Pr(S(X, \alpha + \epsilon) \geq s) &= Pr\Big(\min \{S(X_1, \alpha), S(X_2, \epsilon)\} \geq s\Big) \\
            &\leq Pr(S(X, \alpha) \geq s)
        \end{split}
    \end{equation}
\end{proof}

The following lemma characterizes the set of all "different" balanced scoring functions.

\begin{proof}[Proof of \Cref{thm:UniqueScore}]
    We show the claim in three steps.
    \begin{enumerate}
        \item $\alpha \in \N \cup \{0\}$. Then, by \Cref{thm:CanScoringRules}, $$Pr(S(X, \alpha) \geq s) = Pr\big(\min_{1 \leq i \leq \alpha} \{S(X_i, 1)\} \geq s\big) = Pr(S(X, 1) \geq s)^{\alpha}$$
        \item $\alpha \in \Q_{\geq 0}$. Let $\alpha = \tfrac{p}{q}$ for natural numbers $p$ and $q (\neq 0)$. Then, by $\Cref{thm:CanScoringRules}$,
        $$Pr(S(X, p) \geq s) = Pr\big(\min_{1 \leq i \leq q} \{S(X_i, \frac{p}{q}) \} \geq s\big) = Pr(S(X, \frac{p}{q}) \geq s)^q$$
        From the previous case, $Pr(S(X, p) \geq s) = Pr(S(X, 1) \geq s)^p$ and thus, $Pr(S(X, \frac{p}{q}) \geq s) = Pr(S(X, 1) \geq s)^{\frac{p}{q}}$.
        \item $\alpha \in \R_{\geq 0}$. Let $\big(\alpha_i\big)_{i \in \N}$ be a sequence of rational numbers converging to $\alpha$. Then,
        $$\lim_{i \xrightarrow{} \infty} Pr(S(X, \alpha_i) \geq s) = \lim_{i \xrightarrow{} \infty} Pr(S(X, 1) \geq s)^{\alpha_i} = Pr(S(X, 1) \geq 1)^{\alpha}$$
        If we show that $Pr(S(X, \cdot) \geq s)$ is continuous, then we can conclude $Pr(S(X, \alpha) \geq s) = Pr(S(X, 1) \geq s)^{\alpha}$. By \Cref{thm:IncStake}, $Pr(S(X, a))$ is monotone non-increasing in $a$. Thus, for $\alpha_1, \alpha_2 \in \Q$ and $\alpha \in [\alpha_1, \alpha_2]$,
        $$Pr(S(X, 1) \geq s)^{\alpha_1} = Pr(S(X, \alpha_1) \geq s) \geq Pr(S(X, \alpha) \geq s) \geq Pr(S(X, \alpha_2) \geq s) = Pr(S(X, 1) \geq s)^{\alpha_2}$$
        For the sequences of rationals $\big(\alpha_1^i\big)_{i \in \N}$ such that $\alpha_1^i < \alpha$ and $\big(\alpha_2^i\big)_{i \in \N}$ satisfying $\alpha_2^i > \alpha$, both converging to $\alpha$
        $$Pr(S(X, 1) \geq s)^{\alpha_1^i} \geq Pr(S(X, \alpha) \geq s) \geq Pr(S(X, 1) \geq s)^{\alpha_2^i}$$
        As $i \xrightarrow{} \infty$, both $Pr(S(X, 1) \geq s)^{\alpha_1^i}$ and $Pr(S(X, 1) \geq s)^{\alpha_2^i}$ converge to $Pr(S(X, 1) \geq s)^\alpha$. Thus, $Pr(S(X, \alpha) \geq s) = Pr(S(X, 1) \geq s)^\alpha$.
    \end{enumerate}
\end{proof}

\citet{FHWY22} conjectured that all balanced scoring functions $S(X, \alpha)$ are continuous in $\alpha$. Even though $Pr(S(X, \alpha) \geq s)$ is continuous in $\alpha$ (\Cref{thm:UniqueScore}), $S(X, \alpha)$ need not be continuous in $\alpha$.

\begin{example}
    Consider
    \begin{equation}
        \notag
        S(X, \alpha) = 
        \begin{cases}
            \tfrac{- \ln X}{\alpha} & \text{if } \tfrac{- \ln X}{\alpha} \leq 3\\
            \tfrac{- \ln X}{\alpha}+1 & \text{if } \tfrac{- \ln X}{\alpha} > 3
        \end{cases}
    \end{equation}
    It can be easily verified that $S(X, \alpha)$ is indeed a balanced scoring function. However, $S(X, \alpha)$ is not continuous in $\alpha$.
\end{example}

\section{A More General $\CSSPA$} \label{sec:GenCSSPA}
In this section, we review the general version of $\CSSPA$ defined in \citet{FHWY22}.

\begin{definition}[$\ParCSSPA{\vec{\alpha}}$] \label{def:Strategy}
    In $\ParCSSPA{\vec{\alpha}}$, the network consists of the adversary with stake $\alpha$ and honest players with stakes given by $\vec{\alpha}$. Prior to round $1$, the adversary learns the values of $\alpha, \vec{\alpha}$ and $\beta$ and that the network apart from the adversary is honest. For a choice $n \geq 1$, the adversary distributes its stake arbitrarily over a set $\advAccount$ of $n$ wallets. The adversary makes the following decisions in round $t$:
    \begin{enumerate}
        \item The adversary learns the seed $\seed{t}$ of round $t$.
        \item The adversary computes the credentials $\cred{i}{t}$ for all wallets $i \in \advAccount$. The adversary chooses a subset of honest players $\knownHonest$ with total stake at most $\beta \, (1-\alpha)$ and learns the credentials $\cred{i}{t}$ for all $i \in \knownHonest$. For all wallets $i \in \opaqueHonest$, the adversary knows that $\cred{i}{t}$ will be drawn independently from $\expd{\alpha_i}$.
        \item For any $r \geq 0$ and $(i_t, i_{t+1}, \dots, i_{t+r}) \in \big(\advAccount \cup \knownHonest \big) \times \advAccount^r$, the adversary precomputes the credentials $\cred{i_{t+r'}}{t+r'}$ for $1 \leq r' \leq r$ assuming $i_{t+\hat{r}}$ is elected to lead in round $t+\hat{r}$ for all $0 \leq \hat{r} < r'$.
        \item The adversary either remains silent or broadcasts the credentials of some subset $\advAccount_t$ of the adversarial wallets $\advAccount$.
    \end{enumerate}
\end{definition}

From the above definition of $\ParCSSPA{\vec{\alpha}}$, \citet{FHWY22} make a series of refinements that lead to $\CSSPA$ without compromising on the adversarial reward.

\begin{lemma}[Observation 3.2 from \citealp{FHWY22}] \label{thm:TwoHonest}
    For any $\alpha, \vec{\alpha}, \beta$, define $\vec{\alpha}'$ to have two honest players with stakes $\alpha_1 = \beta \, (1-\alpha)$ and $\alpha_2 = (1-\beta) \, (1-\alpha)$ respectively. For any strategy $\strategy$ in $\ParCSSPA{\vec{\alpha}}$, there exists a strategy $\strategy'$ in $\ParCSSPA{(\alpha_1, \alpha_2)}$ such that $\Rew{\alpha, (\alpha_1, \alpha_2), \beta; \strategy'} = \Rew{\alpha, \vec{\alpha}, \beta; \strategy}$.
\end{lemma}

\begin{lemma}[Lemma 3.1 from \citealp{FHWY22}] \label{thm:Split}
    Let $\strategy$ be a strategy in $\ParCSSPA{(\alpha_1, \alpha_2)}$ where the adversary splits its stake into $n$ wallets. Then there exists a strategy $\strategy'$ such that the adversary divides its stake into $2n$ wallets and $\Rew{\strategy'} \geq \Rew{\strategy}$.
\end{lemma}

\begin{lemma}[Observation 3.1 from \citealp{FHWY22}] \label{thm:Broadcast1}
    For any strategy $\strategy$ in $\ParCSSPA{(\alpha_1, \alpha_2)}$ that distributes the adversarial stake across $n$ wallets, there exists an adversarial strategy $\strategy'$ that also distributes the stake across the same number of wallets, broadcasts the credential of at most one wallet in $\advAccount$ and results in exactly the same leaders as $\strategy$, thereby getting the same reward as $\strategy$.
\end{lemma}

The following conclusion is straightforward from the above lemmas.

\begin{theorem} \label{thm:RefeqUnref}
    The optimal adversarial reward in $\ParCSSPA{\vec{\alpha}}$ is at most the optimal adversarial reward in $\CSSPA$ for all $\alpha, \vec{\alpha}$ and $\beta$.
\end{theorem}

\section{The Omniscient Adversary} \label{sec:Omni}
\citet{FHWY22} abstract the game faced by an omniscient adversary into a Galton-Watson branching process. The branching process maintains a tree with nodes corresponding to possible choices the adversary can make across rounds. A node $\cred{i_t}{t}$ in level $t$ of the tree corresponds to a potential winner $i_t$ in round $t$ for some sequence of leaders $i_1, i_2, \dots, i_{t-1}$ with credentials $\cred{i_1}{1}, \cred{i_2}{2}, \dots, \cred{i_{t-1}}{t-1}$ in the first $t-1$ rounds. A node $\cred{i_{t+1}}{t+1}$ in level $t+1$ is a child of $\cred{i_t}{t}$ if $i_{t+1}$ belongs to the set of potential winners $\potWinners_{t+1}(\seed{t+1})$ assuming $i_t$ gets elected in round $t$ and $\seed{t+1} = \cred{i_t}{t}$. By \Cref{thm:PrPotentialWinner}, the probability that a node will have exactly $i^* + 1$ children equals $\alpha^{i^*} \, (1-\alpha)$. If node $\cred{i_t}{t}$ at level $t$ has exactly one child $\cred{i_{t+1}}{t+1}$ in the choice tree, this corresponds to the adversarial potential winners $\advPot_{t+1}(\seed{t+1})$ being empty for $\seed{t+1} = \cred{i_t}{t}$, which signifies a stopping time. The tree stops branching at $\cred{i_{t}}{t}$.

The height $\fst$ of the choice tree corresponds to the sequence of leaders $i_1, \dots, i_{\fst}$ that causes the maximum delay in the first stopping time $\fst$. We show upper bounds on the expected height $\fst$ of the choice tree.

\begin{definition}[Omniscient Choice Tree] \label{def:OmniChoiceTree}
The omniscient choice tree $\omniChoiceTree$ is built by the following stochastic process:
\begin{enumerate}
    \item Level $0$ contains the root $q_0$ of the tree $\omniChoiceTree$.
    \item For each node $q$ at level $t \geq 0$, $q$ has $i^* + 1$ children with probability $\alpha^{i^*} \, (1-\alpha)$ for all $i^* \geq 0$.
    \item A node stops branching if it is the only child of its parent.
    \item $\omniChoiceTree$ becomes extinct at height $\fst$ when all nodes at level $\fst$ have stopped branching.
\end{enumerate}
    
\end{definition}

\subsection{Extinction of the Choice Tree}
Let $\event{t}$ be the event that the choice tree does not terminate on or before round $t$, i.e, $\fst > t$, and $\nevent{t}$ be its complement. Let $\prevent{t}$ and $\prnevent{t} = 1 - \prevent{t}$ be the probabilities of $\event{t}$ and $\nevent{t}$ respectively. Indeed $\prevent{0} = 1$ and $\prnevent{0} = 0$. We compute a recursive relation between $\prevent{t-1}$ and $\prevent{t}$.

\begin{lemma} \label{thm:Recursive}
    The probabilities $\big(\prevent{t} \big)_{t \in \N \cup \{0\}}$ satisfy $\prevent{0} = 1$ and
    $$\prevent{t} = \frac{\alpha(2-\alpha)\prevent{t-1} - \alpha(1-\alpha)\prevent{t-1}^2}{(1-\alpha) + \alpha \prevent{t-1}}$$
    for all $t \geq 1$.
\end{lemma}
\begin{proof}
    We establish a recursion on $\big(\prnevent{t} \big)_{t \in \N \cup \{0\}}$ instead. Let $\potWinners_q$ be a random variable denoting the number of children of a node $q \in \omniChoiceTree$ and $\h_q$ be the height of the sub-tree below $q$.

    $\nevent{t}$ is the event that the choice tree goes extinct before or as soon as reaching a height $t$ from the root $q_0$. This occurs when either $|\potWinners_{q_0}| = 1$ and the game ends in the first round or every child $q_j$ of $q_0$ satisfies $\h_{q_j} \leq t-1$ for $1 \leq j \leq |W_{q_0}|$. Hence,
    \begin{equation} \label{eqn:RecPr}
        \prnevent{t} = Pr(|W_{q_0}| = 1) + \sum_{i^* = 1}^{\infty} Pr(|W_{q_0}| = i^* + 1) \times Pr(\bigcup_{1 \leq j \leq i^* + 1} (\h_{q_j} \leq t-1) | |W_{q_0}| = i^*+1)
    \end{equation}
    The evolution of the choice tree under distinct nodes are independent and identically distributed random processes. Therefore,
    \begin{equation*}
        \notag
        \begin{split}
            Pr(\bigcup_{1 \leq j \leq i^* + 1} (\h_{q_j} \leq t-1) | |W_{q_0}| = i^*+1) &= Pr(\h_q \leq t-1)^{i^* + 1} \\
            &= \prnevent{t-1}^{i^* + 1}
        \end{split}
    \end{equation*}
    The probability that $|W_{q_0}|$ equals $i^* + 1$ equals $\alpha^{i^*} \, (1-\alpha)$ (by \Cref{def:OmniChoiceTree}).

    Plugging all of the above into \Cref{eqn:RecPr},
    \begin{equation*}
        \notag
        \begin{split}
            \prnevent{t} &= (1-\alpha) + \sum_{i^* = 1}^{\infty} \alpha^{i^*} \, (1-\alpha) \times \prnevent{t-1}^{i^* + 1} \\
            &= (1-\alpha) + \frac{\alpha (1-\alpha) \prnevent{t-1}^2}{1 - \alpha \prnevent{t-1}}
        \end{split}
    \end{equation*}
    Substituting $\prnevent{t} = 1 - \prevent{t}$ and $\prnevent{t-1} = 1 - \prevent{t-1}$ completes the proof.
\end{proof}
 \noindent \Cref{thm:PrStoppingTime} is a direct consequence of \Cref{thm:Recursive}. 
 \begin{proof}[Proof of \cref{thm:PrStoppingTime}]
     We know $\prevent{0} = 1$. $\prevent{1}$ and $\prevent{2}$ can be computed recursively to be equal to $\alpha$ and $\alpha^2 \cdot \tfrac{2-2\alpha +\alpha^2}{1 - \alpha + \alpha^2}$ respectively.
     
     We induct on $t$ for $t \geq 3$. Let $\prevent{t} \leq \alpha^2 \cdot \tfrac{2 - 2\alpha + \alpha^2}{1-\alpha+\alpha^2} \cdot \big(\alpha \cdot  \tfrac{2-\alpha}{1-\alpha}\big)^{t-2}$. Then,
     \begin{equation*}
         \notag
         \begin{split}
             \prevent{t+1} &= \frac{\alpha(2-\alpha)\prevent{t} - \alpha(1-\alpha)\prevent{t}^2}{(1-\alpha) + \alpha \prevent{t}} \\
            &\leq  \big( \alpha \cdot \tfrac{2-\alpha}{1-\alpha} \big) \prevent{t} \\
            &\leq \big( \alpha \cdot \tfrac{2-\alpha}{1-\alpha} \big) \times \alpha^2 \cdot \tfrac{2 - 2\alpha + \alpha^2}{1-\alpha+\alpha^2} \cdot \big(\alpha \cdot  \tfrac{2-\alpha}{1-\alpha}\big)^{t-2} \\
            &= \alpha^2 \cdot \tfrac{2 - 2\alpha + \alpha^2}{1-\alpha+\alpha^2} \cdot \big(\alpha \cdot  \tfrac{2-\alpha}{1-\alpha}\big)^{t-1}
         \end{split}
     \end{equation*}
 \end{proof}

\subsection{Proof of \Cref{thm:OmniSummary}}

Let $\omniOPT$ denote the optimal strategy for the omniscient adversary that broadcasts the credentials of the wallets in the longest path in $\omniChoiceTree$ starting at the root $q_0$.

\begin{lemma} \label{thm:NoUB}
    There exists a constant $\omniconst \approx 0.38$ such that expected first stopping time $\E[\fst]$ of $\omniOPT$ is unbounded for all $\alpha > \omniconst$. The optimal reward $\RewOmni{\omniOPT}$ for the omniscient adversary equals $1$.
\end{lemma}
\begin{proof}
    Denote the smaller root of $1 - 3x + x^2$ by $\omniconst \approx 0.38$. Choose an arbitrary $\delta \in \big( 0, -\tfrac{1- 3\alpha + \alpha^2}{\alpha \, (2-\alpha)} \big)$. Since the larger root of $1-3x+x^2$ is greater than $1$, the interval is non-degenerate for all $\omniconst \leq \alpha \leq 1$.
    
    Suppose we show $\prevent{t} \geq \delta$ for all $t \geq 0$. Then,
    $$\E[\fst] = \sum_{t = 0}^{\infty} \prevent{t} \geq \sum_{t = 0}^{\infty} \delta$$
    which is unbounded, as needed.
    
    We will prove $\prevent{t} \geq \delta$ by induction on $t$. Indeed, $\prevent{0} = 1 \geq \delta$. Suppose that $\prevent{t-1} \geq \delta$.
    \begin{equation*}
        \notag
        \begin{split}
            \prevent{t} &= \frac{\alpha(2-\alpha)\prevent{t-1} - \alpha(1-\alpha)\prevent{t-1}^2}{(1-\alpha) + \alpha \prevent{t-1}} \\
            &\geq \frac{\alpha(2-\alpha)\delta - \alpha(1-\alpha)\delta^2}{(1-\alpha) + \alpha \delta} \\
            &= \delta \cdot \Big(\frac{\alpha (2-\alpha) - \alpha (1-\alpha) \delta}{(1-\alpha) + \alpha \delta} \Big) \\
            & \geq \delta
        \end{split}
    \end{equation*}
    The inequality in the second line holds since $\prevent{t}$ is monotonously increasing in $\prevent{t-1}$. For the choice of $\delta$, $\Big(\frac{\alpha (2-\alpha) - \alpha (1-\alpha) \delta}{(1-\alpha) + \alpha \delta} \Big) \geq 1$ (follows by simple rearrangement) and the last inequality follows.

    Since $\E[\fst]$ is unbounded,
    $$\RewOmni{\omniOPT} = 1 - \frac{1}{\E[\fst]} = 1$$
\end{proof}

\begin{lemma} \label{thm:SecondOrderBound}
    Let $\omniconst \approx 0.38$ be the smaller root of $1-3x+x^2$. For a stake $\alpha \leq \omniconst$, the expected first stopping time $\E[\fst]$ of the optimal adversarial strategy $\omniOPT$ is at most $\tfrac{1- 3\alpha + 3\alpha^2 - 3 \alpha^3}{(1-3\alpha+\alpha^2) \, (1-\alpha+\alpha^2)}$. The optimal adversarial reward $\RewOmni{\omniOPT}$ is at most $\alpha \cdot \big( \tfrac{1 - 2\alpha + \alpha^2 - \alpha^3}{1- 3\alpha + 3\alpha^2 - 3 \alpha^3} \big)$.
\end{lemma}
\begin{proof}
    From \Cref{thm:PrStoppingTime}, $\prevent{0} = 1$, $\prevent{1} = \alpha$ and $\prevent{t} \leq \alpha^2 \cdot \tfrac{2 - 2\alpha + \alpha^2}{1-\alpha+\alpha^2} \cdot \big(\alpha \cdot  \tfrac{2-\alpha}{1-\alpha}\big)^{t-2}$.

    \begin{equation*}
        \notag
        \begin{split}
            \E[\fst] &= \sum_{t = 0}^{\infty} \prevent{t} \\
            &\leq 1 + \alpha + \sum_{t = 2}^{\infty} \alpha^2 \cdot \tfrac{2 - 2\alpha + \alpha^2}{1-\alpha+\alpha^2} \cdot \big(\alpha \cdot  \tfrac{2-\alpha}{1-\alpha}\big)^{t-2} \\
            &= \tfrac{1- 3\alpha + 3\alpha^2 - 3 \alpha^3}{(1-3\alpha+\alpha^2) \, (1-\alpha+\alpha^2)}
        \end{split}
    \end{equation*}
    Further,
    $$\RewOmni{\omniOPT} = 1 - \frac{1}{\E[\fst]} \leq \alpha \cdot \big( \tfrac{1 - 2\alpha + \alpha^2 - \alpha^3}{1- 3\alpha + 3\alpha^2 - 3 \alpha^3} \big)$$
\end{proof}
\Cref{thm:NoUB} and \Cref{thm:SecondOrderBound} conclude the proof of \Cref{thm:OmniSummary}.

\begin{remark}
    \citet{FHWY22} show that $\E[\fst] \leq \tfrac{1-\alpha}{1-3\alpha +\alpha^2}$ and $\RewOmni{\omniOPT} \leq \alpha \cdot \tfrac{2-\alpha}{1-\alpha}$. The bounds for both $\E[\fst]$ and $\RewOmni{\omniOPT}$ in \Cref{thm:SecondOrderBound} are tighter. See \Cref{fig:Omni} for a comparison between the two bounds.
\end{remark}

\Cref{thm:Recursive} gives an explicit method to compute $\prevent{t}$ for all $t \in \N \cup \{0\}$ and thus, $\E[\fst]$ and $\RewOmni{\omniOPT}$. However, we believe $\prevent{t}$ does not admit a closed form solution. Instead, for a sufficiently large $\round_{\delta}$ such that $\E[\fst] - \sum_{t = 0}^{\round_{\delta}} \prevent{t} < \delta$, we compute $\big(\prevent{t} \big)_{0 \leq t \leq \round_{\delta}}$ and $\sum_{t = 0}^{\round_{\delta}} \prevent{t}$ to get a tight bound on $\E[\fst]$ up to an additive error $\delta$.

\begin{lemma}
    For $\alpha < 0.38$, $\delta = 10^{-7}$ and $\round_{\delta} = 3000$,
    $$\E[\fst] - \sum_{t = 0}^{\round_{\delta}} \prevent{t} < \delta \text{.}$$
\end{lemma}
\begin{proof}
    \begin{equation*}
        \notag
        \begin{split}
            \E[\fst] - \sum_{t = 0}^{\round_{\delta}} \prevent{t} &= \sum_{t = \round_{\delta} + 1}^{\infty} \prevent{t} \\
            &\leq \sum_{t = \round_{\delta} + 1}^{\infty} \alpha^2 \cdot \tfrac{2 - 2\alpha + \alpha^2}{1-\alpha+\alpha^2} \cdot \big(\alpha \cdot  \tfrac{2-\alpha}{1-\alpha}\big)^{t-2} \\
            &= \alpha^2 \cdot \tfrac{2 - 2\alpha + \alpha^2}{1-\alpha+\alpha^2} \cdot \tfrac{1-\alpha}{1-3\alpha + \alpha^2} \cdot \big(\alpha \cdot \tfrac{2-\alpha}{1-\alpha}\big)^{\round_{\delta} -1}
        \end{split}
    \end{equation*}
    The inequality follows from \Cref{thm:PrStoppingTime}. We will verify that $\alpha^2 \cdot \tfrac{2 - 2\alpha + \alpha^2}{1-\alpha+\alpha^2} \cdot \tfrac{1-\alpha}{1-3\alpha + \alpha^2} \cdot \big(\alpha \cdot \tfrac{2-\alpha}{1-\alpha}\big)^{\round_{\delta} -1} < 10^{-7}$ for $\alpha = 0.38$. Since the above expression is monotonically increasing in $\alpha$, we have $\E[\fst] - \sum_{t = 0}^{\round_{\delta}} \prevent{t} < 10^{-7}$ for all $\alpha \leq 0.38$. Substituting $\round_{\delta} = 3000$, we have $\big(\alpha \cdot \tfrac{2-\alpha}{1-\alpha}\big)^{\round_{\delta} - 1} < 10^{-9}$ and $\alpha^2 \cdot \tfrac{2 - 2\alpha + \alpha^2}{1-\alpha+\alpha^2} \cdot \tfrac{1-\alpha}{1-3\alpha + \alpha^2} < 37$. Combining the two inequalities, we get the required result.
\end{proof}

The optimal omniscient reward satisfies $$\RewOmni{\omniOPT} = 1 - \frac{1}{\E[\fst]} \leq 1 - \frac{1}{\sum_{t=0}^{3000} \prevent{t} + 10^{-7}}$$
We plot the non-closed form bound thus obtained in \Cref{fig:Omni}.

\begin{figure}[ht]
    \centering
    \includegraphics[width = 0.8\textwidth]{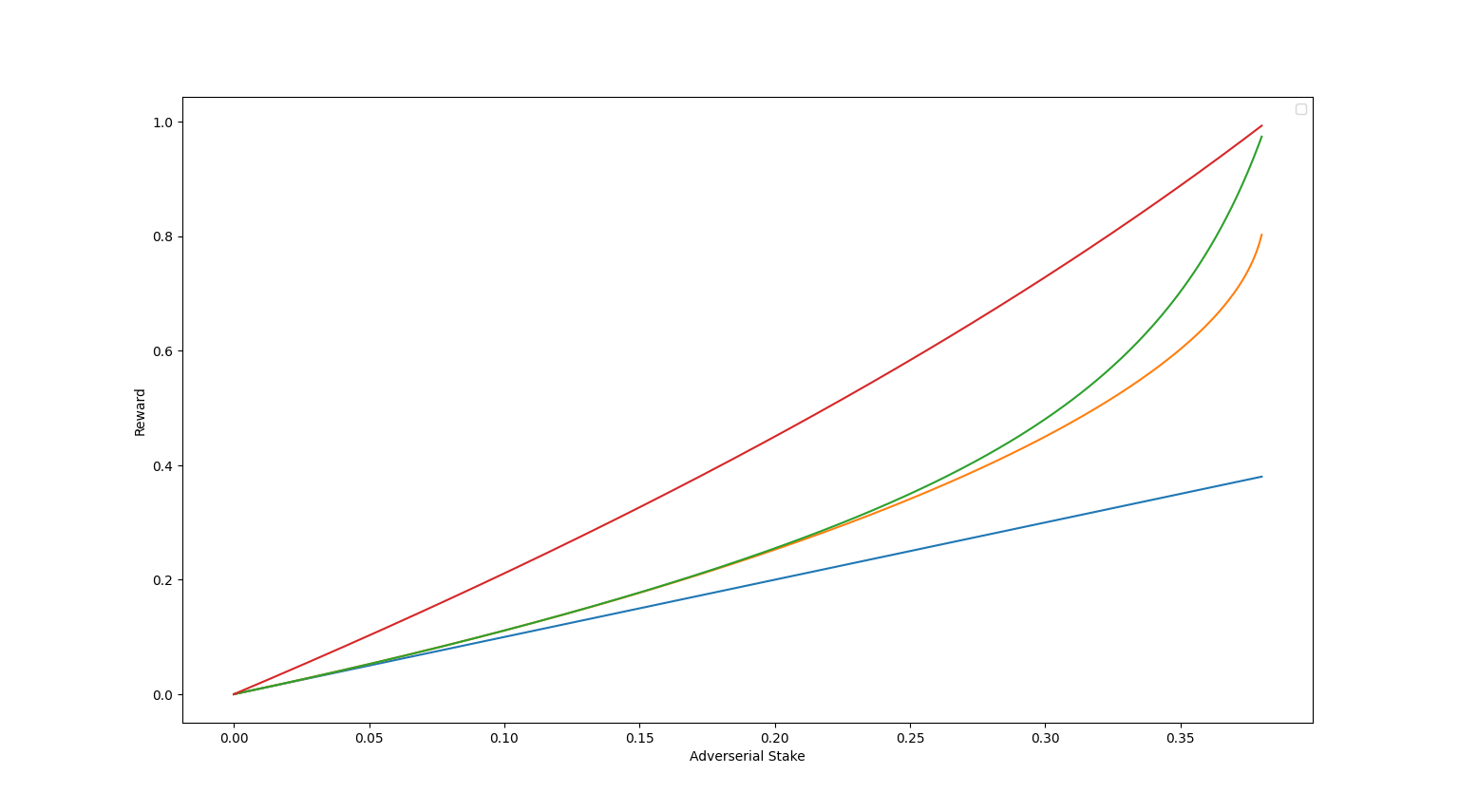}
    \caption{Omniscient adversarial reward vs adversarial stake. The blue line maps the reward of the honest strategy while the red curve maps the upper bound on the omniscient adversarial reward from \citet{FHWY22}. The yellow and the green curve are the non-closed form upper bound and the upper bound in \Cref{thm:OmniSummary} respectively. The non-closed form upper bound is tight up to an additive error of $10^{-7}$.}
    \label{fig:Omni}
\end{figure}

\subsection{A Restricted Omniscient Adversary}

The omniscient adversary discussed in \Cref{sec:OmniMain} is rewarded solely for delaying the first stopping time and might be elected in only a small fraction of rounds despite the game having a large first stopping time. In this section, we consider a restricted omniscient adversary that can precompute future credentials of $\knownHonest$, but is not rewarded for rounds in which $\knownHonest$ is elected the leader ($X_t(\strategy) = 1$ only if the leader $i$ of round $t$ is a wallet in $\advAccount$).

\begin{theorem} \label{thm:RestOmni}
    There exists a constant $\restrictedOmniConst$ such that $\omniconst \approx 0.38 \leq \restrictedOmniConst \leq \frac{1}{2}$ and a strategy $\strategy$ for the restricted omniscient adversary with stake $\alpha > \restrictedOmniConst$ such that $\Rew{\strategy} = 1$.
\end{theorem}

To prove \Cref{thm:RestOmni}, we consider an omniscient adversary with $\alpha \geq \frac{1}{2}$ and much more constraints than the restricted omniscient adversary that can achieve a reward $1$. The restricted omniscient adversary that gets a better reward than its constrained counterpart will also get a reward $1$ when $\alpha \geq \frac{1}{2}$.

We consider a greedy omniscient adversary that always broadcasts a credential whenever the set of adversarial potential winners is non-empty. The greedy omniscient adversary never lets $\knownHonest$ to be elected the leader before the first stopping time $\fst$. Therefore, for a strategy $\pi$, $X_t(\strategy) = 1$ for all $1 \leq t < \fst$ and $\Rew{\strategy} = 1 - \tfrac{1}{\E[\fst]}$.

Since $\knownHonest$ is never elected the leader before the first stopping time, the node corresponding to $\cred{\knownHonest}{t}$ in the omniscient choice tree $\omniChoiceTree$ does not branch and have children. Apart from not branching at $\cred{\knownHonest}{t}$, the branching process governing $\omniChoiceTree$ is identical to the omniscient adversary.

\begin{definition}[Greedy Omniscient Choice Tree] \label{def:GreedyOmniChoiceTree}
The greedy omniscient choice tree $\omniChoiceTree$ is built by the following stochastic process:
\begin{enumerate}
    \item Level $0$ contains the root $q_0$ of the tree $\omniChoiceTree$.
    \item For each node $q$ not of the form $\cred{\knownHonest}{t}$ at level $t \geq 0$, $q$ has $i^* + 1$ children with probability $\alpha^{i^*} \, (1-\alpha)$ for all $i^* \geq 0$.
    \item A node stops branching if it is of the form $\cred{\knownHonest}{t}$.
    \item $\omniChoiceTree$ becomes extinct at height $\fst$ when all nodes at level $\fst$ have stopped branching.
\end{enumerate}   
\end{definition}
Observe that if a node $q$ is the only child of its parent $\cred{i}{t}$, then the adversarial potential winner set with $\seed{t+1} = \cred{i}{t}$ is empty and therefore, $q$ must be of the form $\cred{\knownHonest}{t+1}$. The extinction condition requiring a node to stop branching if it is the only child of its parent is implicit from bullet 3 of \Cref{def:GreedyOmniChoiceTree}. 

Let $\event{t}'$ be the event that $\omniChoiceTree$ does not become extinct on or before stage $t$ and let $\prevent{t}'$ be its probability.

\begin{lemma} \label{thm:GreedyRec}
    The probabilities $\big(\prevent{t}' \big)_{t \in \N \cup \{0\}}$ satisfy $\prevent{0}' = 1$ and
    $$\prevent{t}' = \frac{\alpha \prevent{t-1}'}{(1-\alpha) + \alpha \prevent{t-1}'}$$
    for all $t \geq 1$.
\end{lemma}

\begin{lemma} \label{thm:GreeduNoUB}
    For stake $\alpha > \frac{1}{2}$, the expected first stopping time $\E[\fst]$ of optimal greedy omniscient strategy $\omniOPT$ is unbounded. The adversary receives a reward $1$ by playing $\omniOPT$.
\end{lemma}
The proofs are almost identical to their counterparts \Cref{thm:Recursive} and \Cref{thm:NoUB} of the omniscient adversary and we skip repeating them here.

 \begin{proof}[Proof of \Cref{thm:RestOmni}]
     The restricted omniscient adversary is weaker than the omniscient adversary. For stake $\alpha < \omniconst$, the omniscient adversary receives a reward $< 1$ and the restricted omniscient adversary cannot do better.

     For $\alpha > \frac{1}{2}$, the greedy omniscient adversary receives a reward $1$ and the restricted omniscient adversary being stronger, will also get a reward $1$.

     Thus, the infimum over all $\alpha$ such that the restricted omniscient adversary receives a reward $1$ lies in the range $[\omniconst, \frac{1}{2}]$.
 \end{proof}

\section{Proof of Theorem \ref{thm:Signlambda}} \label{sec:ProofofSignLambda}

For a strategy $\strategy$ in $\CSSPA$, the expected adversarial reward in the linear version of $\CSSPA$ equals
\begin{equation*}
    \notag
    \begin{split}
        \E[\sum_{t= 1}^{\fst} X_t(\strategy)] - \lambda \cdot \E[\fst] &= \Big( \frac{\E[\sum_{t= 1}^{\fst} X_t(\strategy)]}{\E[\fst]} - \lambda \Big) \cdot \E[\fst] \\
        &= \big(\Rew{\strategy} - \lambda \big) \cdot \E[\fst]
    \end{split}
\end{equation*}
where $\Rew{\strategy}$ is the expected adversarial reward from playing $\strategy$ in $\CSSPA$. If $\lambda > \Rew{\strategy}$ ($\lambda < \Rew{\strategy}$), the adversary earns a negative (positive) reward in the linear version of $\CSSPA$. Importantly, if $\lambda = \Rew{\strategy}$, the expected adversarial reward in the linear $\CSSPA$ is exactly zero.

\section{Proof of Theorem \ref{thm:(Tk)CSSPA}} \label{sec:Proofof(Tk)CSSPA}

We prove the theorem in two steps. Let $\FinCSSPA{\round}$ be the variant of the $\CSSPA$ with a $\coin$-scored adversary that terminates after $\round$ rounds. We first bound the difference between the optimal reward in the $\CSSPA$ and the optimal reward in $\InfCoinCSSPA{\round}$ followed by bounding the difference between the optimal rewards in $\InfCoinCSSPA{\round}$ and $\FinCSSPA{T}$.

Remember that $\optrew{\round, \coin}$ is the optimal strategy of the $\coin$-scored adversary in $\FinCSSPA{\round}$ and $\opt$ is the fully functional adversary's strategy in $\CSSPA$.

\subsection{Loss in Reward from Terminating after $\round$ Rounds} \label{sec:LossTRounds}
In this section, we bound the optimal reward of an $\infty$-scored adversary in $\InfCoinCSSPA{T}$ and the optimal CSSPA reward. For convenience, we will abuse notation and write $\E[\Rew{\cdot}]$ for $\Rew{\cdot}$ in \Cref{sec:LossTRounds} and \Cref{sec:LosskScore}.

\begin{lemma} \label{thm:T-CSSPA}
For $\alpha \leq \omniconst \approx 0.38$,
$$0 \leq \E[\Rew{\opt} - \Rew{\optTI}] \leq \alpha^2 \cdot \tfrac{2 - 2\alpha + \alpha^2}{1 - \alpha + \alpha^2} \cdot [\alpha \cdot \tfrac{2-\alpha}{1-\alpha}]^{T-2}$$    
\end{lemma}
\begin{proof}
    The left hand side is easy to see. $\optTI$ is a valid strategy in $\CSSPA$. The optimal strategy $\opt$ is only going to yield a better reward.

    We proceed to the right hand side. As usual, let $\fst$ denote the first stopping time of $\opt$ in $\CSSPA$. Let $\RewT{\strategy}$ be the reward of playing a strategy till time $\round$ and terminating.

    \begin{equation}
        \notag
        \begin{split}
            \E[\Rew{\opt}] &= Pr(\fst \leq \round) \, \E[\Rew{\opt} | \fst \leq \round] + Pr(\fst > \round) \, \E[\Rew{\opt} | \fst > \round] \\
            &= Pr(\fst \leq \round) \, \E[\RewT{\opt} | \fst \leq \round] + Pr(\fst > \round) \times \\
            & \hspace{1.5cm}\E[\RewT{\opt} + \{\Rew{\opt} - \RewT{\opt}\} | \fst > \round] \\
            &= \E[\RewT{\opt}] + Pr(\fst > \round) \, \E[\{\Rew{\opt} - \RewT{\opt}\} | \fst > \round]\\
            &\leq \E[\RewT{\optTI}] + Pr(\fst > \round) \\
            &= \E[\Rew{\optTI}] + Pr(\fst > T)
        \end{split}
    \end{equation}

    Conditioned on the stopping time being at most $\round$, there is no difference between playing $\opt$ and playing $\opt$ only till time $\round$ and terminating. Thus, $\Rew{\opt}$ can be equated to $\RewT{\opt}$ in the first equation. Observe that $\RewT{\optTI}$ yields a better reward than any strategy that recommends termination after $\round$ rounds. In particular, $\E[\RewT{\optrew{T, \infty}}] \geq \E[\RewT{\opt}]$. Further, $0 \leq \Rew{\opt}, \RewT{\opt} \leq 1$ and therefore, $\E[\{\Rew{\opt} - \RewT{\opt}\} | \fst > \round] \leq 1$. Combining the above, we get the first inequality. The last equality follows since $\optTI$ anyways recommends terminating before round $\round$ and therefore $\RewT{\optrew{T, \infty}} = \Rew{\optrew{T, \infty}}$.
    
    The lemma follows by bounding $Pr(\fst > \round)$ by $\alpha^2 \cdot \frac{2 - 2\alpha + \alpha^2}{1 - \alpha + \alpha^2} \cdot [\alpha \cdot \frac{2-\alpha}{1-\alpha}]^{T-2}$ from \Cref{thm:PrStoppingTime}.
\end{proof}

\subsection{Loss in Reward for a $k$-Scored Adversary} \label{sec:LosskScore}

We bound the difference between the optimal rewards in $\InfCoinCSSPA{\round}$ and $\FinCSSPA{\round}$.

\begin{lemma} \label{thm:TkCSSPA}
    For $\alpha \leq 0.29$,
    $$0 \leq \E[\Rew{\optTI}] - \E[\Rew{\optTk}] \leq \alpha^\coin$$
\end{lemma}

We first show that the probability of an adversary ever having $\coin$ adversarial potential winners is at most $\alpha^\coin$. Once this is established, the proof is similar to the bound on the rewards in $\CSSPA$ and $\InfCoinCSSPA{\round}$. 

\begin{lemma} \label{thm:TkCSSPAProb}
    Let $\alpha \leq 0.29$. The probability that the adversary never encounters a round with at least $\coin$ adversarial potential winners before the first stopping time in $\CSSPA$ is at least $1 - \alpha^{\coin}$.
\end{lemma}
\begin{proof}
Consider an omniscient adversary whose sole purpose is to encounter a round with at least $\coin+1$ adversarial potential winners. Let $\IndEvent$ be the event that the adversary never encounters such a round and $\prindevent = Pr(\IndEvent)$. Clearly, the probability of the real adversary never encountering such a round is only larger.

Consider a round $t$ where $\advPot$ is the random set of adversarial potential winners. $\IndEvent$ happens precisely when
\begin{enumerate}
    \item the current round is a stopping time $(\fst = t)$, or
    \item $|\advPot| \leq \coin$ and none of the potential winners (including $\knownHonest$) lead to a future round with at least $\coin+1$ adversarial potential winners.
\end{enumerate}
Thus,
$$Pr(\IndEvent) = Pr(\fst = t) + \sum_{i^* = 1}^{\coin} Pr(|\advPot| = i^*) \times Pr(\IndEvent)^{i^* + 1}$$
$Pr(\fst = t) = 1-\alpha$ and $Pr(|\advPot| = i^*) = \alpha^{i^*} \cdot (1-\alpha)$ (by \Cref{thm:PrPotentialWinner}). Plugging them in,
$$\prindevent = (1-\alpha) + \sum_{i^* = 1}^{\coin} (1-\alpha) \cdot \alpha^{i^*} \prindevent^{i^*+1}$$
Let $\nu = \alpha \, \prindevent$. The above can be rewritten as
    \begin{equation} \label{eqn:kRec}
        \begin{split}
            \nu &= \alpha(1-\alpha) + (1-\alpha) \nu^2 \sum_{i = 0}^{\coin-1} \nu^i \\
            &= \alpha(1-\alpha) + (1-\alpha) \nu^2 \frac{1-\nu^{\coin}}{1-\nu}
        \end{split}
    \end{equation}
    The goal is to locate the root of the above equation. We show that the root $\nu = \alpha \, \prindevent$ is at least $\alpha - \alpha^{\coin+1}$ in \Cref{sec:Root}. Thus, $\prindevent \geq 1 - \alpha^k$, establishing the claim.
\end{proof}

\begin{proof}[Proof of \Cref{thm:TkCSSPA}]
        Similar to \Cref{thm:T-CSSPA}, the left hand side of the inequality is straightforward. The $\infty$-scored adversary's reward is only larger than that of the $\coin$-scored adversary. Thus, $0 \leq \E[\Rew{\optTI}] - \E[\Rew{\optTk}]$.

        Let $\IndEvent$ be the event that the adversary does not see more than $\coin$ adversarial potential winners till the $\fst$ in $\InfCoinCSSPA{\round}$, and let $\overline{\IndEvent}$ be its complement. Let $\optrew{\round, \leq \coin}$ be the strategy that copies $\optTI$ until the adversary encounters $k+1$ adversarial potential winners and then terminates. Observe that conditioned on $\IndEvent$, $\optTI$ and $\optrew{\round, \leq \coin}$ are identical. Thus,

        \begin{equation*}
            \notag
            \begin{split}
                \E[\Rew{\optTI}] &= Pr(\IndEvent) \E[\Rew{\optTI}| \IndEvent] + Pr(\overline{\IndEvent}) \E[\Rew{\optTI}| \overline{\IndEvent}] \\
                &= Pr(\IndEvent) \E[\Rew{\optrew{\round, \leq \coin}}| E] + Pr(\overline{\IndEvent}) \E[\Rew{\optrew{\round, \leq \coin}}) + (\Rew{\optTI} - \Rew{\optrew{\round, \leq \coin}}| \overline{\IndEvent}] \\
                &= \E[\Rew{\optrew{\round, \leq \coin}}] + Pr(\overline{\IndEvent}) \E[(\Rew{\optTI} - \Rew{\optrew{\round, \leq \coin}}| \overline{\IndEvent}]
            \end{split}
        \end{equation*}

        $\E[\Rew{\optrew{\round, \leq \coin}}]$ is only smaller than the $\coin$-scored adversary's optimal reward $\E[\Rew{\optTk}]$. Further, \\$0 \leq \Rew{\optTI}, \Rew{\optrew{\round, \leq \coin}} \leq 1$. Hence,
        \begin{equation*}
            \notag
            \begin{split}
                \E[\Rew{\optTI}] &= \E[\Rew{\optrew{\round, \leq \coin}}] + Pr(\overline{\IndEvent}) \E[(\Rew{\optTI} - \Rew{\optrew{\round, \leq \coin}})| \overline{\IndEvent}] \\
                &\leq \E[\Rew{\optTk}] + Pr(\overline{\IndEvent}) \\
                &\leq \E[\Rew{\optTk}] + \alpha^\coin
            \end{split}
        \end{equation*}
        The last inequality follows from \Cref{thm:TkCSSPAProb}.
        
\end{proof}
        
\begin{proof}[Proof of \Cref{thm:(Tk)CSSPA}]
    For $\alpha \leq 0.29$, we have
    $$0 \leq \Rew{\opt} - \Rew{\optTI} \leq \alpha^2 \cdot \frac{2 - 2\alpha + \alpha^2}{1 - \alpha + \alpha^2} \cdot [\alpha \cdot \frac{2-\alpha}{1-\alpha}]^{\round-2}$$
    and $$0 \leq \E[\Rew{\optTI}] - \E[\Rew{\optTk}] \leq \alpha^\coin$$
    The theorem follows by adding the two inequalities.
\end{proof}

\subsection{Locating the Root of \Cref{eqn:kRec}} \label{sec:Root}
In this section, we will approximate the root of the equation
    \begin{equation}
        \notag
        \begin{split}
            \nu &= \alpha(1-\alpha) + (1-\alpha) \nu^2 \sum_{i = 0}^{\coin-1} \nu^i \\
            &= \alpha(1-\alpha) + (1-\alpha) \nu^2 \frac{1-\nu^{\coin}}{1-\nu}
        \end{split}
    \end{equation}
The left hand side is the equation of a straight line. The right hand side of the first equality is clearly convex in $\nu$. A straight line can intersect a convex curve at most twice. Thus, the above equation can have at most two roots.

At $\nu = 0$, the LHS is zero and thus, is smaller than the RHS, equal to $\alpha(1-\alpha)$. At $\nu = \infty$, the LHS is smaller than the RHS since the RHS is a degree $\coin+1$ polynomial with a positive leading coefficient while the LHS is just a line. At $\nu = \alpha$, the LHS equals $\alpha$. The RHS equals
$$\alpha(1-\alpha) + (1-\alpha) \alpha^2 \frac{1-\alpha^{\coin}}{1-\alpha} \leq \alpha(1-\alpha) + \alpha^2 = \alpha$$
and thus is smaller than the LHS. Thus, one of the roots of the equation is in $[0, \alpha]$ and the other is in $[\alpha, \infty]$. Remember $\nu = \alpha \, \prindevent \leq \alpha \times 1$, since $\prindevent$ is the probability of $\IndEvent$ happening. Thus, we are interested in locating the smaller root of the above equation.

\begin{lemma} \label{thm:RootLocate}
    For $\alpha \leq 0.29$, the smaller root of
    \begin{equation} \label{eqn:TrueHardEqn}
        \nu = \alpha(1-\alpha) + (1-\alpha) \nu^2 \frac{1-\nu^{\coin}}{1-\nu}
    \end{equation}
    is at least $\alpha - \alpha^{\coin+1}$.
\end{lemma}
\begin{proof}
    Consider the equation
    \begin{equation} \label{eqn:AdjustedHardEqn}
        \nu = \alpha(1-\alpha) + (1-\alpha) \nu^2 \frac{1-\alpha^{\coin}}{1-\nu}
    \end{equation}
    instead. The right hand side is still convex in $[0, 1]$. In the region $[0, \alpha]$, the RHS of the above equation is strictly smaller than $\alpha(1-\alpha) + (1-\alpha) \nu^2 \frac{1-\nu^{\coin}}{1-\nu}$. Thus, at the root $\nu_1$ of \Cref{eqn:TrueHardEqn}
    $$\nu_1 > \alpha(1-\alpha) + (1-\alpha) \nu_1^2 \frac{1-\alpha_1^{\coin}}{1-\nu_1}$$
    The left hand side of the new equation has already overtaken the right hand side at $\nu = \nu_1$. The root of \Cref{eqn:AdjustedHardEqn}
    is smaller than $\nu_1$. We will find this root and use it to lower bound $\nu_1$.

    Next, at $\nu = \alpha - \alpha^{\coin+1}$, we will show that the RHS of the new equation is larger than the LHS. Since the LHS overtakes the RHS at $\nu = \alpha$, the root will have to lie in $[\alpha - \alpha^{\coin+1}, \alpha]$. For $\alpha \leq 0.29$, the following the chain of inequalities hold.
    \begin{equation}
        \notag
        \begin{split}
            0 &\geq -2\alpha^2 + 4\alpha - 1 \\
            (1-\alpha) &\geq \alpha^2 + 3(1-\alpha)\alpha \\
            (1-\alpha) \alpha^{\coin-1} &\geq \alpha^{\coin+1} + 3(1-\alpha)\alpha^{\coin} \\
            (1-\alpha) \alpha^{\coin-1} &\geq \alpha^{\coin+1} + (1-\alpha)\alpha^{\coin} [1 + (1-\alpha^{\coin}) + (1-\alpha^{\coin})^2] \\
            (1-\alpha + \alpha^{\coin+1}) \times \alpha^{\coin-1} &\geq (1-\alpha)(1-[1-\alpha^{\coin}]^3) + \alpha^{\coin+1} \\
            \alpha^{\coin+1} &\geq \alpha^2 \times [1 - (1-\alpha) \frac{(1-\alpha^{\coin})^3}{1-\alpha + \alpha^{\coin+1}}]
        \end{split}
    \end{equation}
    Showing $$\alpha - \alpha^{\coin+1} \leq \alpha - \alpha^2 \times [1 - (1-\alpha) \frac{(1-\alpha^{\coin})^3}{1-\alpha + \alpha^{\coin+1}}] = \alpha(1-\alpha) + (1-\alpha)[\alpha - \alpha^{\coin+1}]^2 \times \frac{1-\alpha^{\coin}}{1-\alpha + \alpha^{\coin+1}}$$
    is now straightforward, completing the proof.
\end{proof}

\section{Errors in Estimation through Finite Sampling} \label{sec:ErrorFinSampling}

Given a distribution $\D_0$, we want to approximate the sequence of distributions $\big( \D_t \big)_{0 \leq t \leq \round}$ such that $\D_{t+1} = \kaddl{\D_{t}}$.

\subsection{Bounding Expected Rewards through Chernoff Bound and the McDiarmid's Inequality} \label{sec:McD}

Note that we need an estimate of the $\CDF$ of each $\D_t$ and just getting an estimate of the expected value of $\D_t$ is not sufficient. While the Chernoff bound only gives an estimate of the expected reward of $\D_{t+1}$ given the distribution $\D_t$, using the McDiarmid's inequality to directly get tail bounds on the expected value of $\D_\round$ requires an exponential number of samples in $\round$ as discussed below.

As a thought experiment, we use $\samples$ samples to get an approximation $\apxd{1}$ of the distribution $\D_1$, and allow $\kaddl{\cdot}$ to construct infinitely many samples to approximate distributions from there on to get $\apxd{2}, \apxd{3}, \dots, \apxd{\round -1}, \apxd{\round}$. We sketch a `fake' proof using the McDiarmid's inequality to get an idea on the number of samples that would be recommended by McDiarmid's inequality if we want to bound the error in estimating the expected reward by (say) $0.001$ with probability at least $99\%$.

\begin{theorem}[McDiarmid's inequality, \citealp{McD89}] \label{thm:McD}
    Let $f: \mathcal{X}_1 \times \mathcal{X}_2 \times \dot \times \mathcal{X}_n \xrightarrow{} \R$ satisfy the $(c_1, \dots, c_n)$-bounded difference property for $\delta_1, \delta_2, \dots, \delta_n \geq 0$, i.e.,
    $$\sup_{x_i' \in \mathcal{X}_i} |f(x_1, x_2, \dots, x_i, \dots, x_n) - f(x_1, x_2, \dots, x_i', \dots, x_n)| \leq \delta_i$$
    for all $(x_1, \dots, x_n) \in \mathcal{X}_1 \times \mathcal{X}_2 \times \dot \times \mathcal{X}_n$ and for all $1 \leq i \leq n$. Then, for independent random variables $X_1 \in \mathcal{X}_1, X_2 \in \mathcal{X}_2, \dots, X_n \in \mathcal{X}_n$, and $\epsilon \geq 0$
    $$Pr(|f(X_1, \dots, X_n) - \E[f(X_1, \dots, X_n)]| \geq \epsilon) \leq 2e^{-\tfrac{2\epsilon^2}{\sum_{i = 1}^n \delta_i^2}}$$
\end{theorem}

The entire simulation can be thought of as a function of the samples $\big( s_{\ell} \big)_{1 \leq \ell \leq \samples}$ drawn to construct $\apxd{1}$. To use the McDiarmid's inequality We want to find a constant $\delta$ such that changing a single sample $s_{\ell}$ would not change the outcome of the simulation by more than $\delta$.

Suppose we `taint' a sample $s_{\ell}$ in the approximation $\apxd{1}$ (this amounts to tainting $\tfrac{1}{\samples}$ fraction of the samples constructed for $\apxd{1}$). For each of the infinitely many samples we construct to get $\apxd{2}$ from $\apxd{1}$, we draw a sample for each of the adversary's wallets with the $\coin$ smallest scores. Thus, $\tfrac{\coin}{\samples}$ fractions of the samples constructed for $\apxd{2}$ would have used $s_{\ell}$ and are tainted. This is ignoring the fact that we take an expectation with respect $r_0$, the reward from remaining silent; the number of tainted samples only gets worse if this gets accounted. Inductively, $\frac{\coin^{\round - 1}}{\samples}$ samples of $\apxd{\round}$ are tainted. All tainted samples in round $\round$ can potentially increase from $-\round \, \lambda$ to $\round \, (1-\lambda)$. Thus, the expected value of the distribution can increase by $\delta = \tfrac{\coin^{\round - 1}}{\samples} \, \round$.

For $\delta$ described above, $\round = 5$ and $\coin = 5$, and for an error in the estimated reward of at most $0.001$ with probability $99\%$, we would require $\samples = 0.75 \times 10^{11}$. The choices of $\round, \coin$ and the confidence probability are very mild and we need a huge number of samples just to account for the error from finitely sampling $\apxd{1}$. The number of samples would only grow if $\apxd{2}, \dots, \apxd{\round}$ are also constructed through finite sampling. Of course, the analysis is extremely loose and can be improved. The number of samples recommended by any such concentration bounds turn out to be extremely large.

\subsection{Bias} \label{sec:Bias}
We revisit the computation of the adversary's reward when its $\coin$ smallest scores are $\AdvC$ and the corresponding rewards are $\AdvR$. The reward from broadcasting the credential $i$ (or remaining silent) equals $e^{-c_i \, (1-\beta)\,(1-\alpha)} \, (r_i + \mathbbm{1}(i \neq 0))$ (assuming a stopping time is not reached). The adversary chooses the action that maximizes its reward.

Consider a toy version of the adversarial game over just two rounds. In the first round, a coin with heads probability $\frac{1}{\coin}$ is tossed $\samples$ times. Let the empirically observed probability of heads be $p_1$. For each trial in the second round, toss $\coin$ coins each with heads probability $p_1$. The outcome is $1$ even if one of $\coin$ tosses turns out heads (we take the maximum amongst $\coin$ Bernoulli trials, similar to the adversarial game). Repeat the trial $\samples$ times to observe an empirical probability $p_2$.

We verify that $p_2$ does not remain constant for $\coin = 2$ over different choices of $\samples$. If $\samples = \infty$, $p_1 = \tfrac{1}{\coin} = \tfrac{1}{2}$ and $p_2 = 1 - (1-\tfrac{1}{k})^2 = \tfrac{3}{4}$. However, for $\samples = 10$ samples, $p_1 = \tfrac{m}{\samples}$ with probability $\binom{\samples}{m} \cdot 2^{-\samples}$ and $\E[p_2] = \sum_{m = 0}^n \big[1 - \big(1 - \tfrac{m}{\samples}\big)^2\big] \times \binom{\samples}{m} \cdot 2^{-\samples} \approx 0.784 \neq 0.75$

A similar bias creeps into the adversarial game. The effect of the bias on the estimated rewards worsens with the number of rounds $\round$ and becomes better with the number of samples $\samples$.

\subsection{Proof of Theorem~\ref{thm:SampleSummary}} \label{sec:InflDefl}
We first show the following concentration inequality between the $\CDF$s of a distribution $D$ and a distribution constructed empirically by sampling from $D$.

\begin{theorem} \label{thm:InflDefl} Let $\D$ be some distribution supported on $[-\lambda, t \, (1-\lambda)]$ and let $\LBD{}$ and $\UBD{}$ be the result of $\Defl{\chernoff, \hat{\D}}$ and $\Infl{\chernoff, \strat, \hat{\D}}$ respectively, where $\hat{D}$ is constructed by sampling $\samples$ times from $\D$.
\begin{enumerate}
    \item $\LBD{}$ is dominated by $\D$ with probability at least $1-\chernoff$.
    \item $\UBD{}$ dominates $\D$ with probability at least $1 - \Big(\chernoff +  \tfrac{e^{-\strat \samples}}{\strat} \, \sqrt{\tfrac{\ln \chernoff^{-1}}{2\samples}} \Big)$
\end{enumerate}
\end{theorem}

For a distribution $\D$ and a quantile $\q$, let the inverse $\CDF$ $\D^{-1}(\q)$ denote the value corresponding to the quantile $\q$\footnote{Technically, the value corresponding to a quantile $q$ is given by $D^{-1}(1-q)$. For cleanliness, we pretend that the $\CDF$ $D(s)$ describes the probability $Pr_{S \sim D}(S \geq s)$ and thus, $D^{-1}(q)$ is the value associated to the quantile $q$}.

\begin{theorem}[DKW inequality; \citealp{DKW56, Mas90}] Let $\D$ be some distribution and let $\apxd{}$ be constructed by drawing $\samples$ samples from $\D$. Then,
\begin{equation}
    \notag
    Pr \big(\apxd{}^{-1}(\q + \sqrt{\tfrac{\ln \chernoff^{-1}}{2 \samples}}) \leq \D^{-1}(\q) \leq \apxd{}^{-1}(\q - \sqrt{\tfrac{\ln \chernoff^{-1}}{2 \samples}}) \text{ for all } \q \in [0,1] \big) \geq 1-\chernoff
\end{equation}
\end{theorem}

\begin{lemma} \label{thm:Defl}
    Let $\LBD{} = \Defl{\chernoff, \D}$ for some distribution $\D$ with infimum $\psi$. $\LBD{}$ is dominated by $\D$ with probability at least $1-\chernoff$.
\end{lemma}
\begin{proof}
    Consider the distribution $\apxd{}$ constructed from drawing $\samples$ samples from $\D$. By the DKW inequality, $\apxd{}^{-1}(\q+\sqrt{\tfrac{\ln \chernoff}{2\samples}}) \leq \D^{-1}(\q)$. By deleting the $\sqrt{\tfrac{\ln \chernoff}{2\samples}}$ smallest (strongest) quantiles of $\apxd{}$ and inserting the infimum $-\lambda$ of $\D$ to construct $\LBD{}$, we decrease the quantile of each sample by $\sqrt{\tfrac{\ln \chernoff}{2\samples}}$. Thus,
    $$\LBD{}^{-1}(\q) = \apxd{}^{-1}(\q + +\sqrt{\tfrac{\ln \chernoff}{2\samples}}) \leq \D^{-1}(\q)$$
    for all $\q \in [0, 1 - \sqrt{\tfrac{\ln \chernoff}{2\samples}}]$. For $\q \in [1 - \sqrt{\tfrac{\ln \chernoff}{2\samples}}, 1]$,
    $$\LBD{}^{-1}(\q) = \psi \leq \D^{-1}(\q)$$
\end{proof}

\begin{lemma} \label{thm:Infl}
    Let $\UBD{} = \Infl{\chernoff, \strat, \D}$ for some distribution $\D$ with supremum $\Psi$. $\UBD{}$ dominates $\D$ with probability at least $1-\Big(\chernoff + \tfrac{e^{-\strat \samples}}{\strat} \, \sqrt{\tfrac{\ln \chernoff^{-1}}{2\samples}} \Big)$.
\end{lemma}
\begin{proof}
    Let $\apxd{}$ be the distribution constructed by sampling $\samples$ times from $\D$. We show that, with very high probability, a sample with quantile $q \in (m \, \strat, (m+1) \, \strat]$ was drawn while constructing $\apxd{}$ for all $0 \leq m < \frac{1}{\strat} \cdot \sqrt{\tfrac{\ln \chernoff^{-1}}{2\samples}}$. Inserting $\strat \samples$ of each such sample with extremely small quantiles would ensure $\UBD{}$ dominates $\D$.

    Consider a single draw from $\D$. The probability that the outcome of the draw does not have a quantile $\q \in (m \, \strat, (m+1) \, \strat]$ equals $1-\strat$. The probability that none of the $\samples$ samples have a quantile in the range $(m \, \strat, (m+1) \, \strat]$ equals $(1-\strat)^{\samples} \leq e^{-\strat \, \samples}$. By a union bound, the probability that there does not exist any sample with quantile $\q \in (m \, \strat, (m+1) \, \strat]$ for some $0 \leq m < \frac{1}{\strat} \cdot \sqrt{\tfrac{\ln \chernoff^{-1}}{2\samples}}$ is at most $e^{-\strat \samples} \times \frac{1}{\strat} \cdot \sqrt{\tfrac{\ln \chernoff^{-1}}{2\samples}}$.

    Next, we show that $\UBD{}^{-1}(\q) \geq \D^{-1}(\q)$ for $q \in [0, \sqrt{\tfrac{\ln \chernoff^{-1}}{2\samples}}]$ with probability at least $1 - \tfrac{e^{-\strat \samples}}{\strat} \, \sqrt{\tfrac{\ln \chernoff^{-1}}{2\samples}}$. As argued above, with probability at least $1 - \tfrac{e^{-\strat \samples}}{\strat} \, \sqrt{\tfrac{\ln \chernoff^{-1}}{2\samples}}$, the $m\textsuperscript{th}$ largest sample $s_m$ of $\apxd{}$ has a quantile $\q \leq m \omega$ for $1 \leq m < \frac{1}{\strat} \cdot \sqrt{\tfrac{\ln \chernoff^{-1}}{2\samples}}$. We delete the $\samples \, \sqrt{\tfrac{\ln \chernoff^{-1}}{2\samples}}$ smallest samples and append $\strat \samples$ copies each of $s_0 := \Psi, s_1, \dots, s_{M}$ for $M = \Big(\frac{1}{\strat} \cdot \sqrt{\tfrac{\ln \chernoff^{-1}}{2\samples}} - 1\Big)$. For a quantile $\q \in (m \strat, (m+1) \, \strat]$  \big($0 \leq m \leq \frac{1}{\strat} \cdot \sqrt{\tfrac{\ln \chernoff^{-1}}{2\samples}}$\big), $\UBD{}^{-1}(\q) = s_{m} \geq \D^{-1}(\q)$.

    For $\q \in (\sqrt{\tfrac{\ln \chernoff^{-1}}{2\samples}}, 1]$, with probability at least $1-\gamma$, $\apxd{}^{-1}\big(\q - \sqrt{\tfrac{\ln \chernoff^{-1}}{2\samples}}\big) \geq \D^{-1}(\q)$ (by the DKW inequality). We appended $n \cdot \sqrt{\tfrac{\ln \chernoff^{-1}}{2\samples}}$ samples to $\apxd{}$, all of them with quantiles at most  $\sqrt{\tfrac{\ln \chernoff^{-1}}{2\samples}}$, to get $\UBD{}$. Thus, $\UBD{}^{-1}(\q) = \apxd{}^{-1}(\q - \sqrt{\tfrac{\ln \chernoff^{-1}}{2\samples}}) \geq \D^{-1}(\q)$.

    Combining the two, we get that the probability that $\UBD{}$ dominates $\D$ is at least $1-\Big(\chernoff + \tfrac{e^{-\strat \samples}}{\strat} \, \sqrt{\tfrac{\ln \chernoff^{-1}}{2\samples}} \Big)$.
\end{proof}

\noindent \Cref{thm:Defl} and \Cref{thm:Infl} prove \Cref{thm:InflDefl}.

Thus, to prove \Cref{thm:SampleSummary}, it is sufficient to show that the reward distributions are always within the range $[-\lambda, t \, (1-\lambda)]$.

\begin{proof}[Proof of \Cref{thm:SampleSummary}]
    Clearly, the maximum reward the adversary can earn in each round is $(1 - \lambda)$ and the supremum of $\dist{t, \coin}$ for each round $t$ is at most $t \, (1 - \lambda)$. Combining \Cref{thm:InflDefl} with a straightforward union bound between the $\round$ inflation procedures performed across the $\round$ rounds, we get the first part.

    For the second part, a similar proof to the first will follow provided we show the infimum of $\estdist{t, \coin}$ is $-\lambda$ for each $1 \leq t \leq \round$. We show this by sketching a strategy that guarantees a reward at least $-\lambda$ for the adversary. $\optTk$ being the optimal strategy will only provide the adversary with a better reward. Consider the strategy where the adversary keeps broadcasting the credential of a random adversarial potential winner until a stopping time is reached. The adversary wins $(1-\lambda)$ each round before the first stopping time and $-\lambda$ for the stopping time, since the adversary plays the round without getting elected. Since $\lambda < 1$, the adversary's reward is at least $-\lambda$. Note that the adversary guarantees itself a reward greater than $-\lambda$ even conditioned on the credentials of the current and future rounds. Thus, the optimal strategy $\optTk$ can adapt the above strategy if the adversary faces an extremely unfavourable seed and achieve a reward at least $-\lambda$.
    
\end{proof}

\section{Computing Expectations} \label{sec:PreComp}
We get into the details next. The first step towards pre-computing is chiselling the expression for the samples into a form that allows intensive reuse of values.
\begin{equation}
\notag
    \begin{split}
        s_\ell =\E_{c_0 \sim \expd{\beta \, (1-\alpha)}, r_0 \sim \estk{t-1}}\big[\max_{0 \leq i \leq i^*(\vecc)} \{e^{-c_i\cdot (1-\beta)\cdot(1-\alpha)} \, (r_i+ \mathbbm{1}(i \neq 0)) \} \cdot \mathbbm{1}(i^*(\vecc) \neq 0) \big] -\lambda
    \end{split}
\end{equation}

Remember that $i^*(\vec{c})$ is the number of adversarial potential winners, $h(c_0, r_0) = e^{-c_0 \cdot (1-\beta)(1-\alpha)}r_0$ is the reward the adversary gets by staying silent and $g(c_0, \AdvC, \AdvR) = \max_{0 \leq i \leq i^*(\vecc)}{e^{-c_i \cdot (1-\beta)(1-\alpha)}}(1 + r_i)$ is the maximum reward the adversary gets by broadcasting one of its credentials. Thus,
\begin{equation}
\notag
    \begin{split}
        s_\ell &= \mathbb{E}_{c_0 \leftarrow \expd{\beta(1-\alpha)}, r_0 \leftarrow \estk{t-1}} \left[ \max \{h(c_0, r_0), g(c_0, \AdvC, \AdvR)\} \mathbbm{1}\{i^*(\vec{c}) \neq 0\} \right] -\lambda \\
        &= \int_0^{\infty} \mathbb{E}_{r_0 \leftarrow \estk{t-1}} \left[ \max \{h(c_0, r_0), g(c_0, \AdvC, \AdvR)\} \mathbbm{1}\{i^*(\vec{c}) \neq 0\} \right] \beta \, (1-\alpha) \, e^{-c_0 \, \beta \, (1-\alpha)} \,dc_0 -\lambda \\
    \end{split}
\end{equation}

As a first step, we break down the integral wrt $c_0$ as a sum of integrals over smaller intervals where $i^*(\vec{c})$ remains constant, i.e.,
\begin{equation}
    \notag
    \begin{split}
        s_{\ell} = \int_{0}^{c_1} &\mathbb{E}_{r_0 \leftarrow \estk{t-1}} \left[ \max \{h(c_0, r_0), g(c_0, \AdvC, \AdvR)\} \mathbbm{1}\{i^*(\vec{c}) \neq 0\} \right] \beta \, (1-\alpha) \, e^{-c_0 \, \beta \, (1-\alpha)} \,dc_0 \\
        &+ \sum_{i^* = 1}^\coin \int_{c_{i^*}}^{c_{{i^*}+1}} \mathbb{E}_{r_0 \leftarrow \estk{t-1}} \left[ \max \{h(c_0, r_0), g(c_0, \AdvC, \AdvR)\} \mathbbm{1}\{i^*(\vec{c}) \neq 0\} \right] \beta \, (1-\alpha) \, e^{-c_0 \, \beta \, (1-\alpha)} \,dc_0 -\lambda
    \end{split}
\end{equation}
Observe that when $c_0 \in [c_{i^*}, c_{i^* +1})$, $i^*(\vec{c}) = i^*$. In this range,
\begin{equation}
    \notag
    \begin{split}
        g(c_0, \AdvC, \AdvR) &= \max_{1 \leq i \leq i^*(\vecc)}{e^{-c_i \cdot (1-\beta)(1-\alpha)}}(1 + r_i) \\
        &= \max_{1 \leq i \leq i^*}{e^{-c_i \cdot (1-\beta)(1-\alpha)}}(1 + r_i)
    \end{split}
\end{equation}
and is independent of the value of $c_0$. When $c_0 \in [0, c_1)$, $i^*(\vec{c}) = 0$ and so is \\$\mathbb{E}_{r_0 \leftarrow D} \left[ \max \{h(c_0, r_0), g(c_0, \AdvC, \AdvR)\} \mathbbm{1}\{i^*(\vec{c}) \neq 0\} \right]$. The entire integral wrt $c_0$ over $[0, c_1)$ equals zero. $s_{\ell}$ simplifies to
$$s_{\ell} = \sum_{i^* = 1}^\coin \int_{c_{i^*}}^{c_{{i^*}+1}} \mathbb{E}_{r_0 \leftarrow \estk{t-1}} \left[ \max \{h(c_0, r_0), g(c_0, \AdvC, \AdvR)\} \right] \beta \, (1-\alpha) \, e^{-c_0 \, \beta \, (1-\alpha)} \,dc_0 -\lambda$$
Let $f(i^*, \AdvC, \AdvR) = \int_{c_{i^*}}^{c_{{i^*}+1}} \mathbb{E}_{r_0 \leftarrow \estk{t-1}} \left[ \max \{h(c_0, r_0), g(c_0, \AdvC, \AdvR)\} \right] \beta \, (1-\alpha) \, e^{-c_0 \, \beta \, (1-\alpha)}\,dc_0$ and therefore,
$$s_{\ell} = \sum_{i^* = 1}^\coin f(i^*, \AdvC, \AdvR) - \lambda$$

We speed up the computation of $f(i^*, \AdvC, \AdvR)$ next. $f(i^*, \AdvC, \AdvR)$ is computed by integrating $c_0$ over $[c_{i^*}, c_{i^* +1})$, where, as discussed in the previous paragraph, the value of $g(c_0, \AdvC, \AdvR)$ is independent of $c_0$. To signify this independence, we will re-parameterize $g(c_0, \AdvC, \AdvR)$ into $g(i^*, \AdvC, \AdvR)$.
$$g(i^*, \AdvC, \AdvR) = \max_{1 \leq i \leq i^*}{e^{-c_i \cdot (1-\beta)(1-\alpha)}}(1 + r_i)$$
$$f(i^*, \AdvC, \AdvR) = \int_{c_{i^*}}^{c_{{i^*}+1}} \mathbb{E}_{r_0 \leftarrow \estk{t-1}} \left[ \max \{h(c_0, r_0), g(i^*, \AdvC, \AdvR)\} \right] \beta \, (1-\alpha) \, e^{-c_0 \, \beta \, (1-\alpha)}\,dc_0$$

We set up pre-computations for $\mathbb{E}_{r_0 \leftarrow \estk{t-1}} \left[ \max \{h(c_0, r_0), g(i^*, \AdvC, \AdvR)\} \right]$. Note that we are computing $\mathbb{E}_{r_0 \leftarrow \estk{t-1}} \left[ \max \{h(c_0, r_0), g(i^*, \AdvC, \AdvR)\} \right]$ conditioned on $c_0$ and $i^*$.
\begin{equation}
    \notag
    \begin{split}
        \E_{r_0 \sim \estk{t-1}}[\max \{h(c_0, r_0), g(i^*, \vec{c}_{-0}, \vec{r}_{-0}) \}] &= g(i^*, \vec{c}_{-0}, \vec{r}_{-0}) \times Pr(h(c_0, r_0) \leq g(i^*, \vec{c}_{-0}, \vec{r}_{-0})) \\
        & \hspace{1 cm} + \E_{r_0 \sim \estk{t-1}}[h(c_0, r_0) \times \mathbbm{1}(h(c_0, r_0) > g(i^*, \vec{c}_{-0}, \vec{r}_{-0}))] \\
        &= g(i^*, \vec{c}_{-0}, \vec{r}_{-0}) \times Pr(r_0 \leq \tfrac{g(i^*, \vec{c}_{-0}, \vec{r}_{-0})}{e^{-c_0 \cdot (1-\beta)(1-\alpha)}}) \\
        & \hspace{1 cm} + \E_{r_0 \sim \estk{t-1}}[e^{-c_0 \cdot (1-\beta)(1-\alpha)}r_0 \times \mathbbm{1}(r_0 > \tfrac{g(i^*, \vec{c}_{-0}, \vec{r}_{-0})}{e^{-c_0 \cdot (1-\beta)(1-\alpha)}})] \\
        &= g(i^*, \vec{c}_{-0}, \vec{r}_{-0}) \times Pr(r_0 \leq \tfrac{g(i^*, \vec{c}_{-0}, \vec{r}_{-0})}{e^{-c_0 \cdot (1-\beta)(1-\alpha)}}) \\
        & \hspace{1 cm} + e^{-c_0 \cdot (1-\beta)(1-\alpha)} \times \E_{r_0 \sim \estk{t-1}}[r_0 \times \mathbbm{1}(r_0 > \tfrac{g(i^*, \vec{c}_{-0}, \vec{r}_{-0})}{e^{-c_0 \cdot (1-\beta)(1-\alpha)}})] \\
        &= e^{-c_0 \cdot (1-\beta)(1-\alpha)} \times \Big[ \frac{g(i^*, \vec{c}_{-0}, \vec{r}_{-0})}{e^{-c_0 \cdot (1-\beta)(1-\alpha)}} \times Pr(r_0 \leq \tfrac{g(i^*, \vec{c}_{-0}, \vec{r}_{-0})}{e^{-c_0 \cdot (1-\beta)(1-\alpha)}}) \\
        & \hspace{1 cm} + \E_{r_0 \sim \estk{t-1}}[r_0 \times \mathbbm{1}(r_0 > \tfrac{g(i^*, \vec{c}_{-0}, \vec{r}_{-0})}{e^{-c_0 \cdot (1-\beta)(1-\alpha)}})] \Big]
    \end{split}
\end{equation}
The term inside the big square in the final equation is of the form $\{\theta Pr(r_0 \leq \theta) + \E_{r_0 \sim \estk{t-1}}[r_0 \times \mathbbm{1}(r_0 > \theta)]\}$, and can be pre-computed without knowing $\vec{c}_{-0}$ and $\vec{r}_{-0}$ (up to a discretization error $\epsilon$). We produce a recipe to set-up pre-computations for a general distribution $\D$.

\begin{mdframed}
$\Precomp{\hat{D}}$:
\begin{enumerate}
    \item Construct the pdf $d$ of $\hat{D}$ up to a discretization error $\epsilon$.
    \item Construct the cdf $\hat{D}$ recursively using the following recursion, once again, up to a discretization error $\epsilon$.
    $$\hat{\D}(\theta - \epsilon) = Pr(r_0 \leq \theta - \epsilon) = \hat{D}(\theta) - d(\theta)$$
    \item Construct the right tail of the expectation $E(\theta) = \E_{r_0 \sim D}[r_0 \times \mathbbm{1}(r_0 > \theta)]$ recursively by
    $$E(\theta - \epsilon) = E(\theta) + (\theta - \epsilon) \times d(\theta - \epsilon)$$
    \item Compute $E_{\max}(\theta) = \theta \hat{\D}(\theta) + E(\theta)$ for all $-t\lambda \leq \theta \leq t(1-\lambda)$ up to a discretization error $\epsilon$
\end{enumerate}

\end{mdframed}

The above compute takes time $O(\frac{t}{\epsilon})$ since the range of the reward distribution has a width at $t$ in round $t$. Plugging this back into the equation for $f(i^*, \AdvC, \AdvR)$,
\begin{equation}
    \notag
    \begin{split}
        f(i^*, \AdvC, \AdvR) &= \int_{c_{i^*}}^{c_{i^* + 1}} \E_{r_0 \sim \estk{t-1}}[\max \{h(c_0, r_0), g(i^*, \vec{c}_{-0}, \vec{r}_{-0}) \}] \beta \, (1-\alpha) \, e^{-c_0 \, \beta \, (1-\alpha)} \, dc_0 \\
        &= \int_{c_{i^*}}^{c_{i^* + 1}} \beta (1-\alpha) e^{-c_0 \beta (1-\alpha)} \times e^{-c_0 \, (1-\beta) \, (1-\alpha)} \times E_{\max}(\frac{g(i^*, \AdvC, \AdvR)}{e^{-c_0 (1-\beta) (1-\alpha)}}) \, dc_0
    \end{split}
\end{equation}
The final equality follows since $E_{\max}(\cdot)$ was defined such that $\E_{r_0 \sim D}[\max \{h(c_0, r_0), g(i^*, \vec{c}_{-0}, \vec{r}_{-0}) \}] = e^{-c_0 \, (1-\beta) \, (1-\alpha)} \times E_{\max}(\frac{g(i^*, \AdvC, \AdvR)}{e^{-c_0 (1-\beta) (1-\alpha)}})$. We have all the ingredients to compute $f(i^*, \AdvC, \AdvR)$ and speeding up the computation of the integral wrt $c_0$ is the final frontier. At a high level, we compute
$$G(c, g(i^*, \AdvC, \AdvR)) := \int_0^c \beta (1-\alpha) e^{-c_0 \beta (1-\alpha)} \times e^{-c_0 \, (1-\beta) \, (1-\alpha)} \times E_{\max}(\frac{g(i^*, \AdvC, \AdvR)}{e^{-c_0 (1-\beta) (1-\alpha)}}) \, dc_0$$
and thus,
$$f(i^*, \AdvC, \AdvR) = G(c_{i^* + 1}, g(i^*, \AdvC, \AdvR)) - G(c_{i^*}, g(i^*, \AdvC, \AdvR))$$

\noindent\textbf{Case-1:} $\beta = 1$

This is the simplest case. The input to $E_{\max}(\cdot)$ becomes independent of $c_0$.
\begin{equation}
    \notag
    \begin{split}
        G(c, g(i^*, \AdvC, \AdvR)) &= \int_0^c (1-\alpha) e^{-c_0 (1-\alpha)} \times E_{\max}(g(i^*, \AdvC, \AdvR)) \, dc_0 \\
        &= E_{\max}(g(i^*, \AdvC, \AdvR)) (1-e^{-c \, (1-\alpha)})
    \end{split}
\end{equation}
In this case, $f(i^*, \AdvC, \AdvR)$ can be explicitly computed.
$$f(i^*, \AdvC, \AdvR) = E_{\max}(g(i^*, \AdvC, \AdvR)) \big[e^{-c_{i^*} \, (1-\alpha)} - e^{-c_{i^* + 1} \, (1-\alpha)} \big]$$
All of our pre-compute gets done in $O(\frac{t}{\epsilon})$ time.

Summarizing,
\begin{equation}
    \notag
    \begin{split}
        s_{\ell} &= \sum_{i^* = 1}^\coin f(i^*, \AdvC, \AdvR) - \lambda \\
        &= \sum_{i^* = 1}^\coin E_{\max}(g(i^*, \AdvC, \AdvR)) \big[e^{-c_{i^*} \, (1-\alpha)} - e^{-c_{i^* + 1} \, (1-\alpha)} \big] - \lambda
    \end{split}
\end{equation}
Remember that $E_{\max}(\theta) = \theta D(\theta) + E(\theta)$ has already been pre-computed.

\noindent\textbf{Case-2:} $\beta = 0$

This is another straightforward case. When $\beta = 0$, the adversary cannot abstain from broadcasting. Thus, the adversarial reward equals $g(\coin, \AdvC, \AdvR)$.
$$s_{\ell} = g(\coin, \AdvC, \AdvR) - \lambda$$
In particular, we circumvent computing $f(i^*, \AdvC, \AdvR)$ altogether for $\beta = 0$.

\noindent\textbf{Case-3:} $\beta \notin\{ 1, 0\}$.

The most obvious way to compute $G(c, g(i^*, \AdvC, \AdvR))$ would be to discretize the integral wrt to $c_0$ into steps of $\eta$. However, the domain of $c_0$ is $[0, \infty)$. $E_{\max}(\frac{g(i^*, \AdvC, \AdvR)}{e^{-c_0 (1-\beta) (1-\alpha)}})$ becomes constant for a sufficiently large value $C_0$ of $c_0$, but, it is quite wasteful to compute the integral in discrete steps of $\eta$ from $0$ to a very large $C_0$. Instead, we change variables and integrate wrt $e^{-c_0 \, \beta \, (1-\alpha)}$ and discretize this integral by $\eta$.

The following expressions just substantiate the above intuition.
\begin{equation}
    \notag
    \begin{split}
        G(c, g(i^*, \AdvC, \AdvR)) &= \int_0^c \beta (1-\alpha) e^{-c_0 \beta (1-\alpha)} \times e^{-c_0 \, (1-\beta) \, (1-\alpha)} \times E_{\max}(\frac{g(i^*, \AdvC, \AdvR)}{e^{-c_0 (1-\beta) (1-\alpha)}}) \, dc_0 \\
        &= \int_{e^{-c \, \beta \, (1-\alpha)}}^1 e^{-c_0 \, (1-\beta) \, (1-\alpha)} E_{\max}(\frac{g(i^*, \AdvC, \AdvR)}{e^{-c_0 (1-\beta) (1-\alpha)}}) \, de^{-c_0 \, \beta \, (1-\alpha)} \\
        &= \int_{e^{-c \, \beta \, (1-\alpha)}}^1 \zeta^{\frac{1-\beta}{\beta}} E_{\max}\left(\frac{g(i^*, \AdvC, \AdvR)}{\zeta^{\frac{1-\beta}{\beta}}}\right) \, d \zeta
    \end{split}
\end{equation}
For every possible value $g(i^*, \AdvC, \AdvR) \in [-t \, \lambda, t \, (1-\lambda)]$ and a discretization parameter $\eta$, the above integral can be discretized as follows.
\begin{equation}
    \notag
    \begin{split}
        G(c, g(i^*, \AdvC, \AdvR)) = \eta \times \sum_{\zeta = 0, \text{ in steps of } \eta}^c \zeta^{\frac{1-\beta}{\beta}} E_{\max}\left(\frac{g(i^*, \AdvC, \AdvR)}{\zeta^{\frac{1-\beta}{\beta}}}\right)
    \end{split}
\end{equation}
Remember to take the upper and lower Riemann sums while computing $\UBD{t}$ and $\LBD{t}$ respectively.

Once again, we compute the sample $s_{\ell}$ by
\begin{equation}
    \notag
    \begin{split}
        s_{\ell} &= \sum_{i^* = 1}^\coin G(c_{i^* + 1}, g(i^*, \AdvC, \AdvR)) - G(c_{i^*}, g(i^*, \AdvC, \AdvR)) - \lambda \\
        &= \sum_{i^* = 1}^\coin f(i^*, \AdvC, \AdvR) - \lambda
    \end{split}
\end{equation}

The integrals with respect to $c_0$ take $O(\frac{1}{\eta})$ to pre-compute and we do this once for each of the $\big(\frac{t}{\epsilon}\big)$ possible values of $g(i^*, \AdvC, \AdvR)$, amounting to a total run-time $O(\tfrac{t}{\epsilon \eta})$. \Cref{thm:PrecompRunTime} follows.

By \Cref{thm:PrecompRunTime}, the pre-computations in one execution of $\finsampleaddl{\estk{t}}$ takes $O(\tfrac{t}{\epsilon \, \eta})$ time. Across the $\round$ rounds, a total $O(\tfrac{\round^2}{\epsilon \, \eta})$ time is spent in pre-computations. Computing one sample $s_{\ell}$ involves summing up $\coin$ terms,
$$s_{\ell} = \sum_{i^* = 1}^\coin f(i^*, \AdvC, \AdvR) - \lambda$$
and takes $O(\coin)$ time to compute. A total of $O(\coin \samples)$ time is spent on sampling in each execution of \\$\finsampleaddl{\cdot}$ and a total of $O(\round \coin \samples)$ across the $\round$ rounds. The time needed to inflate and deflate are dominated by $O(\round \coin \samples)$, establishing \Cref{thm:RunTime}.

\section{Locating $\lambda$ and Describing the Optimal Adversarial Strategy} \label{sec:Wrap}
\subsection{Proof of Theorem \ref{thm:BinSearchAPX}} \label{sec:ProofOfBinSearch}
We first begin with a toy version that says if the rewards from $\FinLinLamCSSPA{\lambda_1}$ and $\FinLinLamCSSPA{\lambda_2}$ are close, then $\lambda_1 - \lambda_2$ is small. 

\begin{lemma} \label{thm:BinSearch}
    Suppose that $\lambda_1$ and $\lambda_2$ are such that
    $$|\LinLamRew{\lambda_1}{\optTk(\lambda_1)} - \LinLamRew{\lambda_2}{\optTk(\lambda_2)}| \leq \zeta\text{.}$$
    Then, $|\lambda_1 - \lambda_2| \leq \zeta$.
\end{lemma}
\begin{proof}
    Without loss of generality, assume that $\lambda_1 \leq \lambda_2$. Assume for contradiction that $\lambda_2 - \lambda_1 > \zeta$.

    As a thought experiment, let the adversary ape $\optTk(\lambda_2)$ in $\FinLinLamCSSPA{\lambda_1}$. Since the rewards are linear, the adversary earns $\LinLamRew{\lambda_2}{\optTk(\lambda_2)}$ plus the savings in entry fee from paying just $\lambda_1 < \lambda_2$. Since the adversary participates in at least one round, the adversary saves more than $(\lambda_2 - \lambda_1)$ in entry fee, and thus, the total reward is at least $\LinLamRew{\lambda_2}{\optTk(\lambda_2)} + (\lambda_2 - \lambda_1)$. However, $\LinLamRew{\lambda_1}{\optTk(\lambda_1)} - \LinLamRew{\lambda_2}{\optTk(\lambda_2)} \leq \zeta$ and $(\lambda_2 - \lambda_1) \geq \zeta$. The adversary makes strictly more than $\LinLamRew{\lambda_1}{\optTk(\lambda_1)}$ in $\FinLinLamCSSPA{\lambda_1}$. This is a contradiction, since $\optTk(\lambda_1)$ is the optimal strategy.
\end{proof}

In \Cref{thm:BinSearchAPX}, we show that even if we estimate $\FinLinLamCSSPA{\lambda_1}$ up to an error of $\delta$ then, the true reward in $\FinLinLamCSSPA{\lambda_2}$ and our estimated rewards can still not be close unless $\lambda_1 - \lambda_2$ is small.

\begin{proof}[Proof of \Cref{thm:BinSearchAPX}]
    From \Cref{thm:SampleSummary}, with probability at least $1 - \Big(2 \, \round \, \chernoff + \round \, \tfrac{e^{-\strat \samples}}{\strat} \, \sqrt{\tfrac{\ln \chernoff^{-1}}{2\samples}} \Big)$, $$\LinLamRew{\lambda_1}{\optTk(\lambda_1)} \in \Big[\E_{s \sim \UBD{\round, \coin}(\lambda_1)}[s], \E_{s \sim \LBD{\round, \coin}(\lambda_1)}[s] \Big]$$
    Thus, $|\LinLamRew{\lambda_1}{\optTk(\lambda_1)} - r| \leq \delta$ and $|r-\LinLamRew{\lambda_2}{\optTk(\lambda_2)}| \leq \zeta$. The claim follows from the triangle inequality and \Cref{thm:BinSearch}.
\end{proof}

\subsection{The Optimal Adversarial Strategy} \label{sec:AdvStrategy}
Once we (approximately) locate the optimal reward $\lambda^*$ of the $\coin$-scored adversary in $\CSSPA$ that terminates in $\round$ rounds, the description of a near optimal adversarial strategy becomes succinct. The adversary pretends that it is $\coin$-scored and is participating in $\FinLinCSSPA{\round}$ and plays the actions recommended by $\optTk(\lambda^*)$. In particular, the adversary resets $\CSSPA$ with frequency at least $\round$ and behaves like a $\coin$-scored adversary. The adversary gets a zero reward in $\FinLinCSSPA{\round}$, which corresponds to a reward $\lambda^*$ in $\CSSPA$.

The reward from the above strategy approaches the optimal reward as the error in locating $\lambda^*$ reduces. The loss in the reward from being $\coin$-scored and from resetting $\CSSPA$ every $\round$ rounds would approach zero as $\round \xrightarrow{} \infty$ and $\coin \xrightarrow{} \infty$.

\section{Summary} \label{sec:TablesandSummary}

\subsection{A Summary of the Simulation} \label{sec:Summary}
\begin{mdframed}
$\finsampleaddl{\epsilon, \eta, \estdist{t-1, \coin}}$:
\begin{enumerate}
    \item $\Precomp{\estdist{t-1, \coin}}$:
    \begin{enumerate}
        \item Construct the pdf $d$ of $\estdist{t-1, \coin}$ up to a discretization error $\epsilon$.
        \item Construct the cdf $\estdist{t-1, \coin}$ recursively using the following recursion up to a discretization error $\epsilon$.
        $$\estdist{t-1, \coin}(\theta - \epsilon) = Pr(r_0 \leq \theta - \epsilon) = \estdist{t-1, \coin}(\theta) - d(\theta)$$
        \item Construct the right tail of the expectation $E(\theta) = \E_{r_0 \sim D}[r_0 \times \mathbbm{1}(r_0 > \theta)]$ recursively by
        $$E(\theta - \epsilon) = E(\theta) + (\theta - \epsilon) \times d(\theta - \epsilon)$$
        \item Compute $E_{\max}(\theta) = \theta \estdist{t-1, \coin} + E(\theta)$ for all $-t\lambda \leq \theta \leq t(1-\lambda)$ up to a discretization error $\epsilon$
        \item If $\beta \neq 1, 0$: For $-t \lambda \leq g(i^*, \vec{c}_{-0}, \vec{r}_{-0}) \leq t(1-\lambda)$ in steps of $\epsilon$:
        \begin{enumerate}
            \item For $0 \leq c \leq 1$ in steps of $\eta$
            \begin{itemize}
                \item $G(c, g(i^*, \vec{c}_{-0}, \vec{r}_{-0})) = \eta \times \sum_{\zeta = 0, \text{ in steps of } \eta}^c \zeta^{\frac{1-\beta}{\beta}} E_{\max}(\frac{g(i^*, \AdvC, \AdvR)}{\zeta^{\frac{1-\beta}{\beta}}})$
            \end{itemize}
        \end{enumerate}
    \end{enumerate}

    \item For $1 \leq \ell \leq \samples$.
    \begin{enumerate}
        \item $\drawadv{\estdist{t-1, k}}$:
        \begin{itemize}
            \item Sample $\AdvR$: Draw $\coin$ rewards $r_1, r_2, \dots, r_{\coin}$ i.i.d~from $\dist{t-1}$.
            \item Sample $\AdvC$: Draw $\coin$ scores $c_1, c_2, \dots, c_{\coin}$ of adversarial wallet as follows. Draw $c_1 \xleftarrow{} \expd{\alpha}$ and $c_{i+1} \xleftarrow{} c_i + \expd{\alpha}$ (a fresh sample for each $i$) for $1 \leq i \leq \coin -1$. For convenience, set $c_{\coin + 1} = \infty$.
            \item Return $(\AdvR, \AdvC)$.
        \end{itemize}

        \item $\sampleFromPrecomp{\estdist{t-1, k}}$:
        \begin{enumerate}
            \item For $1 \leq i^* \leq \coin$:
            \begin{itemize}
                \item Compute $g(i^*, \AdvC, \AdvR) = \max_{1 \leq i \leq i^*}{e^{-c_i \cdot (1-\beta)(1-\alpha)}}(1 + r_i)$
            \end{itemize}
            \item If $\beta = 0$:
            \begin{itemize}
                \item return $s_{\ell} = g(\coin, \AdvC, \AdvR) - \lambda$
            \end{itemize}
            \item Else if $\beta = 1$:
            \begin{itemize}
                \item For $1 \leq i^* \leq k$: $f(i^*, \AdvC, \AdvR) = E_{\max}(g(i^*, \AdvC, \AdvR)) \big[e^{-c_{i^*} \, (1-\alpha)} - e^{-c_{i^* + 1} \, (1-\alpha)} \big]$
                \item return $s_{\ell} = \sum_{i^* = 1}^k f(i^*, \AdvC, \AdvR) - \lambda$
            \end{itemize}
            \item $\beta \neq 0, 1$:
            \begin{itemize}
                \item For $1 \leq i^* \leq k$: $f(i^*, \AdvC, \AdvR) = G(c_{i^* + 1}, g(i^*, \AdvC, \AdvR)) - G(c_{i^*}, g(i^*, \AdvC, \AdvR))$
                \item return $s_{\ell} = \sum_{i^* = 1}^k f(i^*, \AdvC, \AdvR) - \lambda$
            \end{itemize}
        \end{enumerate}

    \end{enumerate}
    \item $\tilde{\D}^{\optimal}_{t, \coin}$ be the uniform distribution over $\big(s_{\ell}\big)_{1 \leq \ell \leq \samples}$ (in descending order)
    \item Inflate while computing the upper bound and deflate while computing the lower bound.
    \begin{itemize}
        \item $\Defl{\chernoff, \tilde{\D}^{\optimal}_{t, \coin}}:$
        \begin{enumerate}
            \item Delete the largest $\samples \cdot \sqrt{\tfrac{\ln \chernoff^{-1}}{2\samples}}$ samples from $\tilde{\D}^{\optimal}_{t, \coin}$
            \item Append $\samples \cdot \sqrt{\tfrac{\ln \chernoff^{-1}}{2\samples}}$ copies of $-\lambda$ to $\tilde{\D}^{\optimal}_{t, \coin}$
        \end{enumerate}
        \item $\Infl{\chernoff, \strat, \tilde{\D}^{\optimal}_{t, \coin}}$:
        \begin{enumerate}
            \item Delete the smallest $\samples \cdot \sqrt{\tfrac{\ln \chernoff^{-1}}{2\samples}}$ samples from $\tilde{\D}^{\optimal}_{t, \coin}$
            \item Append $\strat \samples$ copies of $t\,(1-\lambda)$ to $\tilde{\D}^{\optimal}_{t, \coin}$
            \item For $1 \leq \ell < \frac{\samples}{\strat \, \samples} \cdot \sqrt{\tfrac{\ln \chernoff^{-1}}{2 \samples}}$:
            \begin{itemize}
                \item Append $\strat \samples$ copies of $s_{\ell}$
            \end{itemize}
        \end{enumerate}
    \end{itemize}
    
    \item Return $\estdist{t, \coin}$ to be the uniform distribution over $\{s_1, s_2, \dots, s_{\samples_t}\}$.
\end{enumerate}

$\tcsimulate$:
\begin{enumerate}
    \item Initialize $\estdist{0, \coin}$ to be the point-mass distribution at $0$.
    \item For $1 \leq t \leq \round$:
    \begin{enumerate}
        \item $\estdist{t, \coin} = \finsampleaddl{,\epsilon, \eta, \estdist{t-1, \coin}}$.
    \end{enumerate}
    \item Return $\E_{s \sim \estdist{\round, \coin}}[s]$.
\end{enumerate}
\end{mdframed}

\subsection{A Summary of Notations and Functions in the Simulation} \label{sec:Tables}

\begin{table}[ht]
    \centering
    \resizebox{\textwidth}{!}{
    \begin{tabular}{|c| l|}
        \hline
        Notation & Description  \\
        \hline
        $\alpha$ & The fraction of stake held by the adversary \\
        \hline
        $\beta$ & \makecell[l]{The fraction of honest stake held by $\knownHonest$} \\
        \hline
        $\lambda$ & \makecell[l]{The per-round entry fee the adversary has to pay to participate in $\LinearCSSPA$} \\
        \hline
        $\round$ & \makecell[l]{We simulate $\FinLinCSSPA{\round}$}\\
        \hline        
        $\coin$ & \makecell[l]{We simulate a $\coin$-scored adversary}\\
        \hline        
        $\optrew{\round, \coin}$ & \makecell[l]{The optimal adversarial strategy in $\FinLinCSSPA{\round}$ for a $\coin$-scored adversary}\\
        \hline
        $\estdist{t, \coin}$ & \makecell[l]{The estimated distribution of optimal rewards in $\FinLinCSSPA{t}$ for a \\$k$-scored adversary}\\
        \hline
        $\chernoff$ & \makecell[l]{Probability of estimation error from the DKW inequality in inflate/deflate}\\
        \hline    
        $\strat$ & \makecell[l]{A small quantile gets duplicated $\strat \samples$ times while inflating}\\
        \hline        
        $\epsilon$ & \makecell[l]{Discretization parameter of the reward distributions}\\
        \hline        
        $\eta$ & \makecell[l]{Discretization parameter to compute Riemann sums}\\
        \hline        

    \end{tabular}}
    \caption{A summary of notations}
    \label{tab:Notations}
\end{table}

\begin{table}[ht]
    \centering
    \resizebox{\textwidth}{!}{
    \begin{tabular}{|c | l|}
        \hline
        Function & Description  \\
        \hline
        $\addl{\strategy, D}$ & \makecell[l]{Given a distribution $D$ of rewards achieved by playing $\strategy$ in the \\
        last $t-1$ rounds, $\addl{\strategy, D}$ computes the distribution \\of rewards won by playing $\strategy$ in the last $t$ rounds \\of $\FinLinCSSPA{t}$} \\
        \hline
        $\drawadv{D}$ & \makecell[l]{Given the reward distribution $D$, $\drawadv{D}$ samples rewards and \\ scores for adversarial wallets} \\
        \hline
        $\sample{\strategy, D}$ & \makecell[l]{Conditional on the rewards and scores of the adversary's wallets, \\ $\sample{\strategy, D}$ computes the reward of the \\ adversary, in expectation over $\knownHonest$ and $\opaqueHonest$'s scores and the reward \\ from letting $\knownHonest$ win.} \\
        \hline
        $\Infl{\chernoff, \strat, \D}$ & \makecell[l]{For an input distribution $\D$, $\Infl{\chernoff, \D}$ deletes the smallest samples \\and replaces them with samples corresponding to a small quantile}\\
        \hline
        $\Defl{\chernoff, \D}$ & \makecell[l]{For an input distribution $\D$, $\Defl{\chernoff, \D}$ deletes the largest samples \\and replaces them with the infimum of the distribution}\\
        \hline
        $\Precomp{\D}$ & \makecell[l]{For an input distribution $\D$, $\Precomp{D}$ sets up the \\pre-computations required to speed up the sampling procedure}\\
        \hline
        $\simulate{\optTk}$ &\makecell[l]{Executes $\round$ iterations of $\addl{\optTk}$ and returns $\E_{s \sim \estdist{\round, \coin}}[s]$}\\
        \hline
    \end{tabular}}
    \caption{A summary of functions used in the simulation. Functions with similar names across different variants of the simulation have similar functions. For example, $\addl{\optTk, D}$ and $\finsampleaddl{D}$ have the sample functionality (adding an extra layer in the simulation), but in different variants.}
    \label{tab:Func}
\end{table}



\end{document}